\documentclass{LMCS}

\def\doi{8 (2:10) 2012}
\lmcsheading%
{\doi}
{1--51}
{}
{}
{Nov.~11, 2011}
{Jun.~19, 2012}
{}

\pdfoutput=1
\usepackage{enumerate}
\usepackage{hyperref}
\usepackage{microtype}
\usepackage{wrapfig}
\usepackage{cancel}

\usepackage{tikz}
\usetikzlibrary{arrows,automata,backgrounds,calc,chains,fadings,folding,decorations.fractals,decorations.pathreplacing,fit,patterns,positioning,mindmap,shadows,shapes.geometric,shapes.symbols,through,trees,plotmarks,external}
\usepackage{amssymb}
\usepackage{verbatim}

\usepackage[UKenglish]{babel}
\usepackage{stmaryrd}
\usepackage{latexsym}
\usepackage{amssymb}
\usepackage{amsmath}
\usepackage{amsthm}
\usepackage{array}
\usepackage{url}

\theoremstyle{plain}
  \newtheorem{theorem}{Theorem}[section]
  \newtheorem{lemma}[theorem]{Lemma}
  \newtheorem{claim}[theorem]{Claim}

\theoremstyle{definition}
  \newtheorem{definition}[theorem]{Definition}
  \newtheorem{example}[theorem]{Example}

\newcommand{\nat}{\mathsf{nat}}
\newcommand{\rijtje}[1]{[#1]}

\newcommand{\setsorts}{\mathcal{B}}
\newcommand{\settypes}{\mathbb{T}}

\newcommand{\setfun}{\mathcal{F}}
\newcommand{\setvar}{\mathcal{V}}

\newcommand{\F}{\setfun}
\newcommand{\Fex}{\Sigma^{\mathit{ex}}}
\newcommand{\Rules}{\mathcal{R}}

\newcommand{\FV}{\mathit{FV}}
\newcommand{\Constructors}{\mathcal{C}}
\newcommand{\Defineds}{\mathcal{D}}
\newcommand{\Constants}{\mathbb{C}}

\newcommand{\atype}{\sigma}
\newcommand{\btype}{\tau}
\newcommand{\ctype}{\rho}
\newcommand{\dtype}{\alpha}
\newcommand{\abasetype}{\iota}
\newcommand{\bbasetype}{\kappa}
\newcommand{\aterm}{s}
\newcommand{\bterm}{t}
\newcommand{\cterm}{u}
\newcommand{\dterm}{v}
\newcommand{\eterm}{w}
\newcommand{\fterm}{q}
\newcommand{\avar}{x}
\newcommand{\bvar}{y}
\newcommand{\cvar}{z}
\newcommand{\dvar}{u}
\newcommand{\afun}{\ensuremath{{f}}}  
\newcommand{\bfun}{\ensuremath{{g}}}
\newcommand{\cfun}{\ensuremath{{h}}}
\newcommand{\asub}{\gamma}
\newcommand{\bsub}{\delta}
\newcommand{\csub}{\chi}

\newcommand{\clausevar}{\texttt{(var)}}
\newcommand{\clauseapp}{\texttt{(app)}}
\newcommand{\clauseabs}{\texttt{(abs)}}
\newcommand{\clausefun}{\texttt{(fun)}}
\newcommand{\clauserule}{\texttt{(rule)}}

\newcommand{\clausebeta}{\texttt{(beta)}}

\newcommand{\cbeta}{\mathtt{beta}}

\newcommand{\abs}[2]{\lambda #1 . \, #2}
\newcommand{\app}[2]{#1 \cdot #2}
\newcommand{\subst}[2]{#1 #2}

\newcommand{\typepijl}{\!\Rightarrow\!}
\newcommand{\decpijl}{\typepijl}
\newcommand{\ftypepijl}{\Rightarrow}
\newcommand{\arrz}{\ensuremath{\rightarrow}}
\newcommand{\arr}[1]{\ensuremath{\rightarrow_{#1}}}
\newcommand{\arrr}[1]{\arr{#1}^*}
\newcommand{\arrp}[1]{\arr{#1}^+}
\newcommand{\supterm}{\rhd}
\newcommand{\suptermeq}{\unrhd}

\newcommand{\csuptermeq}{\,\overline{\suptermeq}\,}
\newcommand{\dppijl}{\leadsto}
\renewcommand{\c}{\mathsf{c}}

\newcommand{\domain}{\mathsf{dom}}
\newcommand{\head}{\mathit{head}}

\newcommand{\ttag}{\mathsf{tag}}

\newcommand{\DP}{\mathsf{DP}}
\newcommand{\acycle}{\mathcal{C}}
\renewcommand{\P}{\mathcal{P}}
\newcommand{\up}[1]{#1^\sharp}
\newcommand{\gsymbuse}{\sqsupseteq_{\mathit{us}}}
\newcommand{\gsymbform}{\sqsubseteq_{\mathit{fo}}}
\newcommand{\basealgebra}{\mathcal{A}}
\newcommand{\basealgebraset}{A}

\newcommand{\gwm}{\sqsupset_{wm}}
\newcommand{\geqwm}{\sqsupseteq_{wm}}
\newcommand{\WM}{\mathsf{WM}}
\newcommand{\lamalgint}[1]{[#1]}
\newcommand{\lamalgintc}[1]{[#1]_{\constvaluation,\varvaluation}}
\newcommand{\algint}[1]{\llbracket #1 \rrbracket}
\newcommand{\algintc}[1]{\llbracket #1 \rrbracket_{\constvaluation,\varvaluation}}
\newcommand{\candidatesof}[1]{{Cand}(#1)}
\newcommand{\constvaluation}{\mathcal{J}}
\newcommand{\varvaluation}{\alpha}
\newcommand{\fatlambda}{\lambda\!\!\!\lambda}
\newcommand{\argfil}{\pi}
\newcommand{\filter}{\overline{\pi}}
\newcommand{\userules}{\mathit{U\!R}}
\newcommand{\Symbols}{\mathit{Symb}}
\newcommand{\formsymb}{\mathit{F\!S}}
\newcommand{\formrules}{\mathit{F\!R}}

\newcommand{\wanda}{\texttt{WANDA}}
\newcommand{\llfe}{local}
\newcommand{\geqterm}{\succeq}
\newcommand{\gterm}{\succ}
\newcommand{\geqorgterm}{\succ^?}
\newcommand{\gtermcpo}{\succ_{\mathsf{CPO}}}
\newcommand{\geqtermcpo}{\succeq_{\mathsf{CPO}}}
\newcommand{\Ce}{\mathcal{C}_\epsilon}
\newcommand{\p}{\symb{p}}

\newcommand{\setop}{\{}
\newcommand{\setcl}{\}}

\newcommand{\eval}{\mathsf{eval}}
\newcommand{\fun}{\mathsf{fun}}
\newcommand{\dom}{\mathsf{dom}}
\newcommand{\suc}{\mathsf{s}}
\newcommand{\nul}{\mathsf{o}}
\newcommand{\I}{\mathsf{I}}
\newcommand{\twice}{\mathsf{twice}}
\newcommand{\N}{\mathbb{N}}
\newcommand{\M}{\mathsf{M}}

\newcommand{\map}{\mathsf{map}}
\newcommand{\nil}{\mathsf{nil}}
\newcommand{\cons}{\mathsf{cons}}

\renewcommand{\o}{\mathsf{o}}
\newcommand{\scctwice}{\acycle_\twice}
\newcommand{\append}{\mathtt{append}}
\newcommand{\mins}{\mathtt{minus}}
\newcommand{\quot}{\mathtt{quot}}
\newcommand{\ifex}{\mathtt{if}}
\newcommand{\trueex}{\mathtt{true}}
\newcommand{\falseex}{\mathtt{false}}
\newcommand{\headex}{\mathtt{head}}
\newcommand{\tailex}{\mathtt{tail}}
\newcommand{\stringex}{\mathtt{string}}
\newcommand{\booleanex}{\mathtt{bool}}
\newcommand{\listex}{\mathtt{funlist}}
\newcommand{\lijst}{\mathtt{list}}
\newcommand{\symb}[1]{\mathsf{#1}}

\newcommand{\emptyline}{\medskip\noindent}
\newcommand{\paragraaf}[1]{\bigskip\noindent\textit{\textbf{#1}}}
\newcommand{\sparagraaf}[1]{\noindent \textit{\textbf{#1}}}
\newcommand{\summary}[1]{\emph{#1}}

\begin{document}

\title[Dynamic Higher-Order Dependency Pairs]{Dynamic Dependency Pairs \\ for Algebraic Functional Systems\rsuper*}

\author[C. Kop]{Cynthia Kop}
\address{Faculty of Sciences, 
         VU University,
         De Boelelaan 1081a, 
         1081 HV Amsterdam, 
         The Netherlands}
\email{kop@few.vu.nl, femke@few.vu.nl}
\thanks{This research is supported by the Netherlands Organisation
  for Scientific Research (NWO-EW) under grant 612.000.629
  (Higher-Order Termination).}

\author[F.~van Raamsdonk]{Femke van Raamsdonk}
\titlecomment{{\lsuper*}Extended version of \cite{kop:raa:11:1}.}

\begin{abstract}
We extend the higher-order termination method of dynamic dependency
pairs to Algebraic Functional Systems (AFSs).  In this
setting, simply typed lambda-terms with algebraic reduction and
separate $\beta$-steps are considered.
For left-linear AFSs, the method is shown to be complete.
For so-called \llfe\ AFSs
we define a variation of usable rules and an extension of
argument filterings.
All these techniques have been implemented in the higher-order
termination tool WANDA.
\end{abstract}

\maketitle

\section{Introduction} \label{sec:introduction}

An important method to (automatically) prove termination of
first-order term rewrite systems is the dependency pair approach by
Arts and Giesl~\cite{art:gie:00:1}.
This approach transforms a term rewrite system into groups of
ordering constraints, such that rewriting is terminating if and only
if the groups of constraints are (separately) solvable.
These constraints can be simplified using for instance
argument filterings and usable
rules~\cite{art:gie:00:1,gie:thi:sch:fal:06,hir:mid:07:1}.
Various optimisations of the method have been studied, 
see for example~\cite{hir:mid:05:1,gie:thi:sch:05:2}.

This paper contributes to the study of dependency pairs for
higher-order rewriting.
It is not easy to adapt the approach to a higher-order setting,
primarily due to the presence of $\beta$-reduction.  A first, very
natural extension to Nipkow's HRSs (higher-order rewrite systems)
is given
in~\cite{sak:wat:sak:01}, but it relies on the \emph{subterm
property}. 
Due to this property it is impossible to define optimisations like 
argument filterings.
Moreover, unlike the first-order case, the method is not complete:
a terminating system may well have an infinite dependency chain.

Since then, the focus of higher-order dependency pairs has been on
the so-called \emph{static} style. 
This style imposes limitations on the rewrite rules which allow the
subterm property to be dropped.  For static dependency pairs, too,
there are no completeness results available.

Here we return to the original, \emph{dynamic} style of
dependency pairs, and show how the subterm property can be weakened.
We introduce variations of usable rules and argument filterings, and
use weakly monotonic algebras and recursive path orderings to orient
the resulting constraints.
While special interest is reserved for
the class of so-called \emph{\llfe} systems, the core technique is
defined without restrictions, and is complete for left-linear
systems.

Unlike previous approaches, we do not consider rewriting modulo
$\beta/\eta$ (Nipkow's HRSs \cite{nip:91:1}), but
with $\beta$-reduction as a separate step (Jouannaud and
Okada's AFSs~\cite{jou:oka:91:1}).  Although higher-order
path orderings are commonly studied in the setting of
AFSs~\cite{jou:rub:99:1,bla:jou:rub:08:1}, and it is the only
style of higher-order rewriting which currently appears in the
annual termination competition~\cite{termcomp}, there is so far
little work on dependency pairs for this formalism.

This paper is an extended version of~\cite{kop:raa:11:1}, with
complete proofs and some new features.

\paragraaf{Paper Setup}
We briefly discuss the ideas from studies of dependency pairs
for HRSs and also for applicative systems in
Section~\ref{sec:relatedwork}.
In Section~\ref{sec:preliminaries} we recapitulate the AFS formalism,
and in Section~\ref{sec:firstorder} we give a brief overview of the
dependency pair framework for first-order term rewriting.
Basic (unrestricted) definitions of higher-order dependency pairs,
dependency chains, the dependency graph 
and reduction orderings are discussed in Section~\ref{sec:basic}.

In order to obtain stronger results, we then introduce a restriction
``\llfe'', which many common AFSs satisfy, and define \emph{formative
rules} (a variation of usable rules) for \llfe\ AFSs.  We show how
the results from Section~\ref{sec:basic} can be strenghtened for
\llfe\ systems.

To find reduction orderings for dependency pair constraints,
Section~\ref{sec:redpair} discusses two approaches:
first, we will see how dependency pairs interact with weakly
monotonic algebras, and then we define
\emph{argument functions}, a generalisation of argument filterings.

In Section~\ref{sec:noncollapse} we discuss some improvements when
dealing with non-collapsing dependency pairs; for example, in this
setting we can use the subterm criterion and usable rules.
Section~\ref{sec:algorithm} summarises all results, both for \llfe\ 
and non-\llfe\ systems, in a ready-to-use algorithm.

In Section~\ref{sec:improvements} we discuss how the theory in this
paper can be used for polymorphic and otherwise infinite systems, and
how the static and dynamic dependency pair approaches can be combined.
Experimental results with our tool \wanda\ are presented
in Section~\ref{sec:experiments}.

\section{Background and Related Work}\label{sec:relatedwork}

\summary{In this section, we discuss the existing work on
higher-order dependency pairs.}

The existing work on higher-order dependency pairs
can roughly be split along two axes.  On the one axis,
the higher-order formalism; we distinguish between applicative
rewriting, rewriting modulo $\beta$ (HRSs), and with $\beta$ as a
separate step (AFSs). On the other the style of dependency pairs,
with common styles being \emph{dynamic} and \emph{static}.
Figure \ref{fig:overview} gives an overview.

\begin{figure}[h]
\vspace{-3pt}
\begin{tabular}{r|lll}
 & \multicolumn{1}{c}{\textbf{Applicative}} & \textbf{HRS} & \textbf{AFS} \\
\hline
\textbf{Dynamic}
  & \cite{kus:01:1} \cite{aot:yam:05:1}
  & \cite{sak:wat:sak:01} \cite{kop:raa:10:1}
  & \cite{kop:raa:11:1} \emph{this paper} \\
\textbf{Static}
  & \cite{kus:sak:07:1} \cite{kus:sak:09:1}
  & \cite{bla:06:1} \cite{sak:kus:05:1} \cite{kus:iso:sak:bla:09:1}
    \cite{suz:kus:bla:11}
  & \cite{bla:06:1} \\
\textbf{Other}
  & \cite{hir:mid:zan:08}
  & --
  & -- \\
\end{tabular}
\caption{\emph{Papers on Higher Order Dependency Pairs}}
\label{fig:overview}
\vspace{-5pt}
\end{figure}

\noindent
The dynamic and static approach differ in the treatment of leading
variables in the right-hand sides of rules (subterms
$\app{\avar}{\aterm_1} \cdots \aterm_n$ with $n > 0$ and $\avar$ a
free variable).  In the dynamic approach, such subterms lead to a
dependency pair; in the static approach they do not.
First-order techniques like argument filterings, the
subterm criterion and usable rules are easier to extend to a static
approach, while equivalence results tend to be limited to the dynamic
style.  Static dependency pairs rely on certain restrictions on the
rules.

\sparagraaf{Dependency pairs for applicative term rewriting}
In applicative systems, terms are built from variables, constants and
a binary application operator.  Functional variables may be present,
as in $x \cdot a$, but there is no abstraction, as in $\abs{x}{x}$.
There are various styles of applicative rewriting, both untyped,
simply typed, and with alternative forms of typing.

A dynamic approach was defined both for untyped and simply-typed
applicative systems in~\cite{kus:01:1}, along with a definition of
argument filterings.
A first static approach appears in~\cite{kus:sak:07:1} and is
improved in~\cite{kus:sak:09:1}; the method is restricted to
`plain function passing' systems where, intuitively, leading
variables are harmless.
Due to the lack of binders, it is also possible to eliminate leading
variables by instantiating them, as is done for simply-typed systems
in~\cite{aot:yam:05:1}; in~\cite{hir:mid:zan:08}, an
uncurrying transformation from untyped applicative systems to normal
first-order systems is used.
These techniques have no parallel in rewriting with binders.

Unfortunately, strong though the results for applicative systems may
be, they are not directly useful in the setting of AFSs, since
termination may be lost by adding $\lambda$-abstraction and
$\beta$-reduction.
For example, the simply-typed applicative system
${\sf{app}} \cdot ({\sf{abs}} \cdot F) \cdot x \rightarrow F \cdot x$,
with $F : \abasetype \typepijl \abasetype$ a functional variable,
$x: \abasetype$ a variable, and ${\sf{app}} : \abasetype \typepijl
\abasetype \typepijl \abasetype$,\ ${\sf{abs}} : (\abasetype \typepijl
\abasetype) \typepijl \abasetype$ constants,
is terminating because in every step the size of a term decreases.
However, adding $\lambda$-abstraction and $\beta$-reduction 
destroys this property:
with $\omega = {\sf{abs}} \cdot ( \abs{x}{{\sf{app}} \cdot x \cdot x})$
we have
${\sf{app}} \cdot \omega \cdot \omega = 
{\sf{app}} \cdot ( {\sf{abs}} \cdot ( \abs{x}{{\sf{app}} \cdot x \cdot x})) \cdot \omega
\rightarrow
( \abs{x}{{\sf{app}} \cdot x \cdot x} ) \cdot \omega \rightarrow 
{\sf{app}} \cdot \omega \cdot \omega$.

Let us move on to rewriting with binders; most results here are on
Nipkow's HRSs.

\paragraaf{Dynamic Dependency Pairs for HRSs}
A first
definition of dependency pairs for HRSs 
is given in~\cite{sak:wat:sak:01}.
Here termination is not equivalent to the absence of infinite 
dependency chains, and a term is required to be greater than its
subterms (the \emph{subterm property}), which makes many
optimisations impossible.
In~\cite{kop:raa:10:1} (extended abstract) we have discussed
how the subterm property may be weakened by posing restrictions on
the rules, and in~\cite{kop:raa:11:1}, the short version of this
paper, we have explored an extension of the dynamic approach to AFSs.

\paragraaf{Static Dependency Pairs for HRSs}
The static approach in~\cite{kus:sak:07:1} is moved to the setting of
HRSs in~\cite{kus:iso:sak:bla:09:1}, and extended with argument
filterings and usable rules in~\cite{suz:kus:bla:11}.
The static approach omits dependency pairs $\up{\afun}(\vec{l})
\dppijl \app{\avar}{\vec{r}}$ with $\avar$ a variable, which avoids
the need for a subterm property, but it allows bound variables to
become free in the right-hand side of a dependency pair.
The technique is restricted to \emph{plain function passing} HRSs.
A system with for instance the (terminating) rule
$\mathsf{h}(\mathsf{g}(\abs{\avar}{F(\avar)})) \arrz
F(\mathsf{a})$ cannot be handled.  Moreover, the approach is not
complete: a terminating AFS may have a static dependency
chain.

\bigskip \noindent
The definitions for HRSs~\cite{sak:wat:sak:01,kus:iso:sak:bla:09:1}
do not immediately carry over to AFSs, since AFSs may have
rules of functional type, and $\beta$-reduction is a separate rewrite
step.
A short paper by Blanqui~\cite{bla:06:1} introduces static dependency
pairs on a form of rewriting which includes AFSs, but it poses some
restrictions, such as base-type rules.
The present work considers dynamic dependency pairs for AFSs and is
most related to \cite{sak:wat:sak:01}, but is adapted for the
different formalism.
Our method conservatively extends the one for first-order rewriting
and provides a characterisation of termination for left-linear AFSs.
We have chosen for a dynamic rather than a static approach because,
although the static approach is stronger when applicable, the dynamic
definitions can be given without restrictions.
The restrictions we do provide, to weaken the subterm property and
enable for instance argument filterings, are optional.
We will say some words about integrating the static and dynamic
approaches in Section~\ref{sec:improvements}.

\section{Preliminaries} \label{sec:preliminaries}

\summary{In this section, 
we present the formalism of Algebraic Functional Systems (AFSs).}

We consider higher-order rewriting as defined originally by
Jouannaud and Okada~\cite{jou:oka:91:1}, also called Algebraic
Functional Systems (AFSs).
Terms are built from simply-typed variables, abstraction and 
application (as in the simply-typed $\lambda$-calculus), and also
from function symbols which take a fixed number of
arguments. 
Terms and matching are modulo $\alpha$,
and
every AFS contains the $\beta$-reduction rule.
Several variations of the definition of AFSs exist;
here we roughly
follow \cite[Chapter 11.2.3]{ter:03}, which coincides with the format
currently used in the higher-order category of the annual
termination competition~\cite{termcomp}.

\paragraaf{Types and Terms}
Assuming a set $\setsorts$ of \emph{base types},
the set of \emph{simple types} (or just \emph{types}) 
is generated 
using the binary type constructor $\typepijl$, 
according to the following grammar:
\[
\settypes ::= \setsorts \mid \settypes \typepijl \settypes
\]
The arrow operator is right-associative.
Types are denoted by $\atype, \btype, \ldots$ and base types by $\abasetype,
\bbasetype\ldots$.
A type with at least one occurrence of $\ftypepijl$
is called a \emph{functional type}.
A \emph{type declaration} is an expression of the form 
$[\atype_1 \times \ldots \times \atype_n] \decpijl \btype$ with
$\btype$ and all $\atype_i$ types;
we write just $\btype$ if $n = 0$.
Type declarations are not types, but are used for typing purposes.

We assume a set $\setvar$, 
consisting of for each type infinitely many typed \emph{variables},
written as $\avar , \bvar , \cvar ,\ldots$.
We further assume a set $\setfun$, disjoint from $\setvar$,
consisting of \emph{function symbols}, equipped with a type declaration,
and written as $\afun , \bfun , \ldots$ 
or using more suggestive notation.
To stress the type (declaration) of a symbol $a$
we may write $a:\atype$.
The set of \emph{terms} over $\setfun$ consists of expressions $\aterm$
for which we can infer $\aterm : \atype$ for some type $\atype$ using
the clauses:

\begin{tabular}{lllll}
  \clausevar &
  & $\avar : \atype$
  & \mbox{} & if $\avar:\atype \in \setvar$ \\
  \clauseapp &
  & $\aterm \cdot \bterm : \btype$
  & \mbox{}& if $\aterm : \atype \typepijl \btype$ and
      $\bterm : \atype$  \\
  \clauseabs & 
  & $\abs{\avar}{\aterm} : \atype \typepijl \btype$
  & \mbox{} & if $\avar:\atype \in \setvar$ and $\aterm : \btype$ \\
  \clausefun &
  & $\afun (\aterm_1 , \ldots , \aterm_n) : \btype$
  & \mbox{} &
  if $\afun : {[\atype_1 \times \ldots \times \atype_n] \decpijl
    \btype} \in \F$ and 
  $\aterm_1 :\atype_1 , \ldots, \aterm_n:\atype_n$
\end{tabular} \\
Terms built using these clauses are called respectively
a variable, an application, an abstraction, and a functional term.
Note that a function symbol $\afun : {[\atype_1 \times \ldots \times
\atype_n] \typepijl \btype}$ takes exactly $n$ arguments, and
$\btype$ may be a functional type.  
The $\lambda$ binds occurrences of variables as in the
$\lambda$-calculus, and term equality is modulo $\alpha$-conversion
(bound variables may be renamed).
A variable in $\aterm$ which is not bound by some $\lambda$ is
\emph{free}, and the set of free variables of $\aterm$ is denoted by
$\FV(\aterm)$.
Application is left-associative.
Let $\head(\ )$ denote the head of an application, so
$\head(\app{\aterm}{\bterm}) = \head(\aterm)$ and $\head(\aterm) =
\aterm$ for non-applications.

A \emph{substitution} $[\vec{\avar}:=\vec{\aterm}]$, 
with $\vec{\avar}$ and $\vec{\aterm}$ non-empty finite vectors of equal
length, 
is the homomorphic extension of the type-preserving mapping
$\vec{\avar} \mapsto \vec{\aterm}$ from variables to terms.  
Substitutions are denoted $\asub, \bsub$,
and the result of applying $\asub$ to a term $\aterm$ 
is denoted $\subst{\aterm}{\asub}$.
The \emph{domain} $\domain(\gamma)$ of $\gamma =
[\vec{\avar}:= {\vec{\aterm}}]$ is $\{\vec{\avar}\}$.
Substituting does not capture free variables.

Let $\Box_\atype : \atype$ be a fresh symbol for every type $\atype$.
A \emph{context} $C[]$ is a term with a single occurrence of some
$\Box_\atype$.
The result of the replacement of $\Box_\atype$ in $C[]$
by a term $\aterm : \atype$
is denoted by $C[s]$.
Such replacements may capture free variables.
For example, $(\abs{\avar}{\bvar})[\bvar := \avar] =
\abs{\cvar}{\avar}$,
but for $C[] = \abs{\avar}{\Box_\atype}$ 
we have $C[\avar] = \abs{\avar}{\avar}$.

We say $\bterm$ is a \emph{subterm} of $\aterm$, notation $\aterm
\suptermeq \bterm$,
if $\aterm = C[\bterm]$ for some context $C$.
If in addition $C \neq \Box_\atype$, 
then $\bterm$ is a \emph{strict subterm} of $\aterm$,
notation $\aterm \supterm \bterm$.

\paragraaf{Rules and Rewriting}
A \emph{rewrite rule} over a set of function symbols $\setfun$
is a pair of terms $l \arrz r$ over $\setfun$
such that 
$l$ and $r$ have the same type, 
and all free variables of $r$ also occur in $l$.
In~\cite{kop:11:1} some termination-preserving transformations
on the general format of AFS-rules are presented.
Using these results, we can
additionally assume
that a left-hand side $l$ is of the form
$\app{f(l_1,\ldots,l_n)}{l_{n+1}} \cdots l_m$
(with $m \geq n \geq 0$),
and does not contain subterms of the form 
$\app{(\abs{\avar}{\aterm})}{\bterm}$.
(Many AFSs are defined in this way already.)
Note that we do not assume
$\eta$-normal
or $\eta$-exanded forms, and that we may have rules of functional
type.

Given a set of rewrite rules $\Rules$, the \emph{rewrite relation}
$\arr{\Rules}$ on terms is given by: \\
\begin{tabular}{lrcll}
\clauserule & $C[\subst{l}{\asub}]$ & $\arr{\Rules}$ &
  $C[\subst{r}{\asub}]$ & 
  with $l \arrz r \in \Rules$, $C$ a context,
  $\asub$ a substitution \\
\clausebeta & $C[(\abs{\avar}{\aterm}) \cdot \bterm]$ &
  $\arr{\Rules}$ & $C[\aterm[\avar:=\bterm]]$ & \\
\end{tabular} \\
We sometimes use the notation $\aterm \arr{\beta} \bterm$ for a rewrite
step using \clausebeta.
A \emph{headmost step} is a reduction $\aterm \arr{\Rules} \bterm$
using either clause, where $C$ has the form
$\app{\Box_\atype}{\aterm_1} \cdots \aterm_n$ with $n \geq 0$.
\smallskip

An \emph{algebraic functional system} (AFS) is a pair $(\setfun ,
\Rules)$ consisting of a set of function symbols $\setfun$ and a set
$\Rules$ of rewrite rules over $\setfun$;
it is often specified by giving only $\Rules$.

A function symbol $\afun$ is a \emph{defined symbol} of an AFS 
if there is a rule with left-hand side 
$\afun (l_1,\ldots,l_n) \cdot l_{n+1} \cdots l_m$, and a
\emph{constructor symbol} if not.
The sets of defined and constructor symbols are denoted by
$\Defineds$ and $\Constructors$ respectively.
A rewrite rule $l \arrz r$ is \emph{left-linear} if 
every variable occurs at most once free in $l$; 
an AFS is left-linear if all its rewrite rules are.
A rule $l \arrz r$ is \emph{collapsing} if $\head(r)$ is a variable.
A term $\aterm$ is called \emph{terminating} if every reduction
sequence starting in $\aterm$ is finite;
an AFS is terminating if
all its terms are.

We assume that an AFS has only finitely many rules.  
In Section~\ref{subsec:polymorphism} we shortly discuss
how to use dependency pairs to prove termination of 
AFSs with infinitely many rules.

\begin{example}\label{ex:twicedefinition}
The following AFS $\twice$ is the running example of this paper.
It has the following four function symbols:
$\nul : \nat$,  
$\suc: [\nat] \typepijl \nat$,
$\I : [\nat] \typepijl \nat$, and
$\twice : [\nat \typepijl \nat] \typepijl \nat \typepijl \nat$.
There are three rewrite rules:
\[
\begin{array}{rcl}
\I(\nul) & \arrz & \nul \\
\I(\suc(n)) & \arrz & \suc(\twice(\abs{x}{\I(x)}) \cdot n) \\
\twice(F) & \arrz & \abs{y}{F \cdot (F \cdot y)} \\
\end{array}
\]
Recall that we also have $\beta$-reduction steps.
The symbol $\I$ represents the identity function on natural numbers.
This system is terminating, but this is not trivial to prove;
neither recursive path orderings like HORPO~\cite{jou:rub:99:1}
and CPO~\cite{bla:jou:rub:08:1}, nor a static dependency pair
approach, can handle the second $\I$-rule, due to the subterm $\I(x)$.
The static approach gives a constraint $\up{\I}(\suc(n)) \gterm
\up{\I}(x)$, which is impossible to satisfy because $\gterm$ must
be closed under substitution, and $\suc(n)$ might be
substituted for the free variable
$x$.  CPO gives a similar problem.
\end{example}

\section{The First-Order Dependency Pair Approach}\label{sec:firstorder}

\summary{In this section, 
we recall the dependency pair approach for first-order rewriting.  
We emphasise those parts which are relevant for our higher-order approach.}

We assume that first-order term rewrite systems (TRSs) are already known;
they can also be thought of as AFSs 
where all function symbols have a type declaration
$[\o \times \ldots \times \o] \decpijl \o$ and where terms are
formed without clauses \clauseabs\ and \clauseapp.
In this section we recall those definitions and results from the
theory of 
dependency pairs for TRSs that we will generalize or adapt to the
higher-order setting in this paper.
The definitions here are close to those in~\cite{hir:mid:04:1};
our set-up is in between the one for the 
dependency pairs \emph{approach}~\cite{art:gie:00:1} 
and the
dependency pairs \emph{framework}~\cite{gie:thi:sch:05:2}.
This section is also meant to give additional background
for those not familiar with dependency pairs.

\subsection{Motivation}\label{subsec:motivation}
Two important properties of the termination method using
dependency pairs are that it is suitable for automation,
and that it can be used to prove termination of TRSs
which are not simply terminating.
A (well-known) example of such a TRS is:
\[
\begin{array}{rclrcl}
  \mins(\avar,\nul) & \arrz & \avar &
  \quot(\nul,\suc(\bvar)) & \arrz & \nul \\
  \mins(\suc(\avar),\suc(\bvar)) & \arrz & \mins(\avar,\bvar)\ \ \  &
  \quot(\suc(\avar),\suc(\bvar)) & \arrz &
    \suc(\quot(\mins(\avar,\bvar),\suc(\bvar))) \\
\end{array}
\]
A TRS is simply terminating if it can be proved terminating
using a \emph{reduction ordering} $\gterm$
(a well-founded ordering on terms 
which is both monotonic and closed under substitution)
that satisfies the \emph{subterm property}, which means that
$\afun(\aterm_1,\ldots, \aterm_n) \geqterm \aterm_i$ 
for every $i \in \{1 , \ldots , n\}$. 

\subsection{Dependency Pairs}\label{subsec:fo:dp}
An intuition behind the dependency pair approach is to 
identify those parts of the right-hand sides of rewrite rules which
may give rise to an infinite reduction.
Suppose we have a minimal non-termating term $t$, 
so a non-terminating term 
where all proper subterms are terminating.
An infinite reduction from $t$ has the form
$t \arrz^* l\gamma \arrz r\gamma \arrz \ldots$
with $l \rightarrow r$ a rewrite rule.
Then, a minimal non-terminating subterm of $r\gamma$
has as root-symbol a defined symbol from the pattern of $r$.
Thus, we are interested in subterms of right-hand sides of rewrite
rules with a defined symbol at the root.
Such subterms are called \emph{candidate terms} of $r$.

We obtain $\up{\F}$ by adding to each symbol $f$ in the 
signature $\F$ of a first-order TRS a symbol $\up{f}$ with
the same arity.
The \emph{dependency pairs} of a rewrite rule 
$f(l_1 , \ldots , l_n) \arrz r$
are all pairs
$\up{f} (l_1 , \ldots , l_n) \dppijl \up{\bfun} (p_1 , \ldots , p_m)$
with $r \suptermeq \bfun (p_1 , \ldots , p_m)$,
and $\bfun$ a defined symbol, 
and $\bfun(p_1,\ldots,p_m)$ not a subterm of some $l_i$.
The set of all dependency pairs of a TRS $(\F , \Rules)$ 
is denoted by $\DP(\Rules)$.

The $\quot$-example has the following dependency pairs:
\[
\begin{array}{rcl}
\up{\mins}(\suc(\avar),\suc(\bvar)) & \dppijl &
  \up{\mins}(\avar,\bvar) \\
\up{\quot}(\suc(\avar),\suc(\bvar)) & \dppijl &
  \up{\quot}(\mins(\avar,\bvar),\suc(\bvar)) \\
\up{\quot}(\suc(\avar),\suc(\bvar)) & \dppijl &
  \up{\mins}(\avar,\bvar) \\
\end{array}
\]
The first and third rewrite rule do not give dependency pairs, 
because their right-hand sides do not contain defined symbols.
The fourth rule gives two different dependency pairs.

A \emph{dependency chain} is 
a sequence $\rijtje{(l_i \dppijl p_i, \aterm_i, \bterm_i) \mid i \in \N}$, 
such that for all $i$:
\begin{enumerate}[(1)]
\item $l_i \dppijl p_i \in \DP(\Rules)$;
\item $\aterm_i = l_i \asub_i$ and $\bterm_i = p_i\asub_i$
 for some substitution $\asub_i$;
\item $\bterm_i \arrr{\Rules} \aterm_{i+1}$.
\end{enumerate}
Since $\bterm_i$ has the form $\up{\afun}(\vec{\cterm})$
and the marked symbol $\up{\afun}$ is not used in any rule, 
the reduction $\bterm_i \arrr{\Rules} \aterm_{i+1}$ does not use
headmost steps.
The chain is \emph{minimal} if all $\bterm_i$ are terminating under
$\arr{\Rules}$.
Termination of a TRS can be characterized using dependency chains:

\begin{theorem}[\cite{art:gie:00:1}]
\label{thm:fo:dpchain}
A TRS is terminating \emph{if and only if} it does not admit a
minimal dependency chain.
\end{theorem}
A higher-order generalization of this result is provided in 
Theorems~\ref{thm:dependencychain} and~\ref{thm:dependencychainll}.

\begin{example}
The TRS 
$\mathtt{nats}(n) \arrz \cons(n, \mathtt{nats}(\suc(n)))$
has an infinite dependency chain
\[
\begin{array}{lllllll}
( & \up{\mathtt{nats}}(n) \dppijl \up{\mathtt{nats}}(\suc(n)) & , &
  \up{\mathtt{nats}}(\nul) & , & \up{\mathtt{nats}}(\suc(\nul)) & ),
  \\
( & \up{\mathtt{nats}}(n) \dppijl \up{\mathtt{nats}}(\suc(n)) & , &
  \up{\mathtt{nats}}(\suc(\nul)) & , &
  \up{\mathtt{nats}}(\suc(\suc(\nul))) & ), \\
  & \ldots \\
\end{array}
\]
corresponding to
the infinite reduction $\mathtt{nats}(\nul) \arrz \mathtt{nats}(
\suc(\nul)) \arrz \mathtt{nats}(\suc(\suc(\nul))) \arrz \ldots$
\end{example}

\subsection{Using a Reduction Pair}\label{subsec:fo:reduction}
A TRS without 
infinite dependency chain is terminating.  
Absence of infinite dependency chains
can be demonstrated with a \emph{reduction pair}.
This
is a pair $(\geqterm,\gterm)$ of a \emph{quasi-order} and a
\emph{well-founded order}, such that:
\begin{iteMize}{$\bullet$}
\item $\geqterm$ and $\gterm$ are \emph{compatible}: either $\gterm
  \cdot \geqterm$ is included in $\gterm$, or $\geqterm \cdot \gterm$
  is;
\item $\geqterm$ and $\gterm$ are both \emph{stable} (preserved under
  substitution);
\item $\geqterm$ is monotonic (if $\aterm \geqterm \bterm$, then
  $C[\aterm] \geqterm C[\bterm]$).
\end{iteMize}

\begin{theorem}[\cite{art:gie:00:1}]\label{thm:fo:reductionpair}
A TRS is terminating, if there exists 
a reduction pair $(\geqterm,\gterm)$  
such that 
$l \gterm p$ for all dependency pairs $l \dppijl p$, 
and $l \geqterm r$ for all rules $l \arrz r$.
\end{theorem}

\begin{example}\label{ex:quotpoly}
The $\quot$ example is terminating if there is a reduction pair
satisfying:
\[
\begin{array}{rclrcl}
\mins(\avar,\nul) & \geqterm & \avar &
\up{\mins}(\suc(\avar),\suc(\bvar)) & \gterm & \up{\mins}(\avar,\bvar) \\
\mins(\suc(\avar),\suc(\bvar)) & \geqterm & \mins(\avar,\bvar) &
\up{\quot}(\suc(\avar),\suc(\bvar)) & \gterm &
  \up{\quot}(\mins(\avar,\bvar),\suc(\bvar)) \\
\quot(\nul,\suc(\bvar)) & \geqterm & \nul &
\up{\quot}(\suc(\avar),\suc(\bvar)) & \gterm & \up{\mins}(\avar,\bvar) \\
\quot(\suc(\avar),\suc(\bvar)) & \geqterm &
  \multicolumn{2}{l}{\suc(\quot(\mins(\avar,\bvar),\suc(\bvar)))} \\
\end{array}
\]
These constraints are oriented with a polynomial interpretation
with $\algint{\mins(x,y)} =
\algint{\up{\mins}(x,y)} = \algint{x},\ \algint{\quot(x,y)} =
\algint{\up{\quot}(x,y)} = \algint{x}+\algint{y},\ 
\algint{\suc(x)} = \algint{x}+1,\ \algint{\nul} = 0$:
\[
\begin{array}{rclrcl}
\avar & \geq & \avar &
\avar+1 & > & \avar \\
\avar+1 & \geq & \avar &
\avar+\bvar+2 & > & \avar+\bvar+1 \\
\bvar+1 & \geq & 0 &
\avar+\bvar+2 & > & \avar \\
\avar+\bvar+2 & \geq & \avar+\bvar+2 \\
\end{array}
\]
Note that we needed dependency pairs to use an interpretation like
this: the resulting ordering $\gterm$ is not monotonic, since
$\algint{\mins(\aterm,\bterm)} = \algint{\mins(\aterm,\cterm)}$ even
if $\bterm \gterm \cterm$.
\end{example}

Rather than orienting all dependency pairs at once, we can use a
step by step approach.  
Let a set of dependency pairs $\P$ be called \emph{chain-free} 
if there is no minimal dependency chain using only pairs in
$\P$\footnote{In the language of~\cite{gie:thi:sch:05:2},
this corresponds to finiteness of the DP problem 
$(\P,\emptyset,\Rules,\mathtt{minimal})$.}. 
Note that $\emptyset$ is chain-free.
The following result has Theorem~\ref{thm:maintheorem} as a
higher-order counterpart.

\begin{theorem}\label{thm:firstorder:chainfree}
A set $\P = \P_1 \uplus \P_2$ is chain-free if $\P_2$ is chain-free, and
there is a reduction pair $(\geqterm,\gterm)$ such that $l \gterm p$
for $l \dppijl p \in \P_1$,\ $l \geqterm p$ for $l \dppijl p \in \P_2$
and $l \geqterm r$ for $l \arrz r \in \Rules$.
\end{theorem}

\subsection{Argument Filterings}\label{subsec:fo:argfil}

In order to obtain a reduction pair, there are two ways we could go:
either we use approaches like the polynomial interpretations given in
Example~\ref{ex:quotpoly}, which directly give us a pair $(\geqterm,
\gterm)$ (where $\gterm$ may be non-monotonic),
or we use an existing
reduction ordering or pair
and adapt it with \emph{argument filterings}.
An argument filtering is a function $\argfil$ which maps terms of the
form $\afun(\avar_1,\ldots,\avar_n)$ with $\afun \in \up{\F}$ either
to a term $\afun_\argfil(\avar_{i_1},\ldots,\avar_{i_m})$
or to one of the $\avar_i$.
An argument filtering is applied to a term as follows:
\[
\begin{array}{rcll}
\filter(\afun(\aterm_1,\ldots,\aterm_n)) & = &
  \afun_\argfil(\filter(\aterm_{i_1}),\ldots,\filter(\aterm_{i_m})))
  & \mathrm{if}\ \argfil(\afun(\vec{\avar})) =
  \afun_\argfil(\avar_{i_1},\ldots,\avar_{i_m}) \\
\filter(\afun(\aterm_1,\ldots,\aterm_n)) & = &
  \filter(\aterm_i) & \mathrm{if}\ \argfil(\afun(\vec{\avar})) =
  \avar_i \\
\filter(\avar) & = & \avar & \mathrm{if}\ \avar\ \mathrm{a\ 
  variable} \\
\end{array}
\]
Phrased differently, $\filter(\afun(\aterm_1,\ldots,\aterm_n)) =
\argfil(\afun(\avar_1,\ldots,\avar_n))[\avar_1:=\filter(\aterm_1),
\ldots,\avar_n:=\filter(\aterm_n)]$.
Note that an argument filtering works both on unmarked and on
marked symbols.

Using argument filterings, we can eliminate troublesome subterms of
the dependency pair constraints.  To this end, we use the following
result, which corresponds to Theorem~\ref{thm:argfunrespectbeta}:

\begin{theorem}\label{thm:fo:argfil}
Given a reduction pair $(\geq,>)$ and an argument filtering
$\argfil$, define: $\aterm \geqterm \bterm$ iff $\filter(\aterm) \geq
\filter(\bterm)$ and $\aterm \gterm \bterm$ iff $\filter(\aterm) >
\filter(\bterm)$.  Then $(\geqterm,\gterm)$ is a reduction pair.
\end{theorem}

\begin{example}
Recall the constraints given in Example~\ref{ex:quotpoly}.  We use the
argument filtering with $\argfil(\mins(\avar,\bvar)) = \avar,\ 
\argfil(\quot(\avar,\bvar)) = \quot_\argfil(\avar)$ and
$\argfil(\afun(\vec{\avar})) = \afun_\argfil(\vec{\avar})$ for all
other symbols.  By Theorem~\ref{thm:fo:argfil} it suffices to find a
reduction ordering satisfying:
\[
\begin{array}{rclrcl}
\avar & \geq & \avar &
\up{\mins}_\argfil(\suc_\argfil(\avar),\suc_\argfil(\bvar)) & > &
  \up{\mins}_\argfil(\avar,\bvar) \\
\suc_\argfil(\avar) & \geq & \avar &
\up{\quot}_\argfil(\suc_\argfil(\avar),\suc_\argfil(\bvar)) & > &
  \up{\quot}_\argfil(\avar,\suc_\argfil(\bvar)) \\
\quot_\argfil(\nul_\argfil) & \geq & \nul_\argfil &
\up{\quot}_\argfil(\suc_\argfil(\avar),\suc_\argfil(\bvar)) & > &
  \up{\mins}_\argfil(\avar,\bvar) \\
\quot_\argfil(\suc_\argfil(\avar)) & \geq &
  \suc_\argfil(\quot_\argfil(\avar)) \\
\end{array}
\]
These altered constraints can easily be satisfied with a
lexicographic path ordering.
\vspace{-6pt}
\end{example}

\subsection{The Subterm Criterion}\label{subsec:fo:subcrit}

An alternative to reduction pairs, which often suffices to
eliminate some dependency pairs and is typically easy to
check, is the \emph{subterm criterion}.

A function $\nu$ that assigns 
to every $n$-ary dependency pair symbol $\up{\afun}$
one of its argument positions
is said to be a \emph{projection function}.
 We extend $\nu$
to a function on terms by defining:
\[\overline{\nu}(\up{\afun}(\aterm_1,\ldots,\aterm_n)) = \aterm_i\ 
\mathrm{if}\ \nu(\up{\afun}) = i\]

\begin{theorem}\label{thm:fo:subcrit}
A set of dependency pairs $\P = \P_1 \uplus \P_2$ is chain-free if
$\P_2$ is chain-free, and moreover there is a projection function
$\nu$ such that $\overline{\nu}(l) \supterm \overline{\nu}(p)$ for
all dependency pairs $l \dppijl p \in \P_1$ and $\overline{\nu}(l)
= \overline{\nu}(p)$ for all $l \dppijl p \in \P_2$.
\end{theorem}

This holds because for all $\aterm_i$ and $\bterm_i$ in a dependency chain,
$\overline{\nu}(\aterm_i)$ and $\overline{\nu}(\bterm_i)$ are strict
subterms, and therefore terminating under 
$\arr{\Rules} \mathord{\cup}\ \supterm$; see also
Theorem~\ref{thm:subcrit}.

The subterm criterion
is not sufficient to show termination
of the $\quot$ example, but at least we can use it to eliminate some
dependency pairs; choosing $\nu(\up{\mins}) = \nu(\up{\quot}) = 2$:
\[
\begin{array}{rcccccl}
\overline{\nu}(\up{\mins}(\suc(\avar),\suc(\bvar))) & = &
\suc(\bvar) & \supterm & \bvar & = & \overline{\nu}(\up{\mins}(\avar,
\bvar)) \\
\overline{\nu}(\up{\quot}(\suc(\avar),\suc(\bvar))) & = &
\suc(\bvar) & = & \suc(\bvar) & = & \overline{\nu}(
\up{\quot}(\mins(\avar,\bvar),\suc(\bvar))) \\
\overline{\nu}(\up{\quot}(\suc(\avar),\suc(\bvar))) & = &
\suc(\bvar) & \supterm & \bvar & = & \overline{\nu}(\up{\mins}(
\avar,\bvar)) \\
\end{array}
\]
This shows that the TRS $\quot$ is non-terminating if and only if
there is no dependency chain where every step uses the dependency
pair $\up{\quot}(\suc(\avar),\suc(\bvar)) \dppijl
\up{\quot}(\mins(\avar,\bvar),\suc(\bvar))$.

\subsection{The Dependency Graph}\label{subsec:fo:graph}
To determine whether a system has a dependency chain, it
makes sense to ask what form such a chain would have.  This question
is studied with a
\begin{wrapfigure}{r}{0.53\textwidth}
\begin{tikzpicture}[->]
\begin{scope}[>=stealth]

\tikzstyle{veld} = [draw, fill=white,  drop shadow, 
  minimum height=0em, minimum width=0em, rounded corners]

\node [veld] (mins1) {
    $\up{\mins}(\suc(\avar),\suc(\bvar)) \dppijl \up{\mins}(\avar,\bvar)$
  };
\node (minsloop) [left of=mins1,node distance=2.9cm] {};

\node [veld] (quot1) [below of=mins1,node distance=1.2cm] {
    $\up{\quot}(\suc(\avar),\suc(\bvar)) \dppijl \up{\mins}(\avar,\bvar)$
  };

\node [veld] (quot2) [below of=quot1,node distance=1.2cm] {
    $\up{\quot}(\suc(\avar),\suc(\bvar)) \dppijl \up{\quot}(\mins(
    \avar,\bvar),\suc(\bvar))$
  };
\node (quotloop) [left of=quot2,node distance=3.7cm] {};

\draw (quot2) -- (quot1);
\draw (quot1) -- (mins1);
\draw (minsloop) to[out=135,in=225,looseness=5] (minsloop);
\draw (quotloop) to[out=135,in=225,looseness=5] (quotloop);

\end{scope}
\end{tikzpicture}
\end{wrapfigure}
\noindent
\emph{dependency graph}, a graph with as nodes the
dependency pairs of $\Rules$ and an edge
from $l \dppijl p$ to $u \dppijl v$ if $p\asub \arrr{\Rules}
u\bsub$ for some substitutions $\asub,\bsub$.
See for example the dependency graph of the $\quot$-TRS.

If there is a dependency chain $\rijtje{(\rho_i,
\aterm_i,\bterm_i) \mid i \in \N}$, 
then there is an edge in the graph from
each $\rho_i$ to $\rho_{i+1}$.  
Since the graph is finite,
a dependency chain corresponds to a cycle in the graph.

By definition, if a set of dependency pairs is chain-free, then the
same holds for any subset.  Since a dependency graph might have
exponentially many cycles, modern approaches typically consider only
\emph{maximal cycles}, also called \emph{strongly connected
components} (SCCs).

\begin{theorem}\label{thm:fo:graph}
$\Rules$ is terminating iff every SCC of its dependency
graph is chain-free.
\end{theorem}

\noindent
This result is extended to AFSs in Lemma~\ref{lem:cyclenondangerous}.
The dependency graph is not in general computable, which is why
\emph{approximations} are often used.  An approximation of a
dependency graph $G$ is a graph with the same nodes as $G$, but
which may have additional edges.

The dependency graph of our running example $\quot$ has two cycles.
In order to prove termination,
it is sufficient to find a reduction pair such that
$\up{\mins}(\suc(\avar),\suc(\bvar)) \gterm \up{\mins}(\avar,\bvar)$
and $l \geqterm r$ for all rules, 
and a(nother) reduction pair with
$\up{\quot}(\suc(\avar),\suc(\bvar)) \gterm
\up{\quot}(\mins(\avar,\bvar), \suc(\bvar))$ 
and $l \geqterm r$ for all rules.  
The fact that we can deal with groups of dependency pairs separately
can make it significantly simpler to find reduction pairs.

Extending Theorem~\ref{thm:fo:graph}, we can iterate over a
dependency graph approximation, and obtain the following algorithm
(whose higher-order counterpart is presented in
Section~\ref{sec:algorithm}):

\begin{theorem}[\cite{hir:mid:05:1}]\label{thm:fo:alg}
A TRS is terminating if and only if this can be demonstrated with the
following algorithm:
\begin{enumerate}[\em(1)]
\item calculate dependency pairs, find an approximation $G$ for the
  dependency graph;
\item\label{alg:fo:start}
  if $G$ has no cycles, the TRS is terminating; otherwise, choose an
  SCC $\P$
\item\label{alg:fo:subcrit}
  try finding a projection function $\nu$ such that
  $\overline{\nu}(l) \supterm \overline{\nu}(p)$ for at least some
  $l \dppijl p \in \P$ and $\overline{\nu}(l) \suptermeq \overline{
  \nu}(p)$ for the rest; if this succeeds, remove the strictly
  oriented pairs from $G$ and continue with \ref{alg:fo:start}.
\item \label{alg:fo:main}
  find a reduction
  pair $(\geqterm,\gterm)$ such that $l \geqterm r$ for all rules $l
  \arrz r$, and for all dependency pairs $l \dppijl p \in \P$ either
  $l \gterm p$ or $l \geqterm p$; at least one pair must
  be oriented with $\gterm$ \emph{(**)};
\item remove the dependency pairs which were oriented with $\gterm$
  from $G$; continue with \ref{alg:fo:start}.
\end{enumerate}

\noindent\emph{(**)} To find $(\geqterm,\gterm)$ we may for instance
use argument filterings.
\end{theorem}\vspace{-9 pt}

\subsection{Usable Rules}\label{subsec:fo:userules}

We discuss one more optimisation.  
In the algorithm of Theorem~\ref{thm:fo:alg},
we consider in every iteration for a strongly connected
component $\P$ all rewrite rules.
Instead, 
we can restrict attention to the rules
that may be relevant for constructing a dependency chain
using dependency pairs from $\P$.
To this end the concept of \emph{usable rules} is defined.

First we need some definitions.
We denote by $\afun \gsymbuse \bfun$ 
that there is a rewrite rule
$\afun(l_1,\ldots,l_n) \arrz C[\bfun(r_1,\ldots,r_m)]$.
The reflexive-transitive closure of 
$\gsymbuse$ is denoted by $\gsymbuse^*$.
Overloading notation,
we write $\aterm \gsymbuse^* \bfun$
if there is a symbol $\afun$ in the term $\aterm$ such that 
$\afun \gsymbuse^*\bfun$.
So if not
$\aterm \gsymbuse^* \bfun$, then $\aterm$ cannot reduce to
a term containing the symbol $\bfun$.

\begin{definition}
The set of \emph{usable rules} of a term $\aterm$, notation
$\userules(\aterm)$,  consists of rules $\bfun(\vec{l}) \arrz r
\in \Rules$, where
$\aterm \gsymbuse^* \bfun$.  
For a set of dependency pairs $\P$, let $\userules(\P) =
\bigcup_{l \dppijl p \in \P} \userules(p)$.
\end{definition}

Using a reasoning originally due to Gramlich~\cite{gra:95:1}, 
it is shown in~\cite{hir:mid:04:1} that if $\P$ is not
chain-free, then there is a dependency chain over $\P$ where
the reduction $\bterm_i \arrz^* \aterm_{i+1}$ uses only the rules in
$\userules(\P) \cup \{ \symb{p}(\avar,\bvar) \arrz \avar,\ 
\symb{p}(\avar,\bvar) \arrz \bvar \}$ for some fresh symbol
$\symb{p}$ (these two rules are
usually considered harmless).  Thus, in each iteration of the
algorithm of Theorem~\ref{thm:fo:alg} we only have to prove $l
\geqterm r$ for the usable rules of $\P$, rather than for all rules.

\section{The Basic Higher-Order Dependency Pair Approach}
\label{sec:basic}

\summary{In this section we define a basic dependency pair
approach for AFSs.  
We show that an AFS is terminating if it does not have a minimal
dependency chain,
and that
for left-linear AFSs,
the absence of (minimal) dependency chains
characterises termination.
As in the first-order case, 
we organise the dependency pairs in a graph, 
and explain how to use reduction pairs.
}

\noindent
When extending the first-order dependency pair approach to AFSs,
new issues arise:
\begin{iteMize}{$\bullet$}
\item \emph{collapsing rules}: non-termination might also be caused
  by a higher-order variable being instantiated. 
  For example, the
  right-hand side of the non-terminating rule $\symb{f}(\symb{g}(F),
  \avar) \arrz \app{F}{\avar}$ doesn't even have defined
  symbols;
\item \emph{dangling variables}: given a rule $\mathtt{f}(\nul) \arrz
  \mathtt{g}(\abs{\avar}{\mathtt{f}(\avar)})$, the bound variable
  $\avar$ should probably not become free in the corresponding
  dependency pair;
\item \emph{rules of functional type} may lead to non-termination
  only because of their interaction with the (applicative)
  context they appear in;
\item \emph{typing issues}: to be able to use the usual term orderings,
  both sides of a dependency pair (or the constraints generated from
  it) should have, usually, the same type modulo renamings of base
  types.
\end{iteMize}

\noindent
Typing issues will be addressed in Section~\ref{subsec:typechange};
for the other problems we have to take precautions already in the
definition of dependency pairs.\vspace{-6 pt}

\subsection{Dependency Pairs}\label{subsec:deppair}

In order to define dependency pairs, we first
pre-process the rewrite rules and define candidate terms.
The complete definition of dependency pairs  may at first seem
somewhat baroque; this
is partly because we have to work around the issues of functional
rules and dangling variables, and partly because of several
optimisations we include to obtain an easier result system.

\paragraaf{Pre-processing} 
Pre-processing the rewrite rules is done by completion:

\begin{definition}[Pre-processing Rules] \label{def:completing}
An AFS is \emph{completed} by adding for each rule of the form $l
\arrz \abs{\avar_1 \ldots \avar_n}{r}$, with $r$ not an abstraction,
the following $n$ rewrite rules:
$\app{l}{\avar_1} \arrz \abs{\avar_2 \ldots \avar_n}{r},\ 
\ldots,\ 
\app{l}{\avar_1} \cdots \avar_n \arrz r$.
\end{definition}

Note that completing an AFS has no effect on termination, 
since the added rules can be simulated 
by at most $n+1$ steps using only the original rules.

\begin{example}\label{ex:completetwice}
The system $\twice$ from Example~\ref{ex:twicedefinition} is
completed by adding the rewrite rule $\app{\twice(F)}{n} \arrz
\app{F}{(\app{F}{n})}$.
\end{example}

\noindent
\emph{In the remainder of the paper, we assume that all AFSs are
completed.}

\medskip
\noindent
To understand why completion is necessary, consider the AFS with a
single rule $\symb{f}(\nul) \arrz
\abs{\avar}{\app{\symb{f}(\avar)}{\avar}}$.
The term $\symb{f}(\nul)$ in this AFS is terminating, but there
\emph{is} an infinite reduction
$\app{\symb{f}(\nul)}{\nul} \arrz \app{(\abs{\avar}{\app{\symb{f}(
\avar)}{\avar}})}{\nul} \arrz_\beta \app{\symb{f}(\nul)}{\nul}
\arrz \ldots$.
Rules like
this
might complicate the
analysis of dependency chains, because the important step does not
happen at the top.  The pre-processing makes sure that it could also
be done with a topmost step: $\app{\symb{f}(\nul)}{\nul}$ self-reduces
with a single step using the new rule $\app{\symb{f}(\nul)}{\avar}
\arrz \app{\symb{f}(\avar)}{\avar}$ which was added by completion.

It is worth noting that we did not add new rules for all functional
rules, only for those where the right-hand side is an abstraction.
A rule $\symb{f}(\nul) \arrz \symb{f}(\symb{A})$ of functional type is
left alone.  This is an optimisation: it would be natural to add a
rule $\app{\symb{f}(\nul)}{\avar} \arrz
\app{\symb{f}(\mathtt{A})}{\avar}$, but this might give a dependency
pair $\app{\symb{f}(\nul)}{\avar} \dppijl \up{A}$ which won't be
needed.  Instead of completing this rule, we will later add a
special dependency pair for it.

\paragraaf{Candidate terms} In the first-order definition of
dependency pairs, we identify subterms that may give rise to an
infinite reduction.  Taking subterms in a system with binders is
well-known to be problematic because bound variables may become free.
One solution is to substitute fresh constants in the place of a bound
variable which would otherwise become free.  In this way,
$\app{F}{\c}$ is a ``subterm'' of
$\abs{\avar}{\app{F}{(\app{F}{\avar})}}$.  This is the approach we
take here.

We assume for every type $\atype$ a fresh symbol $\c_\atype : \atype$.
Sometimes the sub-script indicating the type is omitted.  The set of
all those symbols is denoted by $\Constants$.  The symbols $\c_\atype$
are used to replace bound variables which become free by taking a
subterm.

\begin{definition}[Candidate Terms]\label{def:candidates}
Let $r$ be a term in an AFS.
A subterm $\bterm$ of $r$ is a \emph{candidate term} of $r$ 
if either
  $\bterm = \app{\afun(\bterm_1,\ldots,\bterm_m)}{\bterm_{m+1}}
  \cdots \bterm_n$ with $\afun$ a defined symbol and
  $n \geq m \geq 0$,
or
  $\bterm = \app{\avar}{\bterm_1} \cdots \bterm_n$ with
  $\avar$ free in $r$ and $n > 0$.

If $\bterm$ is a candidate term of $r$, and $\{\avar_1:\atype_1,
\ldots,\avar_n:\atype_n\}$ is the set of variables which occur bound
in $r$ but free in $\bterm$, then $\bterm[\avar_1:=\c_{\atype_1},
\ldots,\avar_n:=\c_{\atype_n}]$ is a \emph{closed candidate term} of
$r$.
We denote the set of closed candidate terms of $r$ by
$\candidatesof{r}$.
\end{definition}

\noindent
In the AFS $\twice$ we have $\candidatesof{F \cdot (F \cdot m)} =
\{ F \cdot (F \cdot m),\ F \cdot m \}$ and 
$\candidatesof{\suc(\twice(\abs{x}{\I(x)}) \cdot n)} = 
\{\twice(\abs{x}{\I(x)}) \cdot n, \twice(\abs{x}{\I(x)}),\ \I(\c_\nat)
\}$.  
If $\symb{f}$ is a defined symbol, then the candidate terms of 
$\app{\app{\app{\symb{f}(\symb{a})}{\symb{b}}}{\symb{c}}}{\symb{d}}$
are
$\symb{f}(\symb{a}),\ \app{\symb{f}(\symb{a})}{\symb{b}},\ 
\app{\app{\symb{f}(\symb{a})}{\symb{b}}}{\symb{c}}$ and
$\app{\app{\app{\symb{f}(\symb{a})}{\symb{b}}}{\symb{c}}}{\symb{d}}$.
Note that for example $x \cdot y$ is not a candidate term of
$\symb{g}(\abs{\avar}{\avar \cdot \bvar})$ because $\avar$ occurs
only bound.

\paragraaf{Dependency Pairs}
As in the first-order case, the definition of dependency pair
uses marked function symbols.
Let $\up{\setfun} = \setfun  \cup \{ \up{\afun}:\atype \,|\, 
\mbox{$\afun : \atype \in \Defineds$} \}$,
so $\setfun$ extended with for every defined symbol $\afun$ a
marked version $\up{\afun}$ with the same type declaration.
We denote by $\up{\setfun}_c$ the union of $\up{\setfun}$ and
$\Constants$.
The marked counterpart of a term $\aterm$, notation $\up{\aterm}$, is
$\up{\afun}(\aterm_1,\ldots,\aterm_n)$ if $\aterm =
\afun(\aterm_1,\ldots,\aterm_n)$ with $\afun$ in $\Defineds$, and
just $s$ otherwise.
For example,
$\up{(\twice(F))} = \up{\twice}(F)$ and 
$\up{(\twice(F) \cdot m  )} = \twice(F) \cdot m$.
Applications are not marked.

\begin{definition}[Dependency Pair]\label{def:deppair}
The set of \emph{dependency pairs} of a rewrite rule $l \arrz r$,
notation $\DP(l \arrz r)$, consists of:
\begin{iteMize}{$\bullet$}
  \item
  all pairs $\up{l} \dppijl \up{p}$ with $p \in \candidatesof{r}$
  such that $p$ is no strict subterm of $l$;
  \item
  if $l$ has a functional type $\atype_1 \typepijl \ldots \typepijl
  \atype_n \typepijl \abasetype$ ($n \geq 1$) and $\head(r)$ is
  either a variable or a term $\afun(\vec{\aterm})$ with $\afun \in
  \Defineds$:
  all pairs $\app{l}{\bvar_1} \cdots \bvar_k \dppijl \app{r}{\bvar_1} \cdots \bvar_k$
  with $k \in \{1 , \ldots , n\}$ and all $\bvar_i$ are fresh variables.
\end{iteMize}
We use $\DP(\Rules)$ (or just $\DP$ if $\Rules$ is clear from context)
for the set of all dependency pairs of rewrite rules of an AFS $\Rules$.
\end{definition}

\begin{example}\label{ex:dptwice}
The set of dependency pairs of the AFS $\twice$ consists of:
\[
\begin{array}{rclrcl}
\up{\I}(\suc(n)) & \dppijl & \app{\twice(\abs{x}{\I(x)})}{n} &
\up{\twice}(F) & \dppijl & \app{F}{(\app{F}{\c_\nat})} \\
\up{\I}(\suc(n)) & \dppijl & \up{\twice}(\abs{x}{\I(x)}) &
\up{\twice}(F) & \dppijl & \app{F}{\c_\nat} \\
\up{\I}(\suc(n)) & \dppijl & \up{\I}(\c_\nat) &
\app{\twice(F)}{m} & \dppijl & \app{F}{(\app{F}{m})} \\
& & & \app{\twice(F)}{m} & \dppijl & \app{F}{m} \\
\end{array}
\]
The last two dependency pairs originate from the rule added by
completion.  
\end{example}

\noindent
The second form of dependency pair deals with functional rules whose
right-hand side is not an abstraction.  To illustrate why they are
necessary, consider the system 
with function symbols
${\sf{A}} : [\o] \decpijl \o \typepijl \o$ and 
${\sf{B}} : [\o \typepijl \o] \decpijl \o$,
and one rewrite rule:
${\sf{A}}({\sf{B}}(F)) \arrz F$.
This system has no dependency pairs of the first kind, but does
admit a two-step loop:
$\aterm := 
{\sf{A}} ( {\sf{B}} ( \abs{x}{ {\sf{A}} (x) \cdot x})) 
\cdot 
{\sf{B}} ( \abs{x}{{\sf{A}} (x) \cdot x})
\arrz
(\abs{x}{{\sf{A}} (x) \cdot x}) \cdot {\sf{B}} ( \abs{x}{{\sf{A}}(x) \cdot x})
\arr{\beta}
s$.
The rule \emph{does} have a dependency pair of the second form,
$\app{{\sf A}({\sf B}(F))}{\avar} \dppijl \app{F}{\avar}$.

Comparing our approach to static dependency pairs as defined
in \cite{kus:iso:sak:bla:09:1}, the two main differences are that we
avoid bound variables becoming free, and that we include
\emph{collapsing} dependency pairs, where the right-hand side is
headed by a variable.

\subsection{Dependency Chains}

We can now investigate termination using \emph{dependency chains}:

\begin{definition} \label{def:dependencychain}
A \emph{dependency chain} is an infinite sequence 
$\rijtje{(\rho_i,\aterm_i, \bterm_i)\ |\ i \in \N}$
 such that for all $i$:
\begin{enumerate}[(1)]
\item $\rho_i \in \DP \cup \{\mathtt{beta}\}$;
\item \label{chain:dp}
  if $\rho_i = l_i \dppijl p_i \in \DP$ then there exists a
  substitution $\asub$
  such that $\aterm_i =
  \subst{l_i}{\asub}$ and $\bterm_i = \subst{p_i}{\asub}$;
\item if $\rho_i = \mathtt{beta}$ then $\aterm_i =
  \app{\app{(\abs{\avar}{\cterm})}{\dterm}}{\eterm_1} \cdots
  \eterm_k$ and either
  \begin{enumerate}[(a)]
  \item \label{chain:betahead}
    $k > 0$ and $\bterm_i = \app{\cterm[\avar:=\dterm]}{\eterm_1}
    \cdots \eterm_k$, or
  \item \label{chain:subterm}
    $k = 0$ and there exists a term $\eterm$ such that
    $\cterm \suptermeq \eterm$ and $\avar \in \FV(\eterm)$ and
    $\bterm_i = \up{\eterm}[\avar:=\dterm]$, but $\eterm \neq \avar$;
  \end{enumerate}
\item \label{chain:inn} $\bterm_i \arrr{in} \aterm_{i+1}$.
\end{enumerate}
A step $\arr{in}$ is obtained by rewriting some $\cterm_i$ inside a
term of the form 
$\app{\afun(\cterm_1,\ldots,\cterm_n)}{\cterm_{n+1}} \cdots \cterm_m$.
If $\bterm_i = \aterm_{i+1}$, then also $\bterm_i \arrr{in}
\aterm_{i+1}$, regardless of whether $\bterm_i$ has this form.
A dependency chain is \emph{minimal} if the strict subterms of each
$\bterm_i$ are terminating under $\arr{\Rules}$.
\end{definition}

This definition corresponds to the first-order definition, except
that a case for $\beta$-reduction is used, and that we explicitly
require that $\bterm_i \arrr{in} \aterm_{i+1}$: this is necessary
because $\bterm_i$ may be an application rather than a functional
term, and consequently may not be marked.

\begin{theorem}\label{thm:dependencychain}
If $\Rules$ is non-terminating, there is a minimal dependency chain
over $\DP(\Rules)$.
\end{theorem}

\begin{proof}
Given any non-terminating term, let $\fterm_{-1}$ be a minimal-sized
subterm that is still non-terminating ($\fterm_{-1}$ is \emph{MNT},
or \emph{Minimal Non-Terminating}).
We make the
observations:
\begin{enumerate}[(i)]
\item
\emph{If an MNT term is reduced at a non-top position, the
result is either also MNT, or terminating.}
This holds because, if $\fterm = C[\aterm] \arr{\Rules} C[\bterm]$
because $\aterm \arr{\Rules} \bterm$, and $\bterm$ is
non-terminating, then so is $\aterm$, contradicting minimality of
$\fterm$ unless $C = \Box_\atype$.
\item
\emph{If $\cterm \arrr{in} \dterm$, then $\up{\cterm}
\arrr{in} \up{\dterm}$.}
This holds by the nature of an internal step.
\end{enumerate}

\noindent
For any $i \in \N \cup \{-1\}$, let $\fterm_i$ be a MNT term,
and $\bterm_i = \up{\fterm_i}$.  
Then $\fterm_i$ is not an
abstraction, as abstractions can only be reduced by reducing their
immediate subterm, contradicting minimality.
For the
same reason $\fterm_i$ cannot have the form $\app{\avar}{\cterm_1}
\cdots \cterm_n$ with $\avar$ a variable, or
$\app{\afun(\cterm_1,\ldots,\cterm_n)}{\cterm_{n+1}} \cdots \cterm_m$
with $\afun$ a constructor symbol.
What remains are the forms:
\begin{enumerate}[(A)]
\item $\fterm_i = \app{\app{(\abs{\avar}{\cterm})}{\dterm}}{
  \eterm_1} \cdots \eterm_n$;
\item $\fterm_i = \app{\afun(\dterm_1,\ldots,\dterm_n)}{
  \dterm_{n+1}} \cdots \dterm_m$ with $\afun \in \Defineds$.
\end{enumerate}

We consider an infinite reduction starting in $\fterm_i$.
By minimality of $\fterm_i$ eventually a headmost step must be taken.
In case (A) this must be a $\beta$-step because the left-hand sides
of rules have the form $\app{\afun(\vec{l_1})}{\vec{l_2}}$;
therefore, 
the reduction
has the form
$\fterm_i \arrr{\Rules}
\app{\app{(\abs{\avar}{\cterm'})}{\dterm'}}{\eterm_1'} \cdots
\eterm_n' \arr{\beta} \app{\cterm'[\avar:=\dterm']}{\eterm_1'} \cdots
\eterm_n' \arr{\Rules} \ldots$  Since also $\app{\cterm[\avar:=
\dterm]}{\eterm_1} \cdots \eterm_n \arrr{\Rules} \app{\cterm'[\avar:=
\dterm']}{\eterm_1'} \cdots \eterm_n'$ the immediate beta-reduct of
$\fterm_i$ is non-terminating as well.  There are two sub-cases:
\begin{iteMize}{$\bullet$}
\item If $n > 0$, this reduct is MNT by (i); in this case choose
  $\fterm_{i+1} := \app{\cterm[\avar:=\dterm]}{\eterm_1} \cdots
  \eterm_n$ and let $\rho_{i+1},\aterm_{i+1},\bterm_{i+1} := \cbeta,
  \cterm_i, \cterm_{i+1}$.  Note that $\up{\aterm_{i+1}} = \aterm_{i+
  1}$ and $\up{\bterm_{i+1}} = \bterm_{i+1}$, and that case
  \ref{chain:betahead} of the definition of a dependency chain is
  satisfied.
\item If $n = 0$, let $\eterm$ be a minimal-sized subterm of $\cterm$
  where $\eterm[\avar:=\dterm]$ is still non-terminating.  By
  minimality of $\fterm_i$ both $\eterm$ and $\dterm$ are
  terminating, so $\FV(\eterm)$ contains $\avar$, but not $\eterm =
  \avar$.
  Since $\eterm$ is not a variable,
  $\up{(\eterm[\avar:=\dterm])} = \up{\eterm}[\avar:=\dterm]$.
  By minimality of $\eterm$, also $\eterm[\avar:=\dterm]$ is MNT
  (its direct subterms have the form $\eterm'[\avar:=\dterm]$ for
  a subterm $\eterm'$ of $\eterm$).
  Case \ref{chain:subterm} is satisfied with $\cterm_{i+1} :=
  \eterm[\avar:=\dterm]$ and $\rho_{i+1},\aterm_{i+1},\bterm_{i+1}
  := \cbeta, \cterm_i,\up{\cterm_{i+1}}$.
\end{iteMize}
Note that in both sub-cases, case \ref{chain:inn} is also satisfied,
since $\bterm_i = \aterm_{i+1}$.

In case (B), $\fterm_i =
\app{\afun(\dterm_1,\ldots,\dterm_n)}{\dterm_{n+1}} \cdots \dterm_m$,
we can always find a rule $l \arrz r$ and
term $\fterm_i' = \app{\subst{l}{\asub}}{\dterm_{j+1}'} \cdots
\dterm_m'$ such that $\fterm_i \arrr{in} \fterm_i'$, and
$\app{\subst{r}{\asub}}{\dterm_{j+1}'} \cdots \dterm_m'$ is still
non-terminating.
Choose $\aterm_{i+1} := \up{\fterm_i'}$; requirement
\ref{chain:inn} from Definition \ref{def:dependencychain} is satisfied
by (ii).
Since the rules were completed, we can assume that either $m = j$ or
$r$ is not an abstraction: if $r = \abs{\avar}{r'}$ and $m > j$ then
$\app{\subst{r}{\asub}}{\dterm_{j+1}'} \cdots \dterm_m'$ is a
$\beta$-redex, and (like above) may be reduced immediately without
losing termination; the same result would have been obtained with the
rule $\app{l}{\avar} \arrz r'$.

If $m > j$, then by (i) $\app{\subst{r}{\asub}}{\dterm_{j+1}'}
\cdots \dterm_m'$ is MNT.  Consequently, $\head(\subst{r}{\asub})$
cannot be a variable or a functional term $\bfun(\vec{\eterm})$ with
$\bfun$ a constructor symbol: either $\subst{r}{\asub}$ is headed by
an abstraction, or by a functional term with root symbol in
$\Defineds$.
Since $r$ itself is not an abstraction, its head must be a
variable or a functional term with defined root symbol.  Either way,
$\rho_{i+1} :=
\app{l}{\avar_{j+1}} \cdots \avar_m \dppijl \app{r}{\avar_{j+1}}
\cdots \avar_m$ is a dependency pair.  Let $\fterm_{i+1} :=
\app{\subst{r}{\asub}}{\dterm_{j+1}'} \cdots \dterm_m'$, and
$\bterm_{i+1} := \fterm_{i+1}$ (which equals
$\up{\fterm_{i+1}}$ as this is an application).
Requirement \ref{chain:dp} is satisfied.

Finally, if $m = j$, then $\fterm_i' = \subst{l}{\asub}$ and
$\subst{r}{\asub}$ is non-terminating.  Let $p$ be the smallest
subterm of $r$ such that $\subst{\subst{p}{[\vec{\avar}:=\vec{\c}]}}{
\gamma}$ is non-terminating, where $\{\vec{\avar}\} = \FV(p)
\setminus \FV(r)$.  Then $p$ is not a variable, for each
$\asub(\avar)$ is a subterm of $\subst{l}{\asub}$, and therefore
terminating (and the $\c_\atype$ do not reduce).  Thus, the immediate
subterms of $\subst{p}{[\vec{\avar}:=\vec{\c}]\asub}$ all have the
form $\subst{p'[\vec{\avar}:=\vec{\c}]}{\asub}$
with $p \supterm p'$, and are therefore terminating by minimality of
$p$: $\subst{\subst{p}{[\vec{\avar}:=\vec{\c}]}}{\asub}$ is MNT.  As
observed before, this can only be the case if this term is headed by
an abstraction or by a functional term with a defined root symbol.
And that can only be the case if $p$ is either headed by a functional
term with defined root symbol, or is an application headed by a
variable which is free in $r$ (as $r$ has no subterms $\app{(\abs{
\avar}{\cterm})}{\dterm}$).  Thus,
$\subst{p}[\vec{\avar}:=\vec{\c}]$ is a closed candidate term of $r$.
As $p[\vec{\avar}:=\vec{\c}]\asub$ is non-terminating, it is not a
strict subterm of the MNT term $l\asub$, so $\rho_{i+1} := \up{l}
\dppijl\up{\subst{p}{[\vec{\avar}:=\vec{\c}]}}$ is a dependency pair.
Choose $\fterm_{i+1} := \subst{\subst{p}{[\vec{\avar}:=\vec{\c}]}}{
\asub}$ and $\bterm_{i+1} := \up{\fterm_{i+1}} =
\subst{\up{\subst{p}{[\vec{\avar}:=\vec{\c}]}}}{\asub}$ (since
$\subst{p}{[\vec{\avar}:=\vec{\c}]}$ is not a variable).
We see that in this case, too, requirement \ref{chain:dp} is
satisfied.
\end{proof}

\begin{example}\label{ex:nonterm}
As we will see, $\twice$ does not admit a dependency chain.
As an example of a system which does admit one, consider the AFS with
the following three rules:
\[
\mathtt{f}(\nul) \arrz \mathtt{g}(\abs{x}{\mathtt{f}(x)},\mathtt{a})
\ \ \ \ \ \ \ 
\mathtt{g}(F,\mathtt{b}) \arrz \app{F}{\nul}
\ \ \ \ \ \ \ 
\mathtt{a} \arrz \mathtt{b} \\
\]
This system has four dependency pairs:
\[
\up{\mathtt{f}}(\nul) \dppijl \up{\mathtt{g}}(\abs{x}{\mathtt{f}(x)},
  \mathtt{a})\ \ \ \ \ \ \ 
\up{\mathtt{f}}(\nul) \dppijl \up{\mathtt{f}}(\c_\nat)\ \ \ \ \ \ \ 
\up{\mathtt{f}}(\nul) \dppijl \up{\mathtt{a}}\ \ \ \ \ \ \
\up{\mathtt{g}}(F,\mathtt{b}) \dppijl \app{F}{\nul}
\]
The rules admit an infinite reduction: $\mathtt{f}(\nul) \arrz
\mathtt{g}(\abs{x}{\mathtt{f}(x)},\mathtt{a}) \arrz
\mathtt{g}(\abs{x}{\mathtt{f}(x)},\mathtt{b}) \arrz
\app{(\abs{x}{\mathtt{f}(x)})}{\nul} \arrz_\beta
\mathtt{f}(\nul) \arrz \ldots$; following the steps in the proof of
Theorem~\ref{thm:dependencychain} (starting with $\mathtt{f}(\nul)$)
we obtain the following dependency chain:
\[
\begin{array}{rlclcll}
( &
\up{\mathtt{f}}(\nul) \dppijl \up{\mathtt{g}}(\abs{x}{\mathtt{f}(x)},
\mathtt{a}) & , & \up{\mathtt{f}}(\nul) & , &
\up{\mathtt{g}}(\abs{x}{\mathtt{f}(x)},\mathtt{a}) & ), \\
( & \up{\mathtt{g}}(F,\mathtt{b}) \dppijl \app{F}{\nul} & , &
\up{\mathtt{g}}(\abs{x}{\mathtt{f}(x)},\mathtt{b}) & , &
\app{(\abs{x}{\mathtt{f}(x)},\mathtt{a})}{\nul} & ), \\
( & \mathtt{beta} & , &
\app{(\abs{x}{\mathtt{f}(x)},\mathtt{a})}{\nul} & , &
\up{\mathtt{f}}(\nul) & ), \\
( & \up{\mathtt{f}}(\nul) \dppijl \up{\mathtt{g}}(\abs{x}{
\mathtt{f}(x)},\mathtt{a}) & , & \up{\mathtt{f}}(\nul) & , &
\up{\mathtt{g}}(\abs{x}{\mathtt{f}(x)},\mathtt{a}) & ), \\
& \ldots \\
\end{array}
\]
Note that between the first and second step, a $\arr{in}$ step
is done to reduce $\mathtt{a}$ to $\mathtt{b}$.  Also note that in
the third triple we use case~\ref{chain:subterm} from
Definition~\ref{def:dependencychain}, with $\eterm = \mathtt{f}(x)$.
\end{example}

\noindent
The converse of Theorem~\ref{thm:dependencychain} does not hold.
Consider for instance the AFS with symbols $\mathtt{A} : [\nat \times
\nat] \decpijl \nat$ and $\mathtt{B} : [\nat \typepijl \nat
\typepijl \nat] \decpijl \nat$, and a single rule:
$\mathtt{A}(\avar,\avar) \arrz \mathtt{B}(\abs{\bvar \cvar}{
\mathtt{A}(\bvar,\cvar)})$.
This (terminating!) AFS has a dependency pair
$\up{\mathtt{A}}(\avar,\avar) \dppijl \up{\mathtt{A}}(\c_\nat,\c_\nat)$,
which gives a dependency chain
$\up{\mathtt{A}}(\c_\nat,\c_\nat) \dppijl
\up{\mathtt{A}}(\c_\nat,\c_\nat) \dppijl \ldots$

We could try solving this problem by slightly altering the definition
of closed candidate terms: instead of substituting a variable $\avar :
\atype$ by a symbol $\c_\atype$, we could have replaced it with a
symbol $\c_\avar$, substituting all bound variables with different
symbols.  This choice was made in for example the first definition of
dependency pairs for HRSs~\cite{sak:wat:sak:01}.
But even with this change,
Theorem~\ref{thm:dependencychain} does not give an equivalence.
Consider for instance the AFS with the following rules:
\[
\begin{array}{rcl}
\symb{f}(\avar,\bvar,\suc(\cvar)) & \arrz & \symb{g}(\symb{h}(\avar,
  \bvar),\abs{\dvar}{\symb{f}(\dvar,\avar,\cvar)}) \\
\symb{h}(\avar,\avar) & \arrz &
  \symb{f}(\avar,\suc(\avar), \suc(\suc(\avar))) \\
\end{array}
\]
This system has three dependency pairs:
\[
\begin{array}{rcl}
\up{\symb{f}}(\avar,\bvar,\suc(\cvar)) 
  & \dppijl & \up{\symb{h}}(\avar,\bvar) \\
\up{\symb{f}}(\avar,\bvar,\suc(\cvar)) 
  & \dppijl & \up{\symb{f}}(\c_\dvar,\avar,\cvar) \\
\up{\symb{h}}(\avar,\avar) 
  & \dppijl & \up{\symb{f}}(\avar,\suc(\avar),\suc (\suc(\avar))) \\
\end{array}
\]
We get the following dependency chain:
$\up{\symb{f}} (\c_\dvar, \suc(\c_\dvar), \suc(\suc(\c_\dvar)) ) \dppijl
\up{\symb{f}} (\c_\dvar, \c_\dvar, \suc(\c_\dvar) ) \dppijl
\up{\symb{h}} (\c_\dvar , \c_\dvar) $\\$ \dppijl 
\up{\symb{f}} (\c_\dvar, \suc(\c_\dvar), \suc(\suc(\c_\dvar)) ) \dppijl
\ldots$
However, the AFS is terminating, intuitively because the 
bound variable destroys matching possibilities with the
non-left-linear rule.

For this reason, we have chosen to use the more elegant method with
symbols $\c_\atype$ instead of the slightly more powerful, but also
a fair bit more cumbersome, $\c_\avar$.  The latter style is less
pleasant because of $\alpha$-conversion: for example,
$\candidatesof{\symb{f}(\abs{\avar}{\symb{g}(\avar)})}$ should contain
$\symb{g}(\c_\bvar)$ for all variables $\bvar$.
Thus, to preserve correctness of definitions and proofs, we would
have to jump through a few hoops.  However, all results in this paper
also go through with such a definition; this was for instance
explored in the shorter version of this paper~\cite{kop:raa:11:1}.

\emptyline
The crucial point of both examples above is the combination of bound
variables and non-left-linear rules.
However, for left-linear AFSs, no such counterexample exists.
Intuitively, this holds because replacing variables by a symbol
$\c_\atype$ that does not occur in any left-hand side does not affect
applicability of any rule.  Thus, a dependency chain effectively
produces an infinite reduction $|s_i| \arr{\Rules} \cdot \suptermeq
|t_i| \arrr{\Rules} |s_{i+1}|$ (where $|\cterm|$ replaces any
$\up{\afun}$ in a term $\cterm$ by its unmarked counterpart), and
this implies the existence of an infinite $\arr{\Rules}$ reduction.

\begin{theorem}\label{thm:dependencychainll}
A left-linear AFS $\Rules$ is terminating 
if and only if
it does not admit a (minimal) dependency chain.
\end{theorem}

\begin{proof}
Theorem~\ref{thm:dependencychain} gives one direction.
For the other direction, 
assume a left-linear AFS $\Rules$
and suppose we have an infinite dependency chain (minimal or not).
We construct  an infinite $\arr{\Rules}\cdot \suptermeq$ sequence,
following roughly the intuition above.
We note:
\begin{enumerate}[(1)]
\item \label{ll:linearsubstitution}
  If $l$ is a linear term not containing any symbols
  $\c_\atype$, and $\asub$ is a substitution whose domain contains only
  variables in $l$, and if $\subst{l}{\asub} =
  \aterm[\vec{\avar}:=\vec{\c}]$ for some term $\aterm$ and set of
  variables $\{\vec{\avar}\}$, then there is a substitution
  $\bsub$ such that $\subst{l}{\bsub} = \aterm$ and
  $\asub = \bsub[\vec{\avar}:=\vec{\c}]$.
\item \label{ll:cstep}
  If $\aterm[\vec{\avar}:=\vec{\c}] \arr{\Rules} \bterm$, then
  there exists some $\bterm'$ such that $\aterm \arr{\Rules}
  \bterm'$ and $\bterm'[\vec{\avar}:=\vec{\c}] = \bterm$.
\item \label{ll:csubterm}
  If $\aterm[\vec{\avar}:=\vec{\c}] = C[\bterm]$, then there are $C'$
  and $\bterm'$ such that $\aterm = C'[\bterm']$ and
  $\bterm'[\vec{\avar}:=\vec{\c}] = \bterm$.
\end{enumerate}
(\ref{ll:linearsubstitution}) states that, if a linear term $l$
(typically the left-hand side of a rule) matches a term $\aterm$ with
some $\c$-symbols in it, it also matches $\aterm$ with those symbols
replaced by variables.

This holds by induction on $l$, assuming linearity over
$\domain(\asub)$: the cases where $l$ is a variable are
straightforward (if $l \in \domain(\asub)$ take $\bsub = [l:=\aterm]$,
otherwise let $\bsub := \emptyset$), if $l$ is an abstraction
$\abs{\avar}{l'}$ the induction hypothesis suffices ($\avar$ cannot
occur in domain or range of $\bsub$, for then it would also hold for
$\asub$), and if $l$ is an application or functional term we
use the linearity.  For example the functional case, if $l = f(l_1,
\ldots,l_n)$, then let $\asub_i$ be the restriction of $\gamma$ to
$\FV(l_i)$ for $1 \leq i \leq n$; by the induction hypothesis we find
suitable $\bsub_i$, and by linearity of $l$ each of those $l_i$ has
different variables, so $\bsub := \bsub_1 \cup \ldots \cup \bsub_n$
is well-defined.

(\ref{ll:cstep}) states that, if a term with some variables replaced
by $\c$-symbols reduces, then the original term reduces in a similar
way.  This holds by induction on the size of $\aterm$.
When the reduction is done in a subterm, the statement follows easily
with the induction hypothesis (immediate subterms of
$\aterm[\vec{\avar}:=
\vec{\c}]$ have the form $\aterm'[\vec{\avar}:=\vec{\c}]$ with
$\aterm \supterm \aterm'$, so the induction hypothesis is applicable).
In the base case, a $\beta$-step is easy, and if, for some rule $l
\arrz r$ and substitution
$\asub$, the term $\subst{\aterm}{[\vec{\avar}:=\vec{\c}]} =
\subst{l}{\asub}$, then by left-linearity of $\Rules$ we may use
(\ref{ll:linearsubstitution}):
there is a substitution $\bsub$ such that $\aterm = \subst{l}{\bsub}
\arr{\Rules} \subst{r}{\bsub} =: \bterm'$; certainly
$\subst{\subst{r}{\bsub}}{[\vec{\avar}:=\vec{\c}]} = \subst{r}{\asub}
= \bterm$ as required.

(\ref{ll:csubterm}) follows by induction on the size of C: if
C is the empty context take $\bterm' := \aterm$, otherwise use the
induction hypothesis; for instance if $C[] = \afun(\cterm_1,\ldots,
D_i[],\ldots,\cterm_n)$, then $\aterm = \afun(\aterm_1,
\ldots,\aterm_i,\ldots,\aterm_n)$ (with each $\aterm_j[\vec{\avar}:=
\vec{\c}] = \cterm_j$), and by the induction hypothesis on $\aterm_j$
there are $D_i',\bterm'$ such that $C' := \afun(\aterm_1,\ldots,D_i',
\ldots,\aterm_n)$ and $\bterm'$ satisfy the requirement.

\medskip \noindent
Now suppose there is a dependency chain $\rijtje{(\rho_i,
\aterm_i,\bterm_i) \mid i \in \N}$, and define $\aterm_0' :=
|\aterm_0|$ (that is, $\aterm_0$ with all marks removed).
For all $i \in \N$, suppose $\aterm_i'[\vec{\avar_0}:=\vec{\c},\ldots,
\vec{\avar_{i-1}}:=\vec{\c}] = |\aterm_i|$.
Whether $\rho_i$ is $\mathtt{beta}$ or a dependency pair,
$|\aterm_i| \arr{\Rules} C_i[\cterm_i]$ for some
term $\cterm_i$ and context $C_i$,
such that $|\bterm_i| = \cterm_i[\vec{\avar_i}:=
\vec{\c}]$ for some
variables $\vec{\avar_i}$.  By (\ref{ll:cstep}), (\ref{ll:csubterm})
also $\aterm_i' \arr{\Rules} C_i'[\cterm_i']$ and
$\cterm_i'[\vec{\avar_0}:=\vec{\c},\ldots,\vec{\avar_i}:=\vec{\c}] =
|\bterm_i|$.
By (\ref{ll:cstep}) we can find $\aterm_{i+1}'$ such that
$\cterm_i' \arrr{\Rules} \aterm_{i+1}'$ and $\aterm_{i+1}'
[\vec{\avar_0}:=\vec{\c},\ldots,\vec{\avar_i}:=\vec{\c}] =
|\aterm_{i+1}|$.
Thus, $\aterm_0'$ is non-terminating: $\aterm_0' \arrp{\Rules}
C_0'[\aterm_1'] \arrp{\Rules} C_0'[C_1'[\aterm_2']] \arrp{\Rules}
\ldots$
\end{proof}

\subsection{The Dependency Graph}\label{subsec:graph}

As in the first-order case, we use a dependency graph 
to organise the dependency pairs. 
The notions are very similar to the first-order definitions.

The \emph{dependency graph} of an AFS $\Rules$ is a graph with the
dependency pairs of $\Rules$ as nodes, and an edge from node
$l \dppijl p$ to node $l' \dppijl p'$ if either $\head(p)$ is a
variable, or there are substitutions $\asub$ and $\bsub$ such that
$\subst{p}{\asub} \arrr{\Rules,in} \subst{l'}{\bsub}$.

\begin{example}\label{ex:twicegraph}
The dependency graph of the AFS $\twice$:

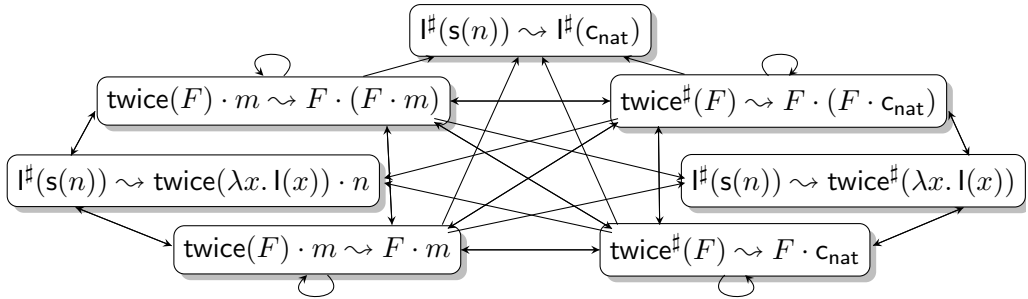
\vspace{11pt}

\begin{tikzpicture}[->]
\begin{scope}[>=stealth]

\tikzstyle{veld} = [draw, fill=white,  drop shadow, 
  minimum height=0em, minimum width=0em, rounded corners]

\node (centre) { };

\node [veld] (ui1) [left of=centre,anchor=east,node distance=20mm] {
    $\up{\I}(\suc(n)) \dppijl \app{\twice(\abs{x}{\I(x)})}{n}$
  };
\node (ui1a) [left of=ui1,node distance=17mm] { };
\node (ui1b) [above of=ui1a,node distance=2mm] { };
\node (ui1c) [below of=ui1a,node distance=3mm] { };
\node (ui1d) [right of=ui1,node distance=24mm] { };

\node [veld] (ui2) [right of=centre,anchor=west,node distance=20mm] {
    $\up{\I}(\suc(n)) \dppijl \up{\twice}(\abs{x}{\I(x)})$
  };
\node (ui2a) [right of=ui2,node distance=16mm] { };
\node (ui2b) [above of=ui2a,node distance=2mm] { };
\node (ui2c) [below of=ui2a,node distance=3mm] { };
\node (ui2d) [left of=ui2,node distance =21mm] { };

\node [veld] (ui3) [above of=centre,node distance=20mm] {
    $\up{\I}(\suc(n)) \dppijl \up{\I}(\c_\nat)$
  };

\node [veld] (twice1) [above left of=centre,anchor=east,node distance=15mm] {
    $\app{\twice(F)}{m} \dppijl \app{F}{(\app{F}{m})}$
  };
\node (twice1a) [below of=twice1,node distance=2mm] { };
\node (twice1b) [right of=twice1a,node distance=15mm] { };
\node (twice1c) [right of=twice1b,node distance=5mm] { };
\node (twice1d) [left of=twice1,node distance=23mm] { };
\node (twice1e) [above of=twice1,node distance=2mm] { };

\node [veld] (twice2) [below left of=centre,anchor=east,node distance=13mm] {
    $\app{\twice(F)}{m} \dppijl \app{F}{m}$
  };
\node (twice2a) [above of=twice2,node distance=2mm] { };
\node (twice2b) [right of=twice2a,node distance=10mm] { };
\node (twice2c) [right of=twice2b,node distance=6mm] { };
\node (twice2d) [left of=twice2,node distance=18mm] { };
\node (twice2e) [below of=twice2,node distance=2mm] { };

\node [veld] (utwice1) [above right of=centre,anchor=west,node distance=15mm] {
    $\up{\twice}(F) \dppijl \app{F}{(\app{F}{\c_\nat})}$
  };
\node (utwice1a) [below of=utwice1,node distance=2mm] { };
\node (utwice1b) [left of=utwice1a,node distance=16mm] { };
\node (utwice1c) [left of=utwice1b,node distance=4mm] { };
\node (utwice1d) [right of=utwice1,node distance=22mm] { };
\node (utwice1e) [above of=utwice1,node distance=2mm] { };

\node [veld] (utwice2) [below right of=centre,anchor=west,node distance=13mm] {
    $\up{\twice}(F) \dppijl \app{F}{\c_\nat}$
  };
\node (utwice2a) [above of=utwice2,node distance=2mm] { };
\node (utwice2b) [left of=utwice2a,node distance=10mm] { };
\node (utwice2c) [left of=utwice2b,node distance=5mm] { };
\node (utwice2d) [right of=utwice2,node distance=17mm] { };
\node (utwice2e) [below of=utwice2,node distance=2mm] { };

\draw (twice1b) -- (twice2b);
\draw (twice2b) -- (twice1b);
\draw (twice1) -- (utwice1);
\draw (utwice1) -- (twice1);
\draw (twice2) -- (utwice2);
\draw (utwice2) -- (twice2);
\draw (utwice1b) -- (utwice2b);
\draw (utwice2b) -- (utwice1b);
\draw (twice1c) -- (utwice2c);
\draw (utwice2c) -- (twice1c);
\draw (utwice1c) -- (twice2c);
\draw (twice2c) -- (utwice1c);

\draw (ui1b) -- (twice1d);
\draw (twice1d) -- (ui1b);
\draw (ui1c) -- (twice2d);
\draw (twice2d) -- (ui1c);
\draw (utwice1c) -- (ui1d);
\draw (utwice2c) -- (ui1d);
\draw (ui2b) -- (utwice1d);
\draw (utwice1d) -- (ui2b);
\draw (ui2c) -- (utwice2d);
\draw (utwice2d) -- (ui2c);
\draw (twice1c) -- (ui2d);
\draw (twice2c) -- (ui2d);
\draw (twice2c) -- (ui3);
\draw (utwice2c) -- (ui3);
\draw (twice1) -- (ui3);
\draw (utwice1) -- (ui3);

\draw (twice1e) to[out=45,in=135,looseness=5] (twice1e);
\draw (utwice1e) to[out=135,in=45,looseness=5] (utwice1e);
\draw (twice2e) to[out=315,in=225,looseness=5] (twice2e);
\draw (utwice2e) to[out=225,in=315,looseness=5] (utwice2e);

\end{scope}
\end{tikzpicture}
\end{example}

\noindent
A \emph{cycle} is a set $\acycle$ of dependency pairs such that
between every two pairs $\rho,\pi \in \acycle$ there is a non-empty
path in the graph using only nodes in $\acycle$.
A cycle that is not contained in any other cycle is called a
\emph{strongly connected component} (SCC).  
To prove termination we must show that cycles in a dependency graph
are ``chain-free'' (see Theorem \ref{thm:maintheorem}).
The requirement to add an edge from any node of the form $l \dppijl
\app{\avar}{r_1} \cdots r_n$ (with $\avar$ a variable) to all other
nodes is necessary by clause~\ref{chain:subterm} in
Definition~\ref{def:dependencychain}: a dependency chain could have
a dependency pair of the form $l \dppijl \app{\avar}{\vec{r}}$
followed by $\mathtt{beta}$, and then any other dependency pair.
Hence a rule with leading free variables in the right-hand
side gives rise to many cycles.

A set of dependency pairs $\P$ is called \emph{chain-free} if there
is no minimal dependency chain using only dependency pairs in $\P
\cup \{\cbeta\}$.

\begin{lemma}\label{lem:emptynondangerous}
$\emptyset$ is chain-free.
\end{lemma}
 
\begin{proof}
Given a dependency chain with all $\rho_i = \mathtt{beta}$,
each $\aterm_i \arr{\beta} C_i[\aterm_{i+1}]$ for some context $C_i$,
contradicting termination of the simply-typed $\lambda$-calculus.
\end{proof}

Because the dependency graph cannot be computed in general, it is
common to use \emph{approximations} of the dependency graph,
which have the same nodes but possibly more edges.
A brute method to find an approximation
is to have an edge between $l \dppijl p$ and $l' \dppijl p'$ 
if either the head of $p$ is a variable, 
or if $p$ and $l'$ both have the form 
$\app{\afun(\aterm_1,\ldots,\aterm_n)}{\aterm_{n+1}} \cdots \aterm_m$
for the same function symbol $\afun$.
It is interesting to study more sophisticated methods 
to find approximations, but this is left for future work. 

As stated in Section~\ref{sec:preliminaries}, we assume a finite set
of rules, which leads to a finite set of dependency pairs.  In
Section~\ref{subsec:polymorphism} we will say a few words on
extending the technique to systems with infinitely many rules
(without having to deal with an infinite graph).

\begin{lemma}\label{lem:cyclenondangerous}
Let $G$ be an approximation of the dependency graph of an AFS $\Rules$.
Suppose that every SCC in $G$ is chain-free.
Then $\Rules$ is terminating.
\end{lemma}

\begin{proof}
Since $\DP$ is finite, any (minimal) dependency chain $\rijtje{
(\rho_i,\aterm_i,\bterm_i) \mid i \in \N}$ has at least one
dependency pair $\rho_i$ which occurs infinitely often.  
Note that if $n < m$ then there is a path in $G$ from $\rho_n$ to
$\rho_m$ (if $\rho_n,\rho_m \neq \mathtt{beta}$).
Therefore,  there is a path in
$G$ from $\rho_i$ to itself, and hence $\rho_i$ is on a cycle.
Let $\acycle$ be the SCC containing $\rho_i$.
Then all $\rho_j$ with $j > i$ and $\rho_j \neq \mathtt{beta}$ are in
$\acycle$: from each such $\rho_j$ there is a path to $\rho_i$ and
back.
But then, $\{ \rho_j \mid j \geq i \}$ is a minimal dependency chain
in $\acycle$, so $\acycle$ is not chain-free, contradicting the
assumption!
Thus, there is no dependency chain, and therefore by
Theorem~\ref{thm:dependencychain}, $\Rules$ is terminating.
\end{proof}

\begin{example}\label{ex:twicecycles}
The dependency graph (approximation) of $\twice$ from Example
\ref{ex:twicegraph} has only one SCC:
\[
\left\{
\begin{array}{rclrcl}
\up{\I}(\suc(n)) & \dppijl & \twice(\abs{x}{\I(x)}) \cdot n &
\up{\twice}(F) & \dppijl & F \cdot (F \cdot \c_\nat) \\
\up{\I}(\suc(n)) & \dppijl & \up{\twice}(\abs{x}{\I(x)}) &
\up{\twice}(F) & \dppijl & F \cdot \c_\nat \\
\twice(F) \cdot m & \dppijl & F \cdot (F \cdot m) &
\twice(F) \cdot m & \dppijl & F \cdot m \\
\end{array}
\right\}
\]
Therefore $\twice$ is terminating if this set, which we shall call
$\scctwice$, is chain-free.
\vspace{-6pt}
\end{example}

\subsection{Reduction Triples}\label{subsec:redord}

The challenge, then, is to prove that given sets of dependency pairs
are chain-free.
We use the following definition:

\begin{definition}\label{def:reductionpair}
A \emph{reduction triple} consists of a quasi-ordering $\geqterm$,
a sub-relation $\geqterm_1$ of $\geqterm$, and a well-founded
ordering $\gterm$, all defined on terms built over $\up{\F}_c$, such
that:
\begin{enumerate}[(1)]
\item $\geqterm$ and $\gterm$ are \emph{compatible}: 
  either $\gterm \cdot \geqterm \mathord{\subseteq} \gterm$
  or $\geqterm \cdot \gterm \mathord{\subseteq} \gterm$;
\item $\geqterm,\ \geqterm_1$ and $\gterm$ are all \emph{stable}
  (closed under substitution);
\item $\geqterm_1$ is \emph{monotonic}: 
  (if $\aterm \geqterm_1 \bterm$ and $\aterm,\bterm$ share a type,
  then $C[\aterm] \geqterm_1 C[\bterm]$ for all $C[]$);
\item $\geqterm_1$ contains \texttt{beta} 
  (always $\app{(\abs{\avar}{\aterm})}{\bterm} \geqterm_1
  \aterm[\avar:=\bterm]$).
\end{enumerate}
\end{definition}

\noindent
A \emph{reduction pair} is a pair $(\geqterm,\gterm)$ such that
$(\geqterm,\geqterm,\gterm)$ is a reduction triple; this corresponds
to the first-order notion of a reduction pair.  The reduction
triple is a generalisation of this notion, where $\geqterm$ itself is
not required to be monotonic; we will need a non-monotonic $\geqterm$
in Section \ref{subsec:typechange} to compare terms with different
types.  This notion of a reduction triple is similar to the one which
appears in \cite{hir:mid:07:1}.

To deal with subterm reduction in dependency chains, 
an additional definition is needed.
\begin{definition}[Limited Subterm Property]
$\geqterm$ has the \emph{limited subterm property} if the following
requirement is satisfied:
\emph{
for all variables $\avar$ and terms $\aterm,\bterm,\cterm$ such that
$\aterm \suptermeq \cterm \supterm \avar$, there is a substitution
$\asub$ such that $\app{(\abs{\avar}{\aterm})}{\bterm} \geqterm
\up{\cterm}[\avar:=\bterm]\asub$.}
\end{definition}
Intuitively, the substitution $\asub$ can be used to replace free
variables in $\cterm$ which are bound in $\aterm$ by the corresponding
constants $\c_\atype$.
However, we will also use a more liberal replacement of those
variables, hence the general $\asub$.

The following theorem shows how reduction triples are used with
dependency pairs.

\begin{theorem}\label{thm:maintheorem}
A set $\P = \P_1 \uplus \P_2 $ of dependency pairs is chain-free if
$\P_2$ is chain-free, and there is a reduction triple $(\geqterm,
\geqterm_1,\gterm)$ such that:
\begin{iteMize}{$\bullet$}
  \item
  $l \gterm p$ for all $l \dppijl p \in \P_1$,
  \item
  $l \geqterm p$ for all $l \dppijl p \in \P_2$, 
  \item
  $l \geqterm_1 r$ for all $l \arrz r \in \Rules$,
  \item
  either $\P$ is non-collapsing 
  or $\geqterm$ has the limited subterm property.
\end{iteMize}
\end{theorem}

\noindent
Here, a set $\P$ of dependency pairs is called
\emph{non-collapsing} if all elements of $\P$ are
non-collapsing.
Symmetrically, $\P$ is \emph{collapsing} if it contains at least one
collapsing pair $\up{l} \arrz \app{\avar}{\vec{p}}$.

\begin{proof}
Towards a contradiction, suppose there is such a reduction triple
$(\geqterm,\geqterm_1,\gterm)$, but $\P$ admits a minimal dependency
chain;
since
$\P_2$ is chain-free, infinitely many $\rho_i$ are in $\P_1$.

If $\P$ is non-collapsing, then the chain may start with some
$\mathtt{beta}$ steps, but once some $\rho_i \in \DP$, all
$\rho_j$ with $j > i$ must also be in $\DP$, because the head of each
$\bterm_i$ is a functional term, rather than an abstraction.  Thus,
for each $j$ either $\aterm_j \geqterm \bterm_j \geqterm
\aterm_{j+1}$, or (if $\rho_j \in \P_1$) even $\aterm_j \gterm
\bterm_j \geqterm \aterm_{j+1}$, contradicting well-foundedness of
$\gterm$ (the latter happens infinitely often).

Alternatively, suppose $\P$ is collapsing, and $\geqterm$ has
the limited subterm property.
Let 
$\rijtje{(\rho_i,\aterm_i,\bterm_i) 
\mid i \geq j}$
be a dependency chain over $\P$;
if $\rho_j \in \P_1$ then $\aterm_j \gterm \bterm_j
\geqterm \aterm_{j+1}$, if $\rho_j \in \P_2$ then $\aterm_j \geqterm
\bterm_j \geqterm \aterm_{j+1}$ and if $\rho_j = \cbeta$ then (by the
limited subterm property) there is a substitution $\bsub$ such that
$\aterm_j \geqterm \subst{\bterm_j}{\bsub} \geqterm \subst{\aterm_{j+
1}}{\bsub}$.  Since $\rijtje{(\rho_i,\subst{\aterm_i}{\bsub},
\subst{\bterm_i}{\bsub}) 
\mid i \geq j+1}$ is also a
dependency chain we can continue this reasoning recursively.  We
obtain a decreasing $\geqterm$ sequence with infinitely many $\gterm$
steps, which contradicts well-foundedness of $\gterm$.
\end{proof}
Theorem~\ref{thm:maintheorem} can be used to prove that every SCC
in the dependency graph approximation of an AFS is chain-free;
termination follows with Lemma \ref{lem:cyclenondangerous}.
In Section~\ref{sec:algorithm} we will give an algorithm similar to
the algorithm in Theorem~\ref{thm:fo:alg}.

\begin{example}\label{ex:twiceconclude}
Termination of $\twice$ is proved if there is a reduction triple
$(\geqterm,\geqterm_1,\gterm)$ with the limited subterm property, 
such that $l \geqterm_1 r$ for all rules, 
and $l \gterm p$ for every dependency pair in $\scctwice$ from
Example~\ref{ex:twicecycles} (choosing $\P_2 = \emptyset$, which is
chain-free).
\end{example}

For left-linear AFSs, where the existence of a minimal dependency
chain characterises termination by Theorem~\ref{thm:dependencychain},
a terminating AFS always has a suitable reduction pair.

\begin{theorem}\label{thm:maintheoremll}
A left-linear AFS is terminating if and only if there is a reduction
triple $(\geqterm,\geqterm_1,\gterm)$ such that 
$l \gterm p$ for every $l \dppijl p \in \DP$,
and $l \geqterm_1 r$ for every $l \arrz r \in \Rules$,
and $\geqterm$ has the limited subterm property.
\end{theorem}

\begin{proof}
By Lemma~\ref{lem:cyclenondangerous} and
Theorem~\ref{thm:maintheorem}, termination of $\Rules$ follows if
such a reduction triple exists.
For the other direction, let $\aterm \geqterm \bterm$ if $|\aterm|
\arrr{\Rules} |\bterm|$,
and let $\aterm \gterm
\bterm$ if $|\aterm|\ (\arr{\Rules} \cdot \csuptermeq)^+\ |\bterm|$,
where $\csuptermeq$ is the (reflexive) subterm relation where bound
variables which become free are replaced with symbols $\c_\atype$,
and $|\cterm|$ removes marks from $\cterm$.
It is evident that $(\geqterm,\geqterm,\gterm)$ is a reduction
triple, that $l \gterm p$ for all dependency pairs $l \dppijl p$ and
$l \geqterm r$ for all rules $l \arrz r$.  Moreover, $\geqterm$
has the limited subterm property with $\gamma$ the substitution
$[\vec{\bvar}:=\vec{\c}]$.
\end{proof}

\subsection{Type Changing}\label{subsec:typechange}

The situation so far is not completely satisfactory, because both
$\geqterm$ and $\gterm$ may have to compare terms of different types.
Consider for example the dependency pair $\up{\twice}(F) \dppijl
F \cdot \c_\nat$, where the left-hand side has a functional type and the
right-hand side does not.
Moreover, the comparison in the definition of limited subterm
property may concern terms of different types.
This is problematic because term orderings do not usually
relate terms of arbitrary different types; 
neither any version of the higher-order path ordering
\cite{jou:rub:99:1,bla:jou:rub:08:1} nor monotonic algebras
\cite{pol:96:1} are equipped to do this.

A solution is to manipulate the ordering requirements.  
Let $(\geqterm,\gterm)$ be a reduction pair (so a pair such
that $(\geqterm,\geqterm,\gterm)$ is a reduction triple).
Define $\geq$,\ $\geq_1$ and $>$ as follows:
\begin{iteMize}{$\bullet$}
\item $\aterm > \bterm$ if there are fresh variables $\avar_1,\ldots,
  \avar_n$ and terms $\cterm_1,\ldots,\cterm_m$ such that
  $\app{\aterm}{\avar_1} \cdots \avar_n \gterm \app{\bterm}{\cterm_1}
  \cdots \cterm_m$ and both sides have some base type;
\item $\aterm \geq \bterm$ if there are fresh variables $\avar_1,
\ldots, \avar_n$ and terms $\cterm_1,\ldots,\cterm_m$ such that
$\app{\aterm}{\avar_1} \cdots \avar_n\ R\ \app{\bterm}{\cterm_1}
  \cdots \cterm_m$ and both sides have some base type; here $R$ is
  the union of $\geqterm,\ \gterm \cdot \geqterm$ and $\geqterm \cdot
  \gterm$;
\item $\aterm \geq_1 \bterm$ if $\aterm \geqterm \bterm$ and
  $\aterm$ and $\bterm$ have the same type.
\end{iteMize}

\begin{lemma}\label{lem:triplepair}
$(\geq,\geq_1,>)$ as generated from a reduction pair 
$(\geqterm,\gterm)$ is a reduction triple.
\end{lemma}

\begin{proof}
We make the following observations:
\begin{enumerate}[(1)]
\item \label{triplepair:geq1}
  if $\aterm \geq_1 \bterm$ then by monotonicity
  $\aterm{\vec{\avar}} \geqterm \bterm{\vec{\avar}}$;
\item \label{triplepair:g}
  if $\aterm > \bterm$ then for any $\vec{\cterm}$ there are
  $\vec{\dterm}$ such that $\app{\aterm}{\vec{\cterm}}
  \gterm \app{\bterm}{\vec{\dterm}}$ (by stability of $\gterm$);
\item \label{triplepair:geq}
  if $\aterm \geq \bterm$ then for any $\vec{\cterm}$ there are
  $\vec{\dterm}$ such that either $\app{\aterm}{\vec{\cterm}}
  \geqterm \app{\bterm}{\vec{\dterm}}$ or $\app{\aterm}{\vec{\cterm}}
  \gterm \cdot \geqterm \app{\bterm}{\vec{\dterm}}$ or
  $\app{\aterm}{\vec{\cterm}} \geqterm \cdot \gterm
  \app{\bterm}{\vec{\dterm}}$ (by stability of both $\gterm$ and
  $\geqterm$).
\end{enumerate}
Each of the required properties on $\geq$, $\geq_1$, and $>$ now
follows easily from the properties on $\geqterm$ and $\gterm$.  For
example transitivity of $>$: if $\aterm > \bterm > \cterm$, then
there are terms $\vec{\dterm}$ such that $\app{\aterm}{\vec{\avar}}
\gterm \app{\bterm}{\vec{\dterm}}$, and by (\ref{triplepair:g}) there
are terms $\vec{\eterm}$ such that $\app{\bterm}{\vec{\dterm}} \gterm
\app{\cterm}{\vec{\eterm}}$; by transitivity of $\gterm$, therefore,
$\app{\aterm}{\vec{\avar}} \gterm \app{\cterm}{\vec{\eterm}}$, so
$\aterm > \cterm$.
Well-foundedness of $>$ follows from (\ref{triplepair:g}) and
well-foundedness of $\gterm$.  For stability, note that if $\aterm >
\bterm$ and $\asub$ is a substitution, then for fresh variables
$\vec{\avar}$ (which do not occur in domain or range of $\asub$) also
$\app{\aterm}{\vec{\avar}} \gterm \app{\bterm}{\vec{\dterm}}$, so
$\app{(\subst{\aterm}{\gamma})}{\vec{\avar}} =
\subst{(\app{\aterm}{\vec{\avar}})}{\gamma} \gterm
\subst{(\app{\bterm}{\vec{\dterm}})}{\gamma} =
\app{(\subst{\bterm}{\gamma})}{(\vec{\dterm}\gamma)}$; stability of
$\geq$ is similar.
$\geq_1$ is included in $\geq$ by (\ref{triplepair:geq1}), and
contains \texttt{beta} because $\geqterm$ does.
For compatibility, and for transitivity of $\geq$, we use a case
distinction on which form of $\geq$ is used, and transitivity of both
$\gterm$ and $\geqterm$, as well as compatibility between the two.
\end{proof}

The relations $\geq$ and $>$ are not necessarily computable, but
they do not need to be: we will only use specific instances.
To prove some set of dependency pairs $\P$ chain-free, we can choose
for every pair $l \dppijl p \in \P$ a corresponding base-type pair
$\overline{l} \dppijl \overline{p}$, and prove either
$\overline{l} \geqterm \overline{p}$ or $\overline{l} \gterm
\overline{p}$.  For example, we could assign
$\overline{l} := \app{l}{\avar_1} \cdots \avar_n$ and
$\overline{p} := \app{p}{\c_{\atype_1}} \cdots \c_{\atype_m}$.  This
is the choice we will use in examples in this paper.
Other choices for $\overline{p}$
are also possible.

\medskip
\emph{We assume a systematic way of choosing $\overline{l} \dppijl
\overline{p}$ given $l \dppijl p$.}

\medskip
To make sure that $\geq$ has the limited subterm property, we
consider a base-type version of subterm reduction, which has a strong
relation with $\beta$-reduction.

\begin{definition}
$\supterm^!$ is the relation on base-type terms (and $\suptermeq^!$
its reflexive closure) generated by the following clauses:
\begin{iteMize}{$\bullet$}
\item $\app{(\abs{\avar}{\aterm})}{\bterm_0} \cdots \bterm_n
  \supterm^! \cterm$ if $\app{\aterm[\avar:=\bterm_0]}{\bterm_1}
  \cdots \bterm_n \suptermeq^! \cterm$;
\item $\app{\afun(\aterm_1,\ldots,\aterm_m)}{\bterm_1} \cdots
  \bterm_n \suptermeq^! \cterm$ if $\app{\aterm_i}{\vec{\c}}
  \suptermeq^! \cterm$ for some $i$;
\item $\app{\aterm}{\bterm_1} \cdots \bterm_n \suptermeq^! \cterm$ if
  $\app{\bterm_i}{\vec{\c}} \suptermeq^! \cterm$ for some $i$
  \ \ ($\aterm$ may have any form).
\end{iteMize}
\end{definition}

\noindent
Here, $\app{\aterm}{\vec{\c}}$ is a term $\aterm$
applied to constants $\c_\atype$ of the right types.
We say $(\geqterm,\gterm)$ respects $\supterm^!$ if $\supterm^!$ is
contained in $(\geqterm \cup \gterm)^*$.  Note that, since $\geqterm$
contains $\cbeta$, the first clause is not likely to give problems.
$\supterm^!$ is interesting because if $\aterm \supterm \bterm$ and
$\aterm$ has base type, then there are terms $\cterm_1,\ldots,\cterm_n$
and a substitution $\asub$ on domain $\FV(\bterm) \setminus
\FV(\aterm)$ such that $\aterm \supterm^!
\app{\bterm\asub}{\cterm_1} \cdots \cterm_n$ (this is easy to see
with induction on the size of $\aterm$).  Consequently, $\geq$
has the limited subterm property if $(\geqterm,\gterm)$ respects
$\supterm^!$ and $\afun(\vec{\avar}) \geqterm \up{\afun}(\vec{\avar})$
for all $\afun \in \Defineds$ (the \emph{marking property}).

Using Theorem~\ref{thm:maintheorem} and the reduction triple
generated from a reduction pair, we obtain:

\begin{theorem}\label{thm:typepreserve}
A set of dependency pairs $\P = \P_1 \uplus \P_2$ is chain-free if
$\P_2$ is chain-free and there is a reduction pair $(\geqterm,
\gterm)$ such that:
\begin{enumerate}[\em(1)]
\item $\overline{l} \gterm \overline{p}$ for all $l \dppijl p \in
  \P_1$;
\item $\overline{l} \geqterm \overline{p}$ for all $l \dppijl p \in
  \P_2$;
\item \label{it:tp:rules}
  $l \geqterm r$ for all $l \arrz r \in \Rules$;
\item if $\P$ is collapsing, then $(\geqterm,\gterm)$ respects
  $\supterm^!$, and $\afun(\vec{\avar}) \geqterm \up{\afun}(
  \vec{\avar})$ for all $\afun \in \Defineds$.
\end{enumerate}
\end{theorem}

\noindent
Note that the theorem does not use the generated (and possibly not
computable) triple directly; we prove $\overline{l} \gterm
\overline{p}$ or $\overline{l} \geqterm \overline{p}$ for a specific
choice of $\overline{l}$ and $\overline{p}$.  The generated triple is
merely used in the reasoning that justifies
Theorem~\ref{thm:typepreserve}.

\begin{example}
To prove that $\scctwice$ is chain-free it suffices to find a
reduction pair $(\geqterm,\gterm)$ such that $l \geqterm r$ for all
rules, $(\geqterm,\gterm)$ respects $\supterm^!$ and satisfies the
marking property, and:
\[
\begin{array}{rclrcl}
\up{\I}(\suc(n)) & \gterm & \twice(\abs{x}{\I(x)}) \cdot n &
\app{\up{\twice}(F)}{x} & \gterm & F \cdot (F \cdot \c_\nat) \\
\up{\I}(\suc(n)) & \gterm & \up{\twice}(\abs{x}{\I(x)}) \cdot \c_\nat\ &
\up{\twice}(F) \cdot x & \gterm & F \cdot \c_\nat \\
\twice(F) \cdot  m & \gterm & F \cdot  (F \cdot  m) &
\twice(F) \cdot m & \gterm & F \cdot  m \\
\end{array}
\]
\end{example}

\emptyline
This completes the basis of dynamic dependency pairs for AFSs.

At this point, we might ask: \emph{what have we gained?}  Is it
easier to use Theorem~\ref{thm:typepreserve} than to use a
conventional approach like CPO~\cite{bla:jou:rub:08:1}? 
Can we even find a reduction pair which respects $\supterm^!$?
And if so, couldn't we use the same reduction pair without dependency
pairs?

The answer to these questions will be explored in the coming
sections.
First (Section~\ref{sec:weakdp}), we will consider an extension
limited to
\emph{fully extended, left-linear AFSs}.  With this restriction, we
can weaken the limited subterm property, and obtain a variation of
usable rules.
Next, in Section~\ref{sec:redpair}, we will study two ways to find a
suitable reduction pair: using interpretations in a \emph{weakly
monotonic algebra}, and \emph{argument functions}, a generalisation
of argument filterings.  Finally, in Section~\ref{sec:noncollapse} we
will see additional ways to prove chain-freeness of a set $\P$ if
$\P$ is \emph{non-collapsing}.
All results are combined in the algorithm of
Section~\ref{sec:algorithm}.

\section{Dependency Pairs for Local AFSs}\label{sec:weakdp}

\summary{In this section we consider \emph{local} AFSs, 
and define \emph{formative rules} for local AFSs.  We add tags to symbols
below a $\lambda$, and prove that we only need the limited subterm
property for tagged symbols.  We use this to weaken the requirements
on a reduction pair.
}

The limited subterm property is weaker than the requirements used in
Theorem~\ref{thm:typepreserve}: subterm reduction \emph{only} has to
be done following $\beta$-reduction.  That is, we only need it for
terms which occur below a $\lambda$-abstraction, when a bound
variable is substituted.

To exploit this property, we will pay special attention to
\emph{local} AFSs.  In a local AFS we can (mostly) avoid reducing
terms below an abstraction.
Knowing this, the limited subterm property only requires that
$\app{\afun(\aterm_1,\ldots,\aterm_n)}{\aterm_{n+1}} \cdots \aterm_m
\geqterm \app{\aterm_i}{\vec{\c}}$ for symbols $\afun$ which cannot,
at that time, be reduced anyway.  This makes it
possible to use for instance argument filterings, as we will see in
this section and Section~\ref{sec:argfun}.
As a bonus,
locality
also allows us to
define \emph{formative rules}, a variation of usable rules.  These
are discussed in Section~\ref{subsec:formrules}.

\subsection{Intuition}\label{subsec:weak:intuition}
The idea to tag symbols and rules in order to weaken the limited
subterm property originates in the notion of \emph{weak reductions},
defined in~\cite{cag:hin:98:1} (following a definition from Howard
in 1968).  A weak reduction in the $\lambda$-calculus does not use
steps between a $\lambda$-abstraction and its binder.  This notion
generalises to AFSs in the obvious way.

Consider AFSs where the left-hand sides of all rules are
\emph{linear} (so no free variables occur more than once), and
\emph{free of abstractions} (so the $\lambda$ symbol does not occur
in them).  This limitation is not as strong as it might seem at
first; the $\beta$-reduction ``rule'' is not included in this.
As it turns out, we can prove the following statement:

\emph{\textbf{Claim:} in a left-linear, and left-abstraction-free
AFS, if there is a minimal dependency chain, then there is one
where the reduction $\bterm_i \arrr{\Rules} \aterm_{i+1}$ always
uses only weak steps.}

To see why this matters, let us consider a
colouring of the function symbols.  In a given term $\aterm$, make
all symbol occurrences either red or green: red if the symbol occurs
between an abstraction and its binder, green otherwise.  So if
$\aterm = C[\afun(\bterm_1,\ldots,\bterm_n)]$, make the $\afun$
red
if some $\bterm_i$ contains freely a variable which is bound in
$\aterm$, green if not.
We say $\aterm$ is well-coloured if it
uses this colouring.  Colour the rules in the same way; by the
restrictions, the left-hand sides are entirely green, while the
right-hand sides may contain red symbols.

Now consider a weak
reduction step on a well-coloured term.  If the term is reduced by a
coloured rule, then the result is also well-coloured.  If the term is
reduced with a $\beta$-step, then the result may have some red
symbols outside an abstraction; however, it can become
well-coloured again by painting these red symbols green.  We never
have to paint green symbols red. 
Inventing notation, we can summarise this as follows:

\emph{\textbf{Claim:} if $\aterm \arr{\Rules,\mathit{weak}} \bterm$,
then $\mathit{colour}(\aterm) \arr{\Rules_{\mathit{colour}}} \cdot
\arrr{\mathit{make\_green}} \mathit{colour}(\bterm)$.}

Combining the two claims, we can colour dependency chains.  In the
$\mathtt{beta}$-with-subterm step (\ref{chain:subterm}),
which led to the need for the limited subterm property, we
take a term which was originally below an abstraction, reduce it to a
subterm which still contains the bound variable, and substitute it.
Importantly, the subterm clause $\cterm \suptermeq \dterm$ can be
derived with steps $\abs{\avar}{\aterm} \supterm \aterm,\ 
\app{\aterm_1}{\aterm_2} \supterm \aterm_i$ and $\afun(\aterm_1,
\ldots,\aterm_n) \supterm \aterm_i$, \emph{where the $\afun$ is
always a red symbol}.

Considering the red and green symbols as different symbols altogether
(related only by the $\mathit{make\_green}$ rules) we thus see that
it will not give problems to use an argument filtering, provided we
use it only for the green symbols!

\emptyline
This summarises the ideas which we shall use to simplify the limited
subterm property.  Since colours do not work well in papers, we will
use tags: a red symbol $\afun$ corresponds with a symbol
$\afun^-$, and a green symbol remains unchanged.  Moreover, if we
focus on the colours, and forget about the weak reductions, it turns
out that we do not need to require that the left-hand sides of rules
contain no $\lambda$-abstractions at all: it suffices if the rules
are \emph{local}.

\subsection{Local AFSs}\label{subsec:weak}
Both to weaken the limited subterm property, and for formative rules,
we shall restrict attention to so-called local AFSs where,
intuitively, matching is purely local.  This means that to
apply a rule we do not have to check whether two subterms are equal,
or whether a symbol occurs in a subterm.  Locality combines the
restrictions that the system is \emph{left-linear} and \emph{fully
extended}.  A left-linear, left-abstraction-free AFS is always local.

The locality restriction appears in the literature both for
HRSs~\cite{oos:97:1,bru:08}, where a pattern is called local if it
is fully extended and linear, and for combinatory reduction
systems (CRSs) in~\cite{mel:96}; the latter definition is slightly
different but has a similar underlying intuition.
The definition for AFSs here follows~\cite{oos:97:1,bru:08}, although
the definition of full-extendedness is technically (but not
conceptually) different from the one for HRSs.
We will use locality to be able to (mostly) postpone reductions below
an abstraction.
In the explanations below, we will argue that left-linearity and
full-extendedness are both necessary to do this.

\paragraaf{Left-linearity}
When a system is not left-linear, a reduction deep inside a term may
be needed to create a topmost redex.  For instance, consider the
non-left-linear AFS with rules $\{ \mathsf{f}(x,x) \arrz \mathsf{b},
\mathsf{a} \arrz \mathsf{b} \}$.
In the reduction 
$\mathsf{f}(\abs{\avar}{\mathsf{a}},\abs{\avar}{\mathsf{b}}) \arrz
\mathsf{f}(\abs{\avar}{\mathsf{b}},\abs{\avar}{\mathsf{b}}) \arrz
\mathsf{b}$
a reduction below an abstraction is necessary to create the syntactic
equality required for the $\mathsf{f}$-rule.
Thus, this step cannot be postponed.

\paragraaf{Full-Extendedness}
We say a term $l$ is \emph{fully extended} if free variables in $l$
do not occur below an abstraction; a rule $l \arrz r$ is called
fully extended if $l$ is.

For the intuition of this restriction, consider a rule
$\mathsf{f}(\abs{\avar}{\bvar}) \arrz \bvar$.
This rule does not match a term $\mathsf{f}(
\abs{\avar}{\suc(\avar)})$, since
$\bvar$ cannot be instantiated with $\suc(\avar)$, as $\avar$ is
bound.  Nor does $\mathsf{f}(\abs{\avar}{\app{F}{\avar}})$ match
this term, since $\suc(\avar)$ does not instantiate the application
$\app{F}{\avar}$.  Whenever the left-hand side of a rule contains a
free variable below an abstraction, this variable matches only
subterms which do not contain the abstraction-variable.
Therefore,
such a rule could require a reduction
deep inside a term to create a topmost redex.  For example, in an
AFS with rules $\{ \mathsf{f}(\abs{\avar}{\bvar}) \arrz \bvar,\ 
\mathsf{g}(\avar,\bvar) \arrz \mathsf{a} \}$, we cannot postpone
the first step in the reduction $\mathsf{f}(\abs{\avar}{\mathsf{h}(
\abs{\bvar}{\mathsf{g}(\avar,\bvar)})}) \arrz \mathsf{f}(\abs{\avar}{
\mathsf{h}(\abs{\bvar}{\mathsf{a}})}) \arrz \mathsf{h}(\abs{\bvar}{
\mathsf{a}})$, as it is needed to create the second redex.

This notion of fully extended mostly corresponds
with the definition for HRSs; there, however, a rule
$\mathsf{f}(\abs{\avar}{\app{F}{\avar}}) \arrz r$ \emph{does} match
$\mathsf{f}(\abs{\avar}{\suc(\avar)})$, so such rules are also
accepted.

\begin{definition} \label{def:llfe}
An AFS $(\F , \Rules)$ is \emph{\llfe} 
if all $l \arrz r \in \Rules$
are left-linear and fully extended.
\end{definition}

\begin{example}
Our running example, $\Rules_\twice$, is \llfe, since all left-hand
sides of rules are linear and fully extended (in fact, they contain
no abstractions at all).
\end{example}
To demonstrate the prominence of \llfe\ AFSs, in the 2011 version of
the Termination Problem Data Base, used in the annual termination
competition~\cite{termcomp}, 138 out of 156 benchmarks in the
higher-order category are \llfe\ (in fact, these are all
left-abstraction-free).

\subsection{Tagging Unreducable Symbols}

Obviously, when there are rules where the left-hand side contains an
abstraction, such as $\mathsf{f}(\abs{\avar}{\mathsf{g}(\avar)},F)
\arrz r$, it may be impossible to avoid reducing inside an
abstraction in order to create a redex.  However, the colouring
intuition still goes through; we merely need to ``paint symbols
green'' a few times more.

Following the colouring intuition, we will mark all function symbols
which occur between a $\lambda$-abstraction and its binder with a
special tag (``colouring red'').  The symbol can only be reduced by
removing the tag first (``painting green'').

\begin{definition}\label{def:ttag}
Let $\F^-$ be the set $\{ \afun^- : \atype \mid \afun : \atype \in \F
\}$, so a set containing a ``tagged'' symbol $\afun^-$ for all
function symbols $\afun \in \F$.
For a set of variables $Z$, define $\ttag_Z$ as follows:
\[
\begin{array}{rcl}
\ttag_Z(\avar) & = & \avar \\
\ttag_Z(\c_\atype) & = & \c_\atype \\
\ttag_Z(\app{\aterm}{\bterm}) & = &
  \app{\ttag_Z(\aterm)}{\ttag_Z(\bterm)} \\
\ttag_Z(\abs{\avar}{\aterm}) & = &
  \abs{\avar}{\ttag_{Z \cup \{\avar\}}(\aterm)} \\
\ttag_Z(\afun(\aterm_1,\ldots,\aterm_n)) & = &
  \left\{\begin{array}{ll}
  \afun(\ttag_Z(\aterm_1),\ldots,\ttag_Z(\aterm_n)) &
    \mathrm{if}\ \FV(\afun(\vec{\aterm})) \cap Z = \emptyset \\
  \afun^-(\ttag_Z(\aterm_1),\ldots,\ttag_Z(\aterm_n)) &
    \mathrm{if}\ \FV(\afun(\vec{\aterm})) \cap Z \neq \emptyset \\
  \end{array}\right. \\
\end{array}
\]
We denote $\ttag(\aterm) := \ttag_\emptyset(\aterm)$.
Define $\Rules^\ttag := \{ l \arrz \ttag_\emptyset(r) \mid
  l \arrz r \in \Rules \} \cup \{ \afun^-(\avar_1,\ldots,\avar_n)
  \arrz \afun(\avar_1,\ldots,\avar_n) \mid \afun^- \in \F^- \}$.
\end{definition}

Note that, apart from the untagging rules, $\Rules^\ttag$ isn't all
that different from $\Rules$: $\ttag(r)$ is almost exactly $r$, only
the symbols below an abstraction may be marked with a $-$ sign.

\begin{example}
$\ttag(\mathtt{f}(\abs{\avar}{\mathtt{g}(\avar,\mathtt{g}(\nul))}))
= \mathtt{f}(\abs{\avar}{\mathtt{g}^-(\avar,\mathtt{g}(\nul))})$.
\end{example}

\begin{example}\label{ex:twicetagged}
Consider our running example $\Rules_\twice$ (with completed rules):
\[
\begin{array}{rclrcl}
\I(\nul) & \arrz & \nul &
\twice(F) & \arrz & \abs{\bvar}{\app{F}{(\app{F}{\bvar})}} \\
\I(\suc(n)) & \arrz & \suc(\app{\twice(\abs{\avar}{\I(\avar)})}{n}
\ \ \ &
\app{\twice(F)}{m} & \arrz & \app{F}{(\app{F}{m})} \\
\end{array}
\]
We have seen that $\Rules_\twice$ is \llfe.
$\Rules^\ttag$ consists of the following rules:
\[
\begin{array}{rclrcl}
\I(\nul) & \arrz & \nul &
\twice(F) & \arrz & \abs{\bvar}{\app{F}{(\app{F}{\bvar})}} \\
\I(\suc(n)) & \arrz & \suc(\app{\twice(\abs{\avar}{\I^-(\avar)})}{n})
\ \ \ &
\app{\twice(F)}{n} & \arrz & \app{F}{(\app{F}{n})} \\
\nul^- & \arrz & \nul &
\suc^-(n) & \arrz & \suc(n) \\
\I^-(n) & \arrz & \I(n) &
\twice^-(F) & \arrz & \twice(F) \\
\end{array}
\]
That is, the rules from $\Rules$, with a tag added to the $\I$
symbol which occurs below an abstraction, and furthermore the
untagging rules.  In termination proofs we can typically
ignore the rules $\afun^-(\vec{\avar}) \arrz \afun(\vec{\avar})$
where $\afun^-$ does not occur in the right-hand side of any rule (in
this example: $\nul^- \arrz \nul,\ \suc^-(n) \arrz \suc(n)$ and
$\twice^-(F,n) \arrz \twice(F,n)$), as they have little function.
\end{example}

In the proofs later on in this section, we will use the following
properties of $\Rules^\ttag$:

\begin{lemma}\label{lem:dropvars}
$\ttag_{X \cup Y}(\aterm) \arrr{\Rules^\ttag} \ttag_X(\aterm)$
for any set of rules $\Rules$; if the variables in $Y$ don't occur
in $\aterm$ even $\ttag_{X \cup Y}(\aterm) = \ttag_X(\aterm)$.
\end{lemma}

\begin{proof}
Easy induction on the size of $\aterm$; we only use the untagging
rules $\afun^-(\vec{\avar}) \arrz \afun(\vec{\avar})$.
\end{proof}

\begin{lemma}\label{lem:tagbasesubstitute}
$\ttag(\aterm)\gamma^\ttag = \ttag(\aterm\gamma)$
where $\gamma^\ttag = [\avar:=\ttag(\gamma(\avar)) \mid \avar \in
\domain(\gamma)]$.
\end{lemma}

\begin{proof}
We prove by induction on the size of $\aterm$: for any set of
variables $Z$, whose members do not occur in either domain or range
of $\gamma$, we have $\ttag_Z(\aterm)\gamma^\ttag =
\ttag_Z(\aterm\gamma)$.

If $\aterm$ is a variable not in $\domain(\asub)$, both sides are
just $\aterm$.

If $\aterm$ is a variable in $\domain(\asub)$, we must see that
$\ttag(\gamma(\aterm)) = \ttag_Z(\gamma(\aterm))$, which holds by
the second part of Lemma~\ref{lem:dropvars}.

If $\aterm$ is an application $\app{\bterm}{\cterm}$, then
$\ttag_Z(\aterm)\gamma^\ttag = \app{(\ttag_Z(\bterm)\gamma^\ttag)}{
(\ttag_Z(\cterm)\gamma^\ttag)}$, which by the induction hypothesis
equals $\app{\ttag(\bterm\gamma)}{\ttag(\cterm\gamma)} =
\ttag(\app{(\bterm\gamma)}{(\cterm\gamma)}) = \ttag(\aterm\gamma)$.

If $\aterm = \afun(\aterm_1,\ldots,\aterm_n)$ the induction
hypothesis on each of the $\aterm_i$ also suffices because $Z \cap
\FV(\aterm) = Z \cap \FV(\subst{\aterm}{\asub})$, which is easy to
see by the requirements on $Z$.

Finally, if $\aterm = \abs{\bvar}{\aterm'}$, then 
$\subst{\ttag_Z(\aterm)}{\asub^\ttag} =
\subst{(\abs{\bvar}{\ttag_{Z \cup \setop
\bvar\setcl}(\aterm')})}{\asub^\ttag} =
\abs{\bvar}{(\subst{\ttag_{Z \cup \setop \bvar\setcl}(\aterm')}{
  \asub^\ttag})}$,
which by the induction hypothesis equals
$\abs{\bvar}{\ttag_{Z \cup \setop \bvar\setcl}(\subst{\aterm'}{
\asub})} = \ttag_Z(\subst{\aterm}{\asub})$ as required.
\end{proof}

\begin{lemma}\label{lem:tagsubstitute}
If $Z$ is a set of variables, $\aterm,\bterm$ terms and $\avar$ a
variable not in $Z$ or $\FV(\bterm)$, then
$\ttag_{Z \cup \{\avar\}}(\aterm)[\avar:=\ttag_Z(\bterm)]
\arrr{\Rules^\ttag} \ttag_Z(\aterm[\avar:=\bterm])$ for any set
of rules $\Rules$.
\end{lemma}

\begin{proof}
By induction on the size of $\aterm$.

If $\aterm = \avar$, then both sides are equal to $\ttag_Z(\bterm)$.

If $\aterm$ is another variable $\bvar$, then both sides are just
$\bvar$.

If $\aterm = \app{\cterm}{\dterm}$ we use the induction hypothesis:
$\ttag_{Z \cup \{\avar\}}(\aterm)[\avar:=\ttag_Z(\bterm)] =
\app{\ttag_{Z \cup \{\avar\}}(\cterm)[\avar:=\ttag_Z(\bterm)]}{
\ttag_{Z \cup \{\avar\}}(\dterm)[\avar:=\ttag_Z(\bterm)]}
\arrr{\Rules^\ttag}
\app{\ttag_Z(\cterm[\avar:=\bterm])}{\ttag_Z(\dterm[\avar:=\bterm])
} = \ttag_Z(\app{\cterm[\avar:=\bterm]}{\dterm[\avar:=\bterm]}) =
\ttag_Z(\aterm[\avar:=\bterm])$.

If $\aterm = \abs{\bvar}{\cterm}$, we use the second part of
Lemma~\ref{lem:dropvars}:
$\ttag_{Z \cup \{\avar\}}(\aterm)[\avar:=\ttag_Z(\bterm)] =
\abs{\bvar}{\ttag_{Z \cup \{\avar,\bvar\}}(\cterm)[\avar:=
\ttag_Z(\bterm)]}$, which by Lemma~\ref{lem:dropvars} equals
$\abs{\bvar}{\ttag_{Z \cup \{\avar,\bvar\}}(\cterm)[\avar:=\ttag_{Z
\cup \{\bvar\}}(\bterm)]} $\\$
\arrr{\Rules^\ttag} \abs{\bvar}{\ttag_{Z
\cup \{\bvar\}}(\cterm [\avar:=\bterm])} = \ttag_Z(\aterm[\avar:=
\bterm])$ by the induction hypothesis.

Finally, if $\aterm = \afun(\cterm_1,\ldots,\cterm_n)$, then there is
little to do if $\avar$ does not occur in $\aterm$: by the second
part of Lemma~\ref{lem:dropvars}, $\ttag_{Z \cup \{\avar\}}(\aterm)[
\avar:=\ttag_Z(\bterm)] = \ttag_Z(\aterm)[\avar:=\ttag_Z(\bterm)]$
and since $\avar$ does not occur in either $\aterm$ or $\ttag_Z(
\aterm)$, this is exactly $\ttag_Z(\aterm[\avar:=\bterm])$.
So assume that $\avar \in \FV(\aterm)$; then $\ttag_{Z \cup
\{\avar\}}(\aterm)[\avar:=\ttag_Z(\bterm)] = \afun^-(\ttag_{Z \cup
\{\avar\}}(\cterm_1),\ldots,\ttag_{Z \cup \{\avar\}}(\cterm_n))
[\avar:=\ttag_Z(\bterm)]$, which by the induction hypothesis reduces
to $\afun^-(\ttag_Z(\cterm_1[\avar:=\bterm]),\ldots,\ttag_Z(
\cterm_n[\avar:=\bterm]))$.  If variables of $Z$ occur in
$\aterm[\avar:=\bterm]$ this is exactly $\ttag_Z(\aterm[\avar:=
\bterm])$, otherwise it reduces in one step to
$\afun(\ttag_Z(\cterm_1[\avar:=\bterm]),\ldots,\ttag_Z(\cterm_n
[\avar:=\bterm]))
= \ttag_Z(\aterm)$.
\end{proof}

The following lemma expresses that a reduction to a term of a certain
form $l$ can be done by only reducing subterms headed by untagged
(``green'') symbols.  Later on, we will use this to see that
the reduction $\bterm_i \arrr{\Rules} \aterm_{i+1}$ in a dependency
chain can be assumed to reduce only untagged symbols.

\begin{lemma}\label{lem:weak}
Let $\Rules$ be a \llfe\ AFS, $l$ a linear, fully extended term and
$\asub$ a substitution on domain $\FV(l)$.
If $\aterm$ is terminating and $\aterm \arrr{\Rules} l\asub$, then
there is a substitution $\bsub$ such that
$\ttag(\aterm) \arrr{\Rules^\ttag} \subst{l}{\bsub^\ttag}$, where
$\bsub^\ttag = [\avar := \ttag(\bsub(\avar)) \mid \avar \in
\domain(\bsub)]$, and $\bsub(\avar) \arrr{\Rules} \asub(\avar)$ for
all $\avar$.
\end{lemma}

\begin{proof}
Towards an induction hypothesis, we will prove the lemma for a term
$l$ which is linear and fully extended in the variables in $\domain(
\asub)$, and such that $\domain(\asub) \subseteq \FV(l)$; $l$ may
have more variables which do not occur in this domain, and which it
is not necessarily linear and fully extended in.  We use induction
first on $\aterm$, using $\arr{\Rules} \cup \supterm$ (this is
well-founded because $\aterm$ is terminating by assumption), second
on the length of the reduction $\aterm \arrr{\Rules}
\subst{l}{\asub}$.

First suppose $l$ is a variable in $\domain(\asub)$, so $\asub =
[l:=\asub(l)]$.  Choose $\bsub(l) = \aterm$.  Then certainly
$\bsub(l) \arrr{\Rules} \asub(l)$, and $\ttag(\aterm)
\arrr{\Rules^\ttag} \ttag(\aterm) = l\bsub^\ttag$.
If $l$ is a variable not in $\domain(\asub)$, and $\aterm
\arrr{\Rules} l\asub = l$ without headmost steps, then $\asub$ is
empty and $\aterm = l$; indeed $\ttag(\aterm) = l \arrr{\Rules^\ttag}
l$.

Next, let $l$ be an abstraction $\abs{\avar}{l'}$ and suppose
$\aterm = \abs{\avar}{\aterm'}$ and $\aterm' \arrr{\Rules} l'\asub$.
Since $l$ is fully extended in the variables of $\domain(\asub)$,
this $l'$ contains no variables in $\domain(\asub)$; that is, $\asub$
is empty, and we must see that $\ttag(\aterm) \arrr{\Rules} l$.  By
the induction hypothesis
$\ttag(\aterm') \arrr{\Rules} l'$, and therefore indeed $\ttag(
\aterm) = \abs{\avar}{\ttag_{\{\avar\}}(\aterm')}
\arrr{\Rules^\ttag} \abs{\avar}{\ttag(\aterm')}$ by
Lemma~\ref{lem:dropvars}, $\arrr{\Rules^\ttag} \abs{\avar}{l'} = l$.

If $l = \afun(l_1,\ldots,l_n)$ and $\aterm = \afun(\aterm_1,\ldots,
\aterm_n)$ and each $\aterm_i \arrr{\Rules} l_i\asub$, then by
linearity of $l$ each $l_i$ has different variables (at least, insofar
as $\domain(\asub)$ is concerned).  We can write $\asub = \asub_1
\cup \ldots \cup \asub_n$ with each $\asub_i$ the restriction of
$\asub$ to $\domain(l_i)$; all $\asub_i$ have disjunct domains.
Then also $\aterm_i \arrr{\Rules} l_i\asub_i$, so by the induction
hypothesis there are $\bsub_1,\ldots,\bsub_n$ such that each
$\ttag(\aterm_i) \arrr{\Rules^\ttag} l_i\bsub_i^\ttag$, and always
$\bsub_i(\avar) \arrr{\Rules} \asub_i(\avar)$.  The induction step
holds with $\bsub := \bsub_1 \cup \ldots \cup \bsub_n$.

If $l = \app{l_1}{l_2}$ and $\aterm = \app{\aterm_1}{\aterm_2}$ and
each $\aterm_i \arrr{\Rules} l_i\asub$, we use linearity in
the same way.

If none of these cases hold, the reduction $\aterm \arrr{\Rules} l\asub$
must use a headmost step, so $\aterm \arrr{\Rules} \cterm \arr{\Rules}
\dterm \arrr{\Rules} l\asub$, and either:
\begin{iteMize}{$\bullet$}
\item $\cterm = \app{(\abs{\avar}{\eterm})}{\fterm_0} \cdots \fterm_n$
  and $\dterm =\app{\eterm[\avar:=\fterm_0]}{\fterm_1} \cdots
  \fterm_n$ ($n \geq 0$),\ \ \emph{or}
\item $\cterm = \app{l'\asub'}{\fterm_1} \cdots \fterm_n$ and
  $\dterm = \app{r'\asub'}{\fterm_1} \cdots \fterm_n$ for some
  $l' \arrz r',\asub,\fterm_1,\ldots,\fterm_n$ ($n \geq 0$)
\end{iteMize}
We can safely assume that the reduction $\aterm \arrr{\Rules}
\cterm$ does not use any headmost steps.

In the first case, $\aterm$ must have the form
$\app{(\abs{\avar}{\eterm'})}{\fterm_0'} \cdots \fterm_n'$
with $\eterm' \arrr{\Rules} \eterm$ and each
$\fterm_i' \arrr{\Rules} \fterm_i$.  But then also
$\app{\eterm'[\avar:=\fterm_0']}{\fterm_1'} \cdots \fterm_n'
\arrr{\Rules} \dterm \arrr{\Rules} l\asub$; by the first induction
hypothesis we find a suitable $\bsub$ such that $\aterm \arr{\beta}
\app{\eterm'[\avar:=\fterm_0']}{\fterm_1'} \cdots \fterm_n'
\arrr{\Rules^\ttag} l\bsub^\ttag$.

In the second case, let $l'' := \app{l}{\avar_1} \cdots \avar_n$ for
fresh variables $\avar_1,\ldots,\avar_n$, and let $\asub'' :=
\asub' \cup [\avar_1:=\fterm_1,\ldots,\avar_n:=\fterm_n]$.  Since
$l' \arrz r'$ is a rule, $l''$ is both linear and fully extended, and
the reduction to $l''\asub'' = \cterm$ is shorter than the original
reduction; by the second induction hypothesis we find $\csub$ such
that $\ttag(\aterm) \arrr{\Rules^\ttag} l''\csub^\ttag =
\app{l'\csub^\ttag}{(\avar_1\csub^\ttag)} \cdots
(\avar_n\csub^\ttag)$, where each $\csub(\bvar)$ reduces to
$\asub'(\bvar)$ for $\bvar \in \domain(\asub')$, and $\csub(\avar_i)
\arrr{\Rules} \fterm_i$.

Now, $l''\csub^\ttag \arr{\Rules^\ttag} \app{\ttag(r')\csub^\ttag
}{(\vec{\avar}\csub^\ttag)}$, which by
Lemma~\ref{lem:tagbasesubstitute} $= \app{\ttag(r'\csub)}{(\vec{\avar
}\csub^\ttag)} = \ttag((\app{r'}{\vec{\avar}})\csub)$.
By simply removing all tags, every $\arr{\Rules^\ttag}$ step can be
translated to a $\arr{\Rules}^=$ step on untagged terms, and therefore
we also see that $\aterm \arrr{\Rules} l''\csub \arr{\Rules} (\app{
r}{\vec{\avar}})\csub$, and by the choice of $\csub$ we know:
$(\app{r}{\vec{\avar}})\csub \arrr{\Rules} \app{(r'\asub')}{\vec{
\fterm}} \arrr{\Rules} l\asub$.  Therefore we can apply the first
induction hypothesis, and see that $\ttag(\aterm) \arr{\Rules^\ttag
}^+ \ttag((\app{r}{\vec{\avar}})\csub) \arrr{\Rules^\ttag}
l\bsub^\ttag$, for a suitable $\bsub$.
\end{proof}

In Section~\ref{subsec:weakrevised}, Lemma~\ref{lem:weak} will play
an essential role in the construction of a ``tagged dependency
chain''.  But first, let us consider \emph{formative rules}, another
gain from locality.

\subsection{Formative Rules}\label{subsec:formrules}

Recall that in the first-order setting it is not required to prove
$l \geqterm r$ for \emph{all} rewrite rules: to prove that a set of
dependency pairs $\P$ is chain-free it suffices to consider only its
\emph{usable rules}.
The definition of usable rules cannot easily be extended to our
setting, because we normally have to deal with collapsing dependency
pairs.  Therefore we take a different approach with the same goal of
restricting attention to rules which are in some way relevant to a
set of dependency pairs.  Where usable rules are defined from the
right-hand sides of dependency pairs, our \emph{formative rules} are
based on the left-hand sides.

The intuition behind formative rules is that (due to left-linearity
and full-extendedness), only the formative rules of some rule $l
\arrz r$ can contribute to the creation of its pattern.

We consider a fixed set of rules $\Rules$, which has already been
completed.  The formative rules are a subset of $\Rules^+$, which is
the set $\Rules \cup \{\app{l}{\avar_1} \cdots \avar_n \arrz \app{r}{
\avar_1} \cdots \avar_n \mid l \arrz r \in \Rules$, all $\avar_i$
fresh variables, $r$ not an abstraction and $\app{l}{\avar_1} \cdots
\avar_n$ well-typed$\}$.

\begin{definition}[Formative Rules]\label{def:formrules}
Let $X$ be a set of variables, and $\aterm$ a $\beta$-normal term
(that is, $\aterm$ has no subterms $\app{(\abs{\avar}{\bterm})}{
\cterm}$) such
that for any subterm $\app{\avar}{\bterm}$ of $\aterm$ with $\avar
\in \setvar$, either $\avar$ is not free in $\aterm$, or $\avar \in X$.
Let $\Symbols_X(\aterm)$ be recursively defined as follows:
\[
\begin{array}{rcl}
\Symbols_X(\abs{\bvar}{\aterm} : \atype)
  & = & \{ \langle \mathit{ABS}, \atype \rangle \} \cup
  \Symbols_{X \cup \{\bvar\}}(\aterm) \\
\Symbols_X(\app{\afun(\aterm_1,\ldots,\aterm_n)}{\aterm_{n+1}} \cdots
  \aterm_m : \atype) & = & \{ \langle \afun, \atype \rangle \} \cup
  \Symbols_X(\aterm_1) \cup \ldots \cup \Symbols_X(\aterm_m) \\
\Symbols_X(\app{\bvar}{\aterm_1} \cdots \aterm_n : \atype) & = &
  \{ \langle \mathit{VAR}, \atype \rangle \cup \Symbols_X(\aterm_1)
  \cup \ldots \cup \Symbols_X(\aterm_m) \\
  & & \ \ \ \ \ \ \ (\bvar \in X,\ n \geq 0) \\
\Symbols_X(\bvar) & = & \emptyset\ \ (\bvar \in \setvar \setminus X) \\
\end{array}
\]
Note that in a local AFS, all left-hand sides of the rules satisfy
these constraints for $X = \emptyset$.

For $a \in \F \cup \{\mathit{ABS},\mathit{VAR}\}$, we say a term
$\aterm : \atype$ \emph{has form} $\langle a, \atype \rangle$ if
either $a = \mathit{ABS}$ and $\aterm$ is an abstraction, or $a \in
\F$ and $\aterm$ can be written $\app{a(\vec{\bterm})}{\vec{\cterm}}$,
or $\aterm = \app{\avar}{\vec{\bterm}}$ for some variable $\avar$
(and $a$ may be anything).  A pair $\langle a, \atype \rangle$ with
$a \in \F \cup \{\mathit{ABS},\mathit{VAR}\}$ is called a \emph{typed
symbol}.

For two typed symbols $A,B$,
write $A \gsymbform B$ if there is a rule $l \arrz r \in \Rules^+$
such that $r$ has form $A$, and $B \in \Symbols_\emptyset(l)$.  Let
$\gsymbform^*$ denote the reflexive-transitive closure of
$\gsymbform$.

The \emph{formative symbols} of a term $\aterm$ are those typed
symbols $B$ such that $A \gsymbform^* B$ for some $A \in
\Symbols_\emptyset(\aterm)$ (if defined).

The \emph{formative rules} of a term $\aterm$, notation
$\formrules(\aterm)$, are those rules $l \arrz r \in \Rules^+$ such
that $r$ has form $B$ for some formative symbol $B$ of $\aterm$.

The set of formative rules of a dependency pair,
$\formrules(\app{\afun(l_1,\ldots,l_n)}{l_{n+1}} \cdots l_m \dppijl
p)$, is defined as $\bigcup_{1 \leq i \leq m} \formrules(l_i)$.  
For a set $\P$ of dependency pairs, $\formrules(\P) =
\bigcup_{l \dppijl p \in \P} \formrules(l \dppijl p)$.
\end{definition}

Note that in a finite system it is easy to calculate the formative
symbols of a term, and consequently the formative rules can be found
automatically.

\begin{example}\label{ex:urtwice}
Recall the rules for the (completed) system $\twice$:
\[
\begin{array}{lrcllrcl}
(A) & \I(\nul) & \arrz & \nul &
(C) & \twice(F) & \arrz & \abs{y}{F \cdot (F \cdot y)}  \\
(B) & \I(\suc(n)) & \arrz & \suc(\twice(\abs{x}{\I(x)}) \cdot  n) &
(D) & \twice(F) \cdot  m & \arrz & F \cdot (F \cdot m) \\
\end{array}
\]
Here $\Rules^+ = \Rules$.
In this context, let $l = \suc(n)$.  Then $\Symbols_\emptyset(l) = \{
\langle \suc, \nat \rangle \}$, and:
\begin{iteMize}{$\bullet$}
\item (B) and (D) both have form $\langle \suc,\nat \rangle$, so
  $\langle \suc,\nat \rangle \gsymbform \langle \suc,\nat \rangle,
  \langle \I,\nat \rangle, \langle \twice,\nat \rangle$
\item (D) also has forms $\langle \I,\nat \rangle$ and $\langle
  \twice,\nat \rangle$, but no other rules do
\end{iteMize}
Thus, the formative symbols of $l$ are exactly $\langle \suc,\nat
\rangle,\ \langle \I,\nat \rangle$ and $\langle \twice,\nat \rangle$.
(B) and (D), but not (A) and (C), are formative rules of $l$.
Observing that a dependency pair with left-hand side
$\twice(F)\cdot n$ or $\up{\twice}(F)$ has no formative rules
(since $\Symbols_\emptyset(F) = \Symbols_\emptyset(n) =
\emptyset$),
the formative rules of the SCC $\scctwice$ from
Example~\ref{ex:twicecycles} are (B) and (D).
\end{example}

\begin{example}
For an example that uses multiple types, and more rules of functional
type, consider the system with symbols
\[
\begin{array}{rclrcl}
\cons & : & [(\nat \typepijl \nat) \times \listex] \decpijl \listex &
\nil & : & \listex \\
\headex & : & [\listex] \decpijl \nat \typepijl \nat &
\tailex & : & [\listex] \decpijl \listex \\
\trueex & : & \booleanex &
\falseex & : & \booleanex \\
\mathtt{test} & : & [\nat \typepijl \nat] \decpijl \booleanex &
\suc & : & [\nat] \decpijl \nat \\
\ifex & : & \multicolumn{4}{l}{
  [\booleanex \times (\nat \typepijl \stringex) \times
  (\nat \typepijl \stringex)] \decpijl \nat \typepijl \stringex} \\
\end{array}
\]
And rules:
\[
\begin{array}{lrcllrcl}
(A) & \ifex(\trueex,F_1,F_2) & \arrz & F_1 &
(D) & \headex(\cons(F,t)) & \arrz & F \\
(B) & \ifex(\falseex,F_1,F_2) & \arrz & F_2 &
(E) & \tailex(\cons(F,t)) & \arrz & t \\
(C) & \mathtt{test}(\abs{\avar}{\suc(\avar)}) & \arrz & \trueex \\
\end{array}
\]
For $\Rules^+$, we add the rules:
\[
\begin{array}{lrcl}
(F) & \app{\ifex(\trueex,F_1,F_2)}{\avar} & \arrz & \app{F_1}{\avar} \\
(G) & \app{\ifex(\falseex,F_1,F_2)}{\avar} & \arrz & \app{F_2}{\avar} \\
(H) & \app{\headex(\cons(F,t))}{\avar} & \arrz & \app{F}{\avar} \\
\end{array}
\]
This is a contrived example, to demonstrate all aspects of formative
rules in one system.  We consider the formative rules of the
dependency pair $\app{\ifex(\trueex,F_1,F_2)}{\avar} \dppijl
\app{F_1}{\avar}$.  That is, $\formrules(\trueex)$,
since the free variables $F_1,F_2,\avar$ do not have formative rules.
We observe:
\begin{iteMize}{$\bullet$}
\item $\Symbols_\emptyset(\trueex) = \langle \trueex, \booleanex
  \rangle$
\item rule (C) is the only rule with
  form $\langle \trueex, \booleanex \rangle$, so
  $\langle \trueex, \booleanex \rangle \gsymbform \langle \mathtt{
  test}, \booleanex \rangle, \linebreak
  \langle \mathit{ABS}, \nat \typepijl
  \nat \rangle, \langle \suc, \nat \rangle, \langle \mathit{VAR},
  \nat \rangle$, that is, the elements of
  $\Symbols_\emptyset(\mathtt{test}(\abs{\avar}{\suc(\avar)}))$;
\item $\langle \mathit{ABS}, \nat \typepijl \nat \rangle \gsymbform
  \langle \headex, \nat \typepijl \nat \rangle, \langle \cons, \listex
  \rangle$ by rule (D), and rule (H) has both form $\langle \suc,\nat
  \rangle$ and $\langle \mathit{VAR}, \nat \rangle$, so these two
  $\gsymbform \langle \headex, \nat \rangle, \langle \cons, \listex
  \rangle$; no other rule has a form $\langle \mathtt{test},
  \booleanex \rangle, \langle \mathit{ABS}, \nat \typepijl \nat
  \rangle, \langle \suc, \nat \rangle$ or $\langle \mathit{VAR},\nat
  \rangle$
\item (D) also has form $\langle \headex, \nat \typepijl \nat
  \rangle$ (but we already know that all symbols in the left-hand
  side are formative symbols of $\trueex$), and only (E) has form
  $\langle \cons,\listex \rangle$, so the latter $\gsymbform \langle
  \tailex, \listex \rangle, \langle \cons, \listex \rangle$
\item Hence, the formative symbols of the given dependency pair are:\\
  $
  \langle \trueex, \booleanex \rangle,
  \langle \mathtt{test}, \booleanex \rangle,
  \langle \suc, \nat \rangle,
  \langle \mathit{VAR}, \nat \rangle,
  \langle \headex, \nat \rangle,
  \langle \cons, \listex \rangle,
  \langle \tailex,\newline \listex \rangle,
  \langle \mathit{ABS}, \nat \typepijl \nat \rangle,
  \langle \headex, \nat \typepijl \nat \rangle
  $.
\end{iteMize}

\noindent
The formative rules are therefore (C), (D), (E) and (H).
\end{example}

\noindent
For formative rules we have a result very similar to
Lemma~\ref{lem:weak}, both in nature and in proof.

\begin{lemma}\label{lem:formativereduction}
Suppose $\Rules$ is local, and let $l$ be a $\beta$-normal, linear,
fully extended term, which does not have leading free variables.
Let $\asub$ be a substitution with domain $\FV(l)$, and $\aterm$ a
term which is terminating over $\arr{\Rules}$, and suppose $\aterm
\arrr{\Rules} l\asub$.
Then there is a substitution $\bsub$ on $\FV(l)$ such that $\aterm
\arrr{\formrules(l)} l\bsub$, and moreover each $\bsub(\avar)
\arrr{\Rules} \asub(\avar)$.
\end{lemma}

\begin{proof}
We will prove something slightly stronger, which implies the lemma.
Let $X$ be a set of variables, and $l$ a $\beta$-normal term,
linear in $\FV(l) \setminus X$, and such that if a
free variable $\avar$ occurs inside an abstraction, or at the head of
an application in $l$, then $\avar \in X$.
Let $\asub$ be a substitution with domain $\FV(l) \setminus X$, and
$\aterm$ a terminating term such that $\aterm \arrr{\Rules} l\asub$.

Let $\formsymb_X(l)$ denote the set of typed symbols $B$ such that
$A \gsymbform^* B$ for some $A \in \Symbols_X(l)$, and
$\formrules_X(l)$ is the set of rules $l' \arrz r'$ in $\Rules^+$ such
that $r'$ has form $B$ for some $B \in \formsymb_X(l)$.
We will find a substitution $\bsub$ on $\FV(l)
\setminus X$ such that $\aterm \arrr{\formrules(l)} l\bsub$, and
always $\bsub(\avar) \arrr{\Rules} \asub(\avar)$.

It is clear that, for $X = \emptyset$, the definitions of
$\formrules_X(l)$ and $\formrules(l)$ coincide.  Thus, the case $X =
\emptyset$ implies the lemma -- but for the induction step we will
need a larger $X$.

Before proving this claim, let us make the following observations:
\begin{enumerate}[(1)]
\item\label{it:symbolssub}
  if $\Symbols_X(\aterm) \subseteq \Symbols_Y(\bterm)$, then
  $\formsymb_X(\aterm) \subseteq \formsymb_Y(\bterm)$, so
  $\formrules_X(\aterm) \subseteq \formrules_Y(\bterm)$
\item\label{it:extravars}
  if the variables of $Y$ do not occur in $\aterm$, then
  $\Symbols_{X \cup Y}(\aterm) = \Symbols_X(\aterm)$
\item\label{it:subtermsub}
  if $\aterm \supterm \bterm$ and $Y = \FV(\bterm) \setminus
  \FV(\aterm)$, then $\Symbols_{X \cup Y}(\bterm) \subseteq
  \Symbols_X(\aterm)$
\item\label{it:formrulesub}
  if $u \arrz v \in \formrules_X(\aterm)$, then $\formsymb_\emptyset(
  u) \subseteq \formsymb_X(\aterm)$
\end{enumerate}
All of these are obvious by considering the respective definitions.

Now we have all the preparations to prove the required result, using
induction on $\aterm$ with $\arr{\Rules} \cup \supterm$.  Because the
rules have been completed and $l$ is $\beta$-normal, we can first
transform the reduction $\aterm \arrr{\Rules} l\asub$ into a
reduction which never takes a headmost step with a rule $l' \arrz
\abs{\avar}{r'}$ which is not also a topmost step (we can replace
these steps one by one, and by induction on $\aterm$ with
$\arr{\Rules}$ we eventually obtain a reduction without such steps).
Having done this, we use a second induction, on the length of the
reduction $\aterm \arrr{\Rules} l\asub$.  Now, we can prove the
claim.  Consider the form of $l$.

If $l$ is a variable in $\domain(\asub)$, then $\asub = [l:=\asub(l)]
$; choosing $\bsub := [l:=\aterm]$ we are done.

If $l = \app{\avar}{l_1} \cdots l_n$ with $\avar \in X$, and $\aterm
= \app{\avar}{\aterm_1} \cdots \aterm_n$ and each $\aterm_i
\arrr{\Rules} l_i\asub$, then by linearity of $l$ over
$\domain(\asub)$ we can write $\asub = \asub_1 \cup \ldots \cup
\asub_n$ where $\asub_i$ is the restriction of $\asub$ to
$\FV(l_i)$.  By the induction hypothesis we can find $\bsub_1,\ldots,
\bsub_n$ such that each $\aterm_i \arrr{\formrules_X(l_i)} l_i\bsub_i$
and $\bsub_i \arrr{\Rules} \asub_i$.  Choose $\bsub := \bsub_1 \cup
\ldots \cup \bsub_n$; this is well-defined because all $\bsub_i$ have
disjunct domains.  By (\ref{it:subtermsub}) and
(\ref{it:symbolssub}) each $\formrules_X(l_i) \subseteq \formrules_X(
l)$, so indeed $\aterm \arrr{\formrules_X(l)} l\bsub$, and also each
$\bsub(\avar) \arrr{\Rules} \asub(\avar)$.

If $l = \app{\afun(l_1,\ldots,l_m)}{l_{m+1}} \cdots l_n$, and $\aterm
= \app{\afun(\aterm_1,\ldots,\aterm_m)}{\aterm_{m+1}} \cdots \aterm_n$
and each $\aterm_i \arrr{\Rules} l_i\asub$, we use linearity in almost
exactly the same way.

If $l = \abs{\avar}{l'}$ and $\aterm = \abs{\avar}{\aterm'}$ and
$\aterm' \arrr{\Rules} l'\asub$, then by assumption the
term $l'$ contains only variables in $X$, so $\domain(\gamma) =
\emptyset$; we must show that $\aterm' \arrr{\formrules_X(l)} l'$.
By (\ref{it:subtermsub}), (\ref{it:extravars}) and
(\ref{it:symbolssub}), it suffices if $\aterm'
\arrr{\formrules_{X \cup \{\avar\}}(l')} l'$, and this is
exactly what the induction hypothesis gives us!

By the restrictions on $l$, it must have one of the forms above; if
we are not yet done, therefore, the reduction $\aterm \arrr{\Rules}
l\asub$ uses a headmost step.

If $\aterm$ has the form $\app{\app{(\abs{\avar}{\bterm})}{\cterm}}{
\vec{\dterm}}$, then the first headmost step must be a $\beta$-step:
$\aterm \arrr{\Rules} \app{\app{(\abs{\avar}{\bterm'})}{\cterm'}}{
\vec{\dterm}} \arr{\beta} \app{\bterm'[\avar:=\cterm']}{\vec{\dterm'}
} \arrr{\Rules} l\asub$; we might as well $\beta$-reduce
immediately, and have $\aterm \arr{\beta} \app{\bterm[\avar:=
\cterm]}{\vec{\dterm}} \arrr{\Rules} \app{\bterm'[\avar:=\cterm']}{
\vec{\cterm'}} \arrr{\Rules} l\asub$; the first induction hypothesis
gives a suitable $\bsub$.

If $\aterm$ does not have this form, there is at least one
headmost step which is not a $\beta$-reduction.  The reduction has a
form $\aterm \arrr{\Rules} \bterm \arr{\Rules}
\cterm \arrr{\Rules} l\asub$, where $\bterm = \app{l'\asub'}{\dterm_1
} \cdots \dterm_n \arr{\Rules} \app{r'\asub'}{\dterm_1} \cdots
\dterm_n$ for some rule $l' \arrz r'$, substitution $\asub'$ and terms
$\dterm_1,\ldots,\dterm_n$ ($n \geq 0$); we can choose $\bterm,
\cterm$ in such a way that the reduction $\cterm \arrr{\Rules} l\asub$
does not contain any headmost steps other than perhaps $\beta$-steps.
Let $l'' := \app{l'}{\avar_1} \cdots \avar_n$ and $r'' :=
\app{r'}{\avar_1} \cdots \avar_n$ for suitably typed fresh variables
$\avar_1,\ldots,\avar_n$; then $l'' \arrz r''$ is in $\Rules^+$,
because we have made sure that either $r'$ is not an abstraction, or
$n = 0$.
Let $\asub'' := \asub' \cup [\avar_1:=\dterm_1,\ldots,\avar_n:=
\dterm_n]$.  Then $\bterm = l''\asub''$ and $\cterm = r''\asub''$.
Applying the second induction hypothesis on the reduction $\aterm
\arrr{\Rules} l''\asub''$, we find some substitution $\csub$ such
that $\aterm \arrr{\formrules_\emptyset(l'')} l''\csub \arr{\Rules}
r''\csub \arrr{\Rules} r''\asub'' \arrr{\Rules} l\asub$.  Note that
$\formrules_\emptyset(l'') \subseteq \Rules^+$, and that
$\arr{\Rules^+}$ defines the same relation as $\arr{\Rules}$.  Thus,
$\aterm \arrr{\Rules} l''\csub \arr{\Rules} r''\csub$; we can
apply the first induction hypothesis to find a suitable $\bsub$
such that $r''\csub \arrr{\formrules_X(l)} l\bsub$.

Suppose $l'' \arrz r'' \in \formrules_X(l)$.  Then by
(\ref{it:formrulesub}) and (\ref{it:symbolssub}), also
$\formrules_\emptyset(l'') \subseteq \formrules_X(l)$, so we have a
reduction $\aterm \arrr{\formrules_X(l)} l''\csub \arr{\formrules_X(
l)} r''\csub \arrr{\formrules_X(l)} l\bsub$, and we are done.
To see that this is indeed the case, first suppose that $\head(r')$
is a variable.  Whatever the form of $l$ is (since $l \notin
\domain(\asub)$), $\Symbols_X(l)$ contains a pair $\langle
\afun,\atype \rangle$, where $\atype$ is the type of $l$ (and also
the type of $\aterm,\ l''$ and $r''$), and $\afun \in \F \cup
\{ \mathit{ABS}, \mathit{VAR} \}$.  We immediately see that $l''
\arrz r'' \in \formrules_X(l)$.
Alternatively, if $\head(r')$ is a function symbol, then $\cterm$ is not a
$\beta$-redex; as the reduction $\cterm \arrr{\Rules} l\asub$
does not use other headmost steps, we have $\cterm
\arrr{\Rules,in} l\asub$, and $l =
\app{\afun(l_1,\ldots,l_k)}{l_{k+1}} \cdots l_m$, where $\afun$ is
also the head-symbol of $r'$.  But then $\langle \afun,\atype \rangle
\in \Symbols_X(l)$, so also $l'' \arrz r'' \in \formrules_X(l)$.
\end{proof}

Of course, Lemma~\ref{lem:formativereduction} and
Lemma~\ref{lem:weak} can be combined: the latter doesn't care
\emph{which} rules it is given, so if $\aterm \arrr{\Rules} l\asub$,
then there are $\bsub,\csub$ such that $\aterm \arrr{\formrules(l)}
l\bsub$, and $\ttag(\aterm) \arrr{\formrules(l)^\ttag}
l\csub^\ttag$ and each $\csub(\avar) \arrr{\formrules(l)}
\bsub(\avar) \arrr{\Rules} \asub(\avar)$.  In the following, we will
use this combination of lemmas to see that, for \llfe\ AFSs, a
dependency chain can be assumed to use tagged steps and formative
rules in the reduction $\bterm_i \arrr{\Rules,in} \aterm_{i+1}$.

\emph{Comment:}
The formative rules technique is also applicable to first-order
rewriting, in particular for many-sorted TRSs (or for innermost
rewriting where types may be added by~\cite{fuh:gie:par:sch:swi:11}).
However, we have not yet investigated whether the technique leads to
an improvement in current state-of-the-art termination provers.

\subsection{Revised Dependency Pair Results for Local AFSs}
\label{subsec:weakrevised}

We may now revise the results from Section~\ref{sec:basic} to take
locality into account.  As before, we assume that the rules in
$\Rules$ are all completed, and let $\DP$ be the dependency pairs of
$\Rules$.  Because of Lemmas~\ref{lem:weak}
and~\ref{lem:formativereduction} we can consider an alternative
definition of dependency chain.

\begin{definition}\label{def:taggedchain}
A \emph{tagged dependency chain} is a sequence
$\rijtje{(\rho_i,\aterm_i, \bterm_i)\ |\ i \in \N}$
with for all $i$:
\begin{enumerate}[(1)]
\item $\rho_i \in \DP \cup \{\mathtt{beta}\}$
\item
  if $\rho_i = l_i \dppijl p_i \in \DP$ then $\aterm_i =
  \subst{l_i}{\asub^\ttag}$ and $\bterm_i = \subst{\ttag(p_i)}{
  \asub^\ttag}$ for some substitution $\asub$
\item if $\rho_i = \mathtt{beta}$ then $\aterm_i =
  \ttag(\app{\app{(\abs{\avar}{\cterm})}{\dterm}}{\eterm_1} \cdots
  \eterm_k)$ and either
  \begin{enumerate}[(a)]
  \item
    $k > 0$ and $\bterm_i = \ttag(\app{\cterm[\avar:=\dterm]}{
    \eterm_1} \cdots \eterm_k)$, or
  \item
    $k = 0$ and there exists a term $\eterm$ such that
    $\cterm \suptermeq \eterm$ and $\avar \in \FV(\eterm)$ and
    $\bterm_i = \ttag(\up{\eterm}[\avar:=\dterm])$, but $\eterm
    \neq \avar$
  \end{enumerate}
\item $\bterm_i \arrr{\formrules(l_{i+1})^\ttag,in} \aterm_{i+1}$
\end{enumerate}
\end{definition}

\noindent
A tagged dependency chain is \emph{minimal} if
$\mathit{untag}(\cterm)$ is terminating under $\arr{\Rules}$ for all
strict subterms $\cterm$ of each $\bterm_i$ (where $\mathit{untag}(\ 
)$ removes the $-$ tags).

This definition is similar to the original definition of a dependency
chain, but uses tags for $\aterm_i$ and $\bterm_i$ and limits the
rules in the $\arrr{in}$ reduction to the formative rules of the
pattern which is created.
We obtain the following variation of
Theorem~\ref{thm:dependencychainll}:

\begin{theorem}\label{thm:dependencychainweak}
A \llfe\ AFS $\Rules$ is non-terminating if and only if it admits a
minimal tagged dependency chain.
\end{theorem}

\begin{proof}
If we remove the tags from a tagged dependency chain, we obtain a
normal dependency chain.  Since \llfe\ AFSs are left-linear,
Theorem~\ref{thm:dependencychainll} provides one direction.

For the other direction, we follow the proof of
Theorem~\ref{thm:dependencychain}; in each step $i$ we have a
minimal non-terminating, untagged term $\fterm_i$, and $\bterm_i =
\ttag(\up{\fterm_i})$.
If $\head(\fterm_i)$ is an abstraction we follow the proof of
Theorem~\ref{thm:dependencychain} to find $\fterm_{i+1}$; the
requirements of Definition~\ref{def:taggedchain} are satisfied for
$\aterm_{i+1} := \ttag(\fterm_i)$ and $\bterm_{i+1} := \ttag(\up{
\fterm_{i+1}})$.  Otherwise, let $\fterm_i = \app{\afun(\dterm_1,
\ldots,\dterm_m)}{\dterm_{m+1}} \cdots \dterm_n$.

Since $\fterm_i$ is MNT, an infinite $\arrr{\Rules}$-reduction
starting in $\fterm_i$ must eventually take a headmost step,
say $\fterm_i \arrr{\Rules,in} \app{\subst{l}{\asub}}{\dterm_{k+1}'}
\cdots \dterm_n'$ (with $k \geq m$), where $l \arrz r \in \Rules$ and
$\app{\subst{r}{\asub}}{\dterm_{k+1}'} \cdots \dterm_m'$ is still
non-terminating.  Write $l = \app{\afun(\tilde{l}_1,\ldots,
\tilde{l}_m)}{\tilde{l}_{m+1}} \cdots \tilde{l}_k$; by
left-linearity all $\tilde{l}_j$ have disjunct free variables.
Applying Lemmas~\ref{lem:weak} and~\ref{lem:formativereduction} on
each of the
$\dterm_j$ and $\tilde{l}_j$ (with the suitable part $\asub_j$ of
$\gamma$), such a redex can be reached with tagged steps and
formative rules:
$\ttag(\dterm_j) \arrr{\formrules(\tilde{l_j})^\ttag} \subst{\tilde{
l_j}}{\bsub_j^\ttag}$, and $\bsub_j \arrr{\Rules} \asub_j$.
Choosing $\bsub := \bsub_1 \cup \ldots \cup
\bsub_n$ we have that $\app{\subst{r}{\bsub}}{\dterm_{k+1}} \cdots
\dterm_n \arrr{\Rules} \app{\subst{r}{\asub}}{\dterm_{k+1}'} \cdots
\dterm_n'$ is still non-terminating.
Let $\fterm_i' := \app{\subst{l}{\bsub}}{\dterm_{k+1}} \cdots
\dterm_m$ and $\aterm_{i+1} := \app{\subst{l}{\bsub^\ttag}}{
\ttag(\dterm_{k+1})} \cdots \ttag(\dterm_m)$ and continue the
proof as before; in the resulting dependency pair $l_{i+1} \dppijl
p_{i+1}$ all $\tilde{l}_j$ are immediate subterms of $l_{i+1}$, so
$\formrules(\tilde{l}_j) \subseteq \formrules(l_{i+1} \dppijl
p_{i+1})$.  We have $\fterm_{i+1} = |p_{i+1}|\csub$ for some
substitution $\chi$, so $\bterm_{i+1} := \ttag(\fterm_{i+1}) =
\ttag(p_{i+1})\csub^\ttag$ as required, by
Lemma~\ref{lem:tagbasesubstitute}.
\end{proof}

\begin{example}
Consider once more the non-terminating system from
Example~\ref{ex:nonterm}
\[
\begin{array}{rclrclrcl}
\mathtt{f}(\nul) & \arrz & \mathtt{g}(\abs{\avar}{\mathtt{f}(\avar)},
  \mathtt{a}) \ \ \ \ &
\mathtt{g}(F, \mathtt{b}) & \arrz & \app{F}{\nul} \ \ \ \ &
\mathtt{a} & \arrz & \mathtt{b} \\
\end{array}
\]
Noting that $\Rules^\ttag$ consists of the rules
\[
\begin{array}{rclrcl}
\mathtt{f}(\nul) & \arrz & \mathtt{g}(\abs{\avar}{\mathtt{f}^-
  (\avar)}, \mathtt{a})\ \ \ \ \ \ \ \ &
\mathtt{g}(F, \mathtt{b}) & \arrz & \app{F}{\nul} \\
\mathtt{a} & \arrz & \mathtt{b} &
\mathtt{f}^-(\avar) & \arrz & \mathtt{f}(\avar) \\
\end{array}
\]
as well as some other rules
$\cfun^-(\vec{\avar}) \arrz \cfun(\vec{\avar})$,
we have the following tagged dependency chain:
\[
\begin{array}{lllll}
(& \up{\mathtt{f}}(\nul) \dppijl \up{\mathtt{g}}(\abs{\avar}{\mathtt{
  f}^-(\avar)},\mathtt{a}), & \up{\mathtt{f}}(\nul), &
  \up{\mathtt{g}}(\abs{\avar}{\mathtt{f}^-(\avar)},\mathtt{a}) & ) \\
(& \up{\mathtt{g}}(F,\mathtt{b}) \dppijl \app{F}{\nul}, &
  \mathtt{g}(\abs{\avar}{\mathtt{f}^-(\avar)},\mathtt{b}), &
  \app{(\abs{\avar}{\mathtt{f}^-(\avar)})}{\nul} & ) \\
(& \mathtt{beta}, & \app{(\abs{\avar}{\mathtt{f}^-(\avar)})}{\nul}, &
  \up{\mathtt{f}}(\nul) & )\\
(& \up{\mathtt{f}}(\nul) \dppijl \up{\mathtt{g}}(\abs{\avar}{\mathtt{
  f}^-(\avar)},\mathtt{a}), & \up{\mathtt{f}}(\nul), &
  \up{\mathtt{g}}(\abs{\avar}{\mathtt{f}^-(\avar)},\mathtt{a}) & ) \\
 & \ldots \\
\end{array}
\]
Here, the $\mathtt{beta}$ step uses case \ref{chain:subterm} with
$\eterm = \mathtt{f}(\avar)$.
\end{example}

\noindent
As in Section~\ref{sec:basic} we consider the dependency graph of
$\Rules$.
A set $\P \subseteq \DP$ is \emph{tagged-chain-free} if there is no
minimal tagged dependency chain using only dependency pairs from
$\P$, and $\mathtt{beta}$.  As before,
$\emptyset$ is tagged-chain-free, and
$\Rules$ is terminating if and only
if every SCC in a graph approximation is tagged-chain-free.
Thus, we can use reduction triples which orient the
parts of a \emph{tagged} dependency chain.  Importantly, this
affects the limited subterm property.
\begin{definition}[Tagged Subterm Property]
$\geqterm$ has the \emph{tagged subterm property} if the following
requirement is satisfied:
\emph{
for all variables $\avar$ and terms $\aterm,\bterm,\cterm$ such that
$\aterm \suptermeq \cterm \supterm \avar$, there is a substitution
$\asub$
such that
$\ttag(\app{(\abs{\avar}{\aterm})}{\bterm}) \geqterm
\ttag(\up{\cterm}[\avar:=\bterm]\asub)$.}
\end{definition}

As we will see shortly, the tagged subterm property is an improvement
over the limited subterm property because we do not have to take the
subterms of untagged functional terms $\afun(\vec{\aterm})$.
It is easy to adapt the proof of Theorem~\ref{thm:maintheorem}
to obtain the following result:

\begin{theorem}\label{thm:maintheoremweak}
A set $\P = \P_1 \uplus \P_2$ of dependency pairs is tagged-chain-free
if $\P_2$ is tagged-chain-free and there is a reduction triple
$(\geqterm,\geqterm_1,\gterm)$ such that:
\begin{iteMize}{$\bullet$}
\item $l \gterm \ttag(p)$ for all $l \dppijl p \in \P_1$;
\item $l \geqterm \ttag(p)$ for all $l \dppijl p \in \P_2$;
\item $l \geqterm_1 \ttag(r)$ for all $l \arrz r \in \formrules(\P)$;
\item $\afun^-(\vec{\avar}) \geqterm_1 \afun(\vec{\avar})$ for all
  $\afun^- \in \F^-$;
\item either $\P$ is non-collapsing or $\geqterm$ has the tagged
  subterm property.
\end{iteMize}
\end{theorem}

\begin{proof}
If the properties above are satisfied, then every minimal tagged
dependency chain over $\P$ leads to an infinite decreasing
$\gterm$-chain, contradicting well-foundedness of $\gterm$.  The
elements of this proof are straightforward, following the proof of
Theorem~\ref{thm:maintheorem}, except perhaps for the proof that
there is a
substitution $\bsub$ such that
$\aterm_i \geqterm \bterm_i\bsub$ when $\aterm_i =
\ttag(\app{\app{(\abs{\avar}{\cterm})}{\dterm}}{\vec{\eterm}})$ and
either $\bterm_i = \ttag(\up{\fterm}[\avar:=\dterm])$ (if $|\vec{
\eterm}| = 0$ and $\cterm \suptermeq \fterm \supterm \avar$), or
$\bterm_i = \ttag(\app{\subst{\cterm}{[\avar:=\dterm]}}{\vec{\eterm}})$.

The first case of this holds by the tagged subterm property:
$\ttag(\aterm_i) = \ttag(\app{(\abs{\avar}{\cterm})}{\dterm}) \geqterm
\ttag(\up{\fterm}[\avar:=\dterm]\asub)$ for some $\gamma$, and this
equals $\ttag(\up{\fterm}[\avar:=\dterm])\asub^\ttag$ by
Lemma~\ref{lem:tagbasesubstitute}; let $\bsub := \asub^\ttag$.

For the second case,
$\ttag(\aterm_i)
= \app{\app{(\abs{\avar}{\ttag_{\setop\avar\setcl}(\cterm)})}{
\ttag(\dterm)}}{\ttag(\vec{\eterm})} \geqterm
\app{\subst{\ttag_{\setop\avar\setcl}(\cterm)}{[\avar:=
\ttag(\dterm)]}}{\ttag(\vec{\eterm})}$ (since $\geqterm_1$
includes $\mathtt{beta}$), $\geqterm \app{\ttag(\subst{\cterm}{[\avar:=
\dterm]})}{\ttag(\vec{\eterm})} = \ttag(\bterm_i)$ by
Lemma~\ref{lem:tagsubstitute} and because always
$\afun^-(\vec{\avar}) \geqterm_1 \afun(\vec{\avar})$.
\end{proof}

Theorem~\ref{thm:maintheoremll} also has a counterpart: if a \llfe\ 
AFS is non-terminating, then there is a reduction triple which
satisfies the constraints of Theorem~\ref{thm:maintheoremweak} for
$\P = \P_1 = \DP$: if $(\geqterm,\geqterm_1,\gterm)$ is the reduction
triple from Theorem~\ref{thm:maintheoremll}, let $\aterm\ R'\ \bterm$
if $\mathit{untag}(\aterm)\ R\ \mathit{untag}(\bterm)$.  Then
$(\geqterm',\geqterm_1',\gterm')$ satisfies the required properties.

\emptyline
Theorem~\ref{thm:maintheoremweak} is comparable to Theorem
\ref{thm:maintheorem}, and as before, the
result is likely not immediately usable due to typing problems.
Moreover, it is not obvious that the tagged subterm property is
really weaker than the limited subterm property.
So to complete the work, we re-examine the results of
Section~\ref{subsec:typechange}.  To start, let us reconsider the
definition of $\supterm^!$.

\begin{definition}[Refinement of $\supterm^!$]\label{def:suptermS}
Let $S$ be a special
set of function symbols.  $\supterm^S$ is the relation on base-type
terms (and $\suptermeq^S$ its reflexive closure) generated by the
following clauses:
\begin{iteMize}{$\bullet$}
\item $\app{(\abs{\avar}{\aterm})}{\bterm_0} \cdots \bterm_n
  \supterm^S \cterm$ if $\app{\aterm[\avar:=\bterm_0]}{\bterm_1}
  \cdots \bterm_n \suptermeq^S \cterm$
\item $\app{\afun(\aterm_1,\ldots,\aterm_m)}{\bterm_1} \cdots
  \bterm_n \suptermeq^S \cterm$ if $\app{\aterm_i}{\vec{\c}}
  \suptermeq^S \cterm$ and $\afun \in S$\ \ \ \ 
  \small{$\longleftarrow$ \emph{here we differ from $\supterm^!$}}
\item $\app{\aterm}{\bterm_1} \cdots \bterm_n \suptermeq^S \cterm$ if
  $\app{\bterm_i}{\vec{\c}} \suptermeq^S \cterm$\ \ ($\aterm$ may
  have any form)
\end{iteMize}
\end{definition}

Note that our original definition of $\supterm^!$ is just a special
case of this definition; $\supterm^!$ can be described as
$\supterm^\F$.  For \llfe\ AFSs, we can limit ourselves to
$\supterm^{\F^-}$, shortly denoted $\supterm^-$.

In correspondence with Theorem~\ref{thm:typepreserve}, we
derive the following result:

\begin{theorem}\label{thm:weakfinish}
A set of dependency pairs $\P = \P_1 \uplus \P_2$ is tagged-chain-free
if $\P_2$ is tagged-chain-free and there is a reduction pair
$(\geqterm,\gterm)$ such that:
\begin{enumerate}[\em(1)]
\item $\overline{l} \gterm \overline{\ttag(p)}$ for all $l \dppijl
  p \in \P_1$;
\item $\overline{l} \geqterm \overline{\ttag(p)}$ for all $l \dppijl
  p \in \P_2$;
\item $l \geqterm \ttag(r)$ for all $l \arrz r \in \formrules(\P)$;
\item $\afun^-(\avar_1,\ldots,\avar_n) \geqterm \afun(\avar_1,\ldots,
  \avar_n)$ for all $\afun^- \in \F^-$;
\item \label{it:weak:collapse}
  if $\P$ is collapsing, then $(\geqterm, \gterm)$ respects
  $\supterm^-$, and $\afun^-(\vec{\avar}) \geqterm \up{\afun}(
  \vec{\avar})$ for all
  $\afun \in \Defineds$.
\end{enumerate}
\end{theorem}

\begin{proof}[Proof of Theorem~\ref{thm:weakfinish}]
Let $(\geqterm,\gterm)$ be a reduction pair satisfying the requirements
in the Theorem, and let $(\geq,\geq_1,>)$ be the reduction triple
generated by $(\geqterm,\gterm)$ as defined in
Section~\ref{subsec:typechange}.  This triple clearly satisfies the
first four requirements of Theorem~\ref{thm:maintheoremweak}.  For
the last one, let $\P$ be collapsing; we must see that $\geq$
has the tagged subterm property.

So let $\avar,\aterm,\bterm,\cterm$ be given such that $\aterm
\suptermeq \cterm \supterm \avar$.  We must see that
$\app{\app{(\abs{\avar}{\ttag_{\setop\avar\setcl}(\aterm)})}{\ttag(
\bterm)}}{\vec{\bvar}}\ (\gterm \cup \geqterm)^*\ 
\app{\ttag(\up{\cterm}[\avar:=\bterm]\asub)}{
\vec{\dterm}}$ for some substitution $\asub$ with $\avar \notin
\domain(\asub)$ and some terms $\vec{\dterm}$.
Since $\geqterm$ contains $\mathtt{beta}$ and
$\ttag_{\setop\avar\setcl}(\app{\aterm}{\vec{\bvar}}) =
\app{\ttag_{\setop\avar\setcl}(\aterm)}{\vec{\bvar}}$ it
suffices if we can prove that for all base-type terms $\fterm$ with
$\fterm \suptermeq \cterm \supterm \avar$ we have:
$\subst{\ttag_{\setop\avar\setcl}(\fterm)}{[\avar:=\ttag(\bterm)]}
\ (\gterm \cup \geqterm)^*\ 
\app{\ttag(\up{\cterm}[\avar:=\bterm]\asub)}{\vec{\dterm}}$
for some $\asub,\vec{\dterm}$.  This gives what we need by
choosing $\fterm = \app{\aterm}{\vec{\bvar}}$.
We prove this statement by induction on $\fterm$, ordered with
$\supterm^-$ (it is easy to see that this relation is well-founded).

Note (**): $\ttag_{\setop\avar\setcl}(\eterm) \geqterm \ttag(\eterm)$
by Lemma~\ref{lem:dropvars}
and because always $\bfun^-(\vec{\cvar}) \geqterm \bfun(\vec{\cvar})$.

For the base case, let $\fterm = \app{\cterm}{\dterm_1} \cdots
\dterm_n$ for some terms $\vec{\dterm}$.  Then
$\subst{\ttag_{\setop\avar\setcl}(\fterm)}{[\avar:=\ttag(\bterm)]} =
\app{\subst{\ttag_{\setop\avar\setcl}(\cterm)}{[\avar:=\ttag(\bterm
)]}}{\dterm_1'} \cdots \dterm_n'$, where each $\dterm_i' =
\subst{\ttag_{\setop\avar\setcl}(\dterm_i)}{[\avar:=\ttag(\bterm)]}$.
Since $\cterm \supterm \avar$ we know that $\ttag_{\setop\avar\setcl}(
\cterm) \geqterm \ttag(\up{\cterm})$: either $\cterm$ does not have the
form $\afun(\cterm_1,\ldots,\cterm_m)$ with $\afun \in \Defineds$, in
which case $\ttag_{\setop\avar\setcl}(\cterm) \geqterm \ttag(\cterm)
= \ttag(\up{\cterm})$ by (**), or $\cterm$ does have this form and
$\ttag_{\setop\avar\setcl}(\cterm) =
\afun^-(\ttag_{\setop\avar\setcl}(\cterm_1),\ldots,
\ttag_{\setop\avar\setcl}(\cterm_m)) \geqterm
\afun^-(\ttag(\cterm_1),\ldots,\ttag(\cterm_m)) \geqterm
\up{\afun}(\ttag(\cterm_1),\ldots,\ttag(\cterm_m)) =
\ttag(\up{\cterm})$ by assumption ($\P$ is collapsing,
so both properties from (\ref{it:weak:collapse}) hold).
Thus,
$\ttag_{\setop\avar\setcl}(\fterm)[\avar:=\ttag(\bterm)] \geqterm
\app{\subst{\ttag(\up{\cterm})}{[\avar:=\ttag(\bterm)]}}{\vec{
\dterm'}} = \app{\ttag(\up{\cterm}[\avar:=\bterm])}{\vec{\dterm'}}$.

Now to consider each of the induction cases:
\begin{enumerate}[(1)]
\item \label{ic:weakf:lam}
  $\fterm = \app{\app{(\abs{\bvar}{\aterm})}{\eterm}}{\dterm_1}
  \cdots \dterm_n$ and $\aterm \suptermeq \cterm$;
\item \label{ic:weakf:fun}
  $\fterm = \app{\afun(\eterm_1,\ldots,\eterm_m)}{\dterm_1} \cdots
  \dterm_n$ and some $\eterm_i \suptermeq \cterm$.
\item \label{ic:weakf:app}
  $\fterm = \app{\fterm'}{\dterm_1} \cdots \dterm_n$ and one of the
  $\dterm_i \suptermeq \cterm$;
\end{enumerate}
These are the only forms $\fterm$ can have.
In very general terms, each of these cases is easy because $\gterm
\cup \geqterm$ includes $\supterm^-$ (in case~\ref{ic:weakf:fun} we
use that $\FV(\afun(\vec{\eterm}))$ contains $\avar$, so the tagging
function replaces $\afun$ by $\afun^-$).  Precise derivations are
given below.

In case~\ref{ic:weakf:lam},
$\ttag_{\setop\avar\setcl}(\fterm)[\avar:=\ttag(\bterm)] \geqterm
(\app{\subst{\ttag_{\setop\avar,\bvar\setcl}(\aterm)}{[\bvar:=
\ttag_{\setop\avar\setcl}(\eterm)]}}{\ttag_{\setop\avar\setcl}(
\vec{\dterm})})[\avar:=\ttag(\bterm)]$ (since $\geqterm$ contains
$\mathtt{beta}$), $\geqterm (\app{\ttag_{\setop\avar\setcl}(\aterm
[\bvar:=\eterm])}{\ttag_{\setop\avar\setcl}(\vec{\dterm})})[\avar:=
\ttag(\bterm)]$ by Lemma~\ref{lem:tagsubstitute}, and this is exactly
$\subst{\ttag_{\setop\avar\setcl}(\app{\aterm[\bvar:=
\eterm]}{\vec{\dterm}})}{[\avar:=\ttag(\bterm)]}$.
Since $\aterm[\bvar:=\eterm] \suptermeq \cterm[\bvar:=\eterm]$ we
can use the induction hypothesis to find $\asub',\vec{\dterm'}$ such
that this term $(\gterm \cup \geqterm)^*\ 
\app{\ttag(\up{\cterm[\bvar:=\eterm][\avar:=\bterm]\asub'})}{
\vec{\dterm'}} =
\app{\ttag(\up{\cterm[\avar:=\bterm][\bvar:=\eterm[\avar:=\bterm]]
\asub'})}{
\vec{\dterm'}}$, which proves the statement for $\asub := [\bvar:=
\eterm[\avar:=\bterm]]\asub'$.

In case~\ref{ic:weakf:fun}, $\ttag_{\setop\avar\setcl}(\fterm) =
\app{\afun^-(\ttag_{\setop\avar\setcl}(\vec{\eterm}))}{\ttag_{
\setop\avar\setcl}(\vec{\dterm})}$ because $\afun(\vec{\eterm})
\supterm  \avar$, and this $(\gterm \cup \geqterm)^*\ 
\app{\ttag_{\setop\avar\setcl}(\eterm_i)}{\vec{\c}} =
\ttag_{\setop\avar\setcl}(\app{\eterm_i}{\vec{\c}})$ because
$(\gterm,\geqterm)$ respects $\supterm^-$.  We complete
by induction.

Finally, case~\ref{ic:weakf:app}.
$\ttag_{\setop\avar\setcl}(\fterm) = \app{\ttag_{\setop\avar
\setcl}(\fterm')}{\ttag_{\setop\avar\setcl}(\dterm_1)} \cdots
\ttag_{\setop\avar\setcl}(\dterm_n)\ (\gterm \cup \geqterm)^*\ 
\app{\ttag_{\setop\avar\setcl}(\dterm_i)}{\vec{\c}} =
\ttag_{\setop\avar\setcl}(\app{\dterm_i}{\vec{\c}})$, which by the
induction hypothesis $(\gterm \cup \geqterm)^*\ 
\app{\ttag(\subst{\subst{\up{\cterm}}{\asub}}{[\avar:=
\bterm]})}{\vec{\dterm}}$ as required.
\end{proof}

\begin{example}\label{ex:weaktwicefinish}
The dependency graph of $\twice$ has only one SCC, as given in
Example~\ref{ex:twicecycles}, and whose formative rules we calculated
in Example~\ref{ex:urtwice}.  Therefore, by
Theorems~\ref{thm:maintheoremweak} and~\ref{thm:weakfinish},
$\twice$ is terminating if there is a reduction pair
$(\geqterm,\gterm)$ which respects $\supterm^-$, and orients the
following constraints:
\[
\begin{array}{rclrcl}
\up{\I}(\suc(n)) & \gterm & \twice(\abs{x}{\I^-(x)}) \cdot n &
\app{\up{\twice}(F)}{x} & \gterm & F \cdot (F \cdot \c_\nat) \\
\up{\I}(\suc(n)) & \gterm & \up{\twice}(\abs{x}{\I^-(x)}) \cdot \c_\nat\ &
\up{\twice}(F) \cdot x & \gterm & F \cdot \c_\nat \\
\twice(F) \cdot m & \gterm & F \cdot (F \cdot  m) &
\twice(F) \cdot m & \gterm & F \cdot m \\
\I(\suc(n)) & \geqterm & \suc(\twice(\abs{x}{\I^-(x)}) \cdot n) &
\twice(F) \cdot m & \geqterm & F \cdot (F \cdot m) \\
\afun^-(\vec{\avar}) & \geqterm & \afun(\vec{\avar})\ (\forall \afun
  \in \F) &
\afun^-(n) & \geqterm & \up{\afun}(n)\ (\forall \afun \in \Defineds) \\
\end{array}
\]
The top six are requirements for orienting dependency pairs, the
next two are the formative rules of this SCC, and the final two are
required to deal with the marked symbols.
\end{example}

Theorem~\ref{thm:weakfinish} is a real improvement over
Theorem~\ref{thm:typepreserve} because (A) we only need to orient
the formative rules of a set of dependency pairs, (B) the requirement
that $\supterm^-$ is included in $\gterm \cup \geqterm$ is
significantly weaker than the requirement for $\supterm^!$ to be
included, and (C) the requirement that $\afun(\vec{\avar}) \geqterm
\up{\afun}(\vec{\avar})$ was replaced by the requirement that
$\afun^-(\vec{\avar}) \geqterm \up{\afun}(\vec{\avar})$, which removes
the direct relationship between a defined symbol and its marked
version.  In the next sections we will see how we can use this
increased strength.

\section{Finding a Reduction Pair}\label{sec:redpair}

\summary{In this section, we will see two different ways of finding a
reduction pair to solve the constraints generated by either the basic
or local dependency pair approach.  First, we consider how
interpretations in a so-called \emph{weakly monotonic algebra} can be
used in the dependency pair setting.  Second, we show how to use
\emph{argument functions} to alter an existing reduction pair such as
the higher-order recursive path ordering.}

Whether we use Theorem~\ref{thm:typepreserve} or
Theorem~\ref{thm:weakfinish}, we have:
\begin{iteMize}{$\bullet$}
\item a set $\Sigma$ of function symbols: in the basic case,
  $\Sigma = \up{\F}_c$, in the local case $\Sigma = \up{\F}_c \cup
  \{ \afun^- : \atype \mid \afun : \atype \in \F \}$
\item a set $S$ of ``protected'' funtion symbols: in the basic
  case, $S = \F$, in the local case $S = \{ \afun^- : \atype \mid
  \afun : \atype \in \F \}$
\item a set $A$ of constraints of the form $l \geqorgterm p$ and a
  set $B$ of constraints $l \geqterm r$: in the basic case these are
  given by the dependency pairs and the rules, in the local case the
  right-hand side is adapted with $\ttag$ and $B$ contains only the
  formative rules.
\end{iteMize}
In both cases, we must find a reduction pair $(\geqterm,\gterm)$ such
that $l \gterm p$ for at least one of the constraints in $A$,\ 
$l \geqterm p$ for the remainder of them, and $l \geqterm r$ for the
constraints in $B$.  Moreover, the reduction pair may have to respect
$\supterm^S$, which is definitely the case if:
\begin{iteMize}{$\bullet$}
\item $\app{\aterm}{\bterm_1} \cdots \bterm_n \geqterm \app{\bterm_i}{
  \vec{\c}}$ if both sides have base type;
\item $\app{\afun(\aterm_1,\ldots,\aterm_m)}{\vec{\bterm}} \geqterm
  \app{\aterm_i}{\vec{\c}}$ if $\afun \in S$ and both sides have base
  type
\end{iteMize}
We will consider two ways of finding a reduction pair which uses the
possibilities created by the dependency pair approach.  First, we
shall consider \emph{weakly monotonic algebras}, where we natively
have a quasi-ordering $\geqterm$ which is not just the reflexive
closure of $\gterm$, and which is not a simplification ordering.
Second, we show how an existing reduction pair can be modified with
\emph{argument functions}, a generalisation of the notion of argument
filterings.

\subsection{Weakly Monotonic Algebras}\label{sec:monalg}

A semantic method to prove termination of first-order term rewriting
is to interpret terms in a well-founded algebra, 
such that whenever $\aterm \arrz \bterm$, 
their interpretations in the algebra decrease: $\algint{\aterm} >
\algint{\bterm}$.
Such an algebra is called a \emph{termination model}
if $\algint{l} > \algint{r}$ for every rewrite rule $l \arrz r$,
and some additional properties
guarantee that this implies $\algint{C[l\asub]} > \algint{C[r\asub]}$
for all contexts $C$ and substitutions $\asub$.
A first-order TRS is terminating if and only if 
it has a termination model~\cite{hue:opp:80:1,zan:94:1}.
In his PhD thesis~\cite{pol:96:1}, van de Pol (extending on a joint
paper with Schwichtenberg~\cite{pol:sch:95:1}), generalises this
approach to HRSs, with higher-order rewriting modulo
$\alpha\beta\eta$, and shows that a HRS is terminating if it has a
termination model; the converse does not hold.

Here we consider interpretations of AFS terms in a weakly monotonic
algebra, and use these
to solve dependency pair
constraints.  Since $>$ does not have to be monotonic when using
dependency pairs, the theory of \cite{pol:96:1} can be significantly
simplified.  

\paragraaf{Type Interpretation}
In first-order algebra interpretations, all terms are mapped to an
element of some well-founded set $\basealgebra$.  In the higher-order
setting it turns out to be impractical to map all terms to the same
set.  Rather, terms of a type $\atype \typepijl \btype$ are mapped to
functions.

As a basis,
let $\basealgebra = (\basealgebraset, \vee,
0, >)$\footnote{$\basealgebra$ might be a well-ordered monoid or a
join-semilattice, but both of these have requirements on $\vee$ which
we will not need; with our definitions, $\vee$ may be a supremum
operator, but also for instance addition.}, where $\basealgebraset$
is a set,
$\vee$ a binary operator on $\basealgebraset$, $0$ an element of
$\basealgebraset$ and $>$ a well-founded partial order on
$\basealgebraset$ (with reflexive closure $\geq$), such that:
\begin{iteMize}{$\bullet$}
\item $0$ is a minimum element, so $a \geq 0$ for all $a \in
  \basealgebraset$;
\item $x \vee y \geq x$ and $x \vee y \geq y$ for all $x,y \in \basealgebraset$
\item $x \vee 0 = x$ for all $x,y \in \basealgebraset$
\end{iteMize}

\medskip \noindent
To each type $\atype$ we associate a set $\WM_\atype$ of
\emph{weakly monotonic functionals} and two relations $\gwm^\atype$
and $\geqwm^\atype$, defined inductively as follows.

For a base type $\abasetype$:  
\begin{iteMize}{$\bullet$}
\item 
  $\WM_\abasetype = \basealgebraset$,
\item 
  $\gwm^\abasetype \mathord{=} >$,
  and $\geqwm^\abasetype \mathord{=} \geq$
\end{iteMize}

For a functional type $\atype \ftypepijl \btype$:
\begin{iteMize}{$\bullet$}
\item
$\WM_{\atype \ftypepijl \btype}$ 
consists of the functions from 
  $\WM_\atype$ to $\WM_\btype$,
  such that $\geqwm$ is preserved
  (that is, if $x \geqwm^\atype y$ then
  $f(x) \geqwm^\btype f(y)$),
\item  
  $f \gwm^{\atype \ftypepijl \btype} g$ iff $f(x) \gwm g(x)$ for all
  $x \in \WM_\atype$,
\item
  $f \geqwm^{\atype \ftypepijl \btype} g$ iff $f(x) \geqwm g(x)$ for
  all $x \in \WM_\atype$. 
\end{iteMize}

\medskip \noindent
Thus, $\WM_{\atype \ftypepijl \btype}$ is a subset of the function
space $\WM_\atype \rightarrow \WM_\btype$, consisting of 
functions which preserve $\geqwm$.
$\gwm^\atype$ and $\geqwm^\atype$ are an order and
quasi-order respectively for all types $\atype$, and they are
compatible.
We
usually omit the type $\atype$ in the notation, and 
write just $\gwm$ and $\geqwm$.
If either $x \gwm y$ or $x = y$ then 
$x\geqwm y$, but the converse implication does not hold.

This definition differs from the one in~\cite{pol:96:1} in that we use
only one set $\basealgebraset$ rather than different sets
$\basealgebraset_\abasetype$ for all base types $\abasetype$; this is
done because we must occasionally compare terms of different types.
The original definition also does not use an operator $\vee$ (but
does use a $\oplus$, which satisfies the requirements); this we added
because we work with AFSs rather than HRS, and it is convenient to
have a ``maximum'' function to interpret application.

\paragraaf{Term Interpretation} For some valuable background, let us
first consider the most relevant definitions and results
from~\cite{pol:96:1}.  In the HRS formalism considered there, function
symbols do not have an arity; they come equipped with a type, rather
than a type declaration.  An \emph{interpretation function}
$\constvaluation$ on the signature is used to associate to each
closed term a weakly monotonic functional.
Let $\aterm$ be a $\lambda$-term, $\constvaluation_\afun$ an element
of $\WM_\atype$ for all $\afun : \atype \in \F$, and $\varvaluation$
a \emph{valuation} which assigns to all variables $\avar : \atype \in
\FV(\aterm)$ an element of $\WM_\atype$.  Then $\lamalgintc{\aterm}$
is defined by the following clauses:
\[
\begin{array}{llll}
\lamalgintc{\afun} & = & \constvaluation_\afun & \mathrm{if}\ \afun
  : \atype \in \F \\
\lamalgintc{\avar} & = & \varvaluation(\avar) &
  \mathrm{if}\ \avar \in \setvar \\
\lamalgintc{\app{\aterm}{\bterm}} & = &
  \lamalgintc{\aterm}(\lamalgintc{\bterm}) \\
\lamalgintc{\abs{\avar}{\aterm}} & = & \fatlambda n.\lamalgint{\aterm
  }_{\constvaluation,\varvaluation \cup \{\avar \mapsto n \}} &
  \mathrm{if}\ \avar \notin \domain(\varvaluation)\ \ 
  \text{(always applicable with $\alpha$-conversion)} \\
\end{array}
\]
Here, $\fatlambda$ denotes function construction: $\fatlambda x_1
\ldots x_n.P(x_1,\ldots,x_n)$ is the function which takes $n$
arguments $x_1,\ldots,x_n$ and returns $P(x_1,\ldots,x_n)$.

\begin{lemma}\label{lem:algintfacts}
Some facts on these interpretations: 
\begin{enumerate}[\em(1)]
\item \emph{(Substitution Lemma)}
  \label{lem:algintfacts:substitution}
  Given a substitution $\asub = [\avar_1:=\aterm_1,\ldots,\avar_n:=
  \aterm_n]$ and a valuation $\varvaluation$ whose domain does not
  include the $\avar_i$: $\lamalgintc{\aterm\asub} =
  \lamalgint{\aterm}_{\constvaluation,\varvaluation \circ \asub}$.
  Here, $\varvaluation \circ \asub$ is the valuation $\varvaluation
  \cup \{ \avar_1 \mapsto \lamalgintc{\aterm_1}, \ldots,\avar_n
  \mapsto \lamalgintc{\aterm_n} \}$.
\item\label{lem:algintfacts:interprete}
  If $\aterm : \atype$ is a term, then $\lamalgintc{\aterm} \in
  \WM_\atype$ for all valuations $\varvaluation$.
\end{enumerate}
\end{lemma}
Lemma~\ref{lem:algintfacts}(\ref{lem:algintfacts:interprete})
provides a convenient way to find weakly monotonic functionals.
For example, 
For $n \in \basealgebraset$ and $\atype = \btype_1 \typepijl \ldots
\typepijl \btype_k \typepijl \abasetype$ the \emph{constant function}
$n_\atype = \fatlambda x_1 \ldots x_k.n \in \WM_\atype$ (as $n_\atype
= [\abs{\avar_1 \ldots \avar_n}{\bvar}]_{\{\bvar \mapsto n\}}$).
Similarly, writing $f(\vec{0})$ for $f(0_{\btype_1},\ldots,0_{
\btype_n})$, the function $\fatlambda f.f(\vec{0})$ is in
$\WM_{\atype \ftypepijl \abasetype}$.
A weakly monotonic functional not given in \cite{pol:96:1}, but
which will be needed to deal with term application, is $\max_\atype$,
defined as follows: \\
$\begin{array}{lrcll}
\ \ \ \ \ &
\max_\abasetype(x,y) & = & x \vee y & (\mathrm{for}\ x,y \in \basealgebraset)\\
\ \ \ \ \ &
\max_{\atype \ftypepijl \btype}(f,y) & = &
  \fatlambda x.\max_{\btype}(f(x),y) & (\mathrm{for}\ f \in
  \WM_{\atype \ftypepijl \btype}, y \in \basealgebra) \\
\end{array}$ \\
Using induction on the type of the first argument (and once more
Lemma~\ref{lem:algintfacts}(\ref{lem:algintfacts:interprete})),
it is easy to see that
$\max_{\atype} \in \WM_{\atype \ftypepijl \abasetype
\ftypepijl \atype}$.

\medskip
In HRSs, terms are $\alpha \beta \eta$-equivalence classes, so
$\lamalgintc{\app{(\abs{\avar}{\aterm})}{\bterm}} =
\lamalgintc{\aterm}$ if $\avar \notin \FV(\aterm)$, which is not very
convenient in the present setting of AFSs, where terms are considered
modulo $\alpha$ only.
To adapt the results, we can think of application as a function
symbol.

\begin{definition}\label{def:interpretation}
A \emph{signature interpretation} associates a weakly monotonic
functional $\constvaluation_\afun \in \WM_{\atype_1 \ftypepijl \ldots
\ftypepijl \atype_n \ftypepijl \btype}$ to every function symbol
$\afun : [\atype_1 \times \ldots \times \atype_n] \decpijl \btype$
of the signature.
The pair $(\basealgebra, \constvaluation)$ is a \emph{weakly montonic
algebra}.
A \emph{valuation} is a function $\varvaluation$ with a
finite domain of variables, such that $\varvaluation(\avar) \in
\WM_\atype$ for every variable $\avar : \atype$ in its domain. 

Fixing $\constvaluation$ and $\varvaluation$,
the \emph{interpretation}  of a term $\aterm$,
denoted $\algintc{\aterm}$,
is a weakly monotonic functional defined by induction on 
the definition of terms as follows:
\[
\begin{array}{lll}
\algintc{\avar} & = & \varvaluation(\avar)\ \ \mathrm{if}\ \avar
  \in \setvar \\
\algintc{\afun(\aterm_1,\ldots,\aterm_n)} & = &
  \constvaluation_\afun(\algintc{\aterm_1},\ldots,\algintc{\aterm_n})
  \\
\algintc{\abs{\avar}{\aterm}} & = & \fatlambda n.
  \algint{\aterm}_{\constvaluation,\varvaluation \cup \{\avar \mapsto
  n \}}\ \ 
  \mathrm{if}\ \avar \notin \domain(\varvaluation) \\
\algintc{\app{\aterm}{\bterm}} & = & \max(\algintc{\aterm}(
  \algintc{\bterm}), \algintc{\bterm}(\vec{0})) \\
\end{array}
\]
\end{definition}
\noindent
We assume $\varvaluation$ is defined on all free variables
of $\aterm$.  This definition corresponds with the one for HRSs,
if we replace application by a function symbol which is
interpreted with $\max$.

\begin{example}\label{ex:algebratwice}
In our running example $\twice$, consider an interpretation in the
natural numbers.
Say $\constvaluation_\I = \fatlambda n.n$ and
$\constvaluation_\suc = \fatlambda n.n+1$.  Then $\algintc{\I(\suc(
\avar))} = \varvaluation(\avar)+1$.
\end{example}

\sparagraaf{Reduction Pair.} 
Since $\gwm$ is in general not closed under taking contexts, it
cannot be used directly like in first-order rewriting: $\algintc{l}
\gwm \algintc{r}$ does not in general imply $\algintc{C[l\asub]} \gwm
\algintc{C[r\asub]}$.  This issue (which Van De Pol handles by
imposing the restriction that $\constvaluation_\afun$ must be
\emph{strict}) disappears entirely in the context
of dependency pairs.

\begin{theorem}\label{thm:monalg}
Let $\constvaluation$ be an interpretion of the signature 
$\Sigma$\footnote{$\Sigma$ is the signature introduced at the start
of Section~\ref{sec:redpair}, either $\up{\F}_c$ or $\up{\F}_c \cup
\{\afun^- : \atype \mid \afun : \atype \in \F \}$.} such that:
\begin{iteMize}{$\bullet$}
\item $\constvaluation$ maps each $\c_\atype$ to $0_\atype$;
\item for all $\afun : [\atype_1 \times \ldots \times \atype_m]
  \typepijl \btype_1 \typepijl \ldots \typepijl \btype_k \typepijl
  \abasetype \in S$, all $1 \leq i \leq m$ and all $n \in
  \WM_{\atype_i}$: $\constvaluation_\afun(0_{\atype_1},\ldots,n,
  \ldots,0_{\atype_m},0_{\btype_1},\ldots,0_{\btype_k}) \geqwm
  n(\vec{0})$.
\end{iteMize}
Define $\aterm \geqterm \bterm$ if
$\algintc{\aterm} \geqwm \algintc{\bterm}$ for all valuations
$\varvaluation$, and $\aterm \gterm \bterm$ if
$\algintc{\aterm} \gwm \algintc{\bterm}$ for all valuations
$\varvaluation$.
Then $(\geqterm,\gterm)$ is a reduction pair which respects
$\supterm^S$.
\end{theorem}

\begin{proof}
It is easy to see that $\gterm$ is an ordering and $\geqterm$ a
quasi-ordering; $\gterm$ is well-founded because if $f \gwm g$ then
also $f(\vec{0}) \gwm f(\vec{0})$ (and $\gwm$ is well-founded in
$\basealgebra$).  $\geqterm$ is monotonic as we see with a simple
case distinction; if $\algintc{\aterm} \geqwm \algintc{\bterm}$ for
all $\varvaluation$, then:
\begin{iteMize}{$\bullet$}
\item $\algintc{\abs{\avar}{\aterm}} = \fatlambda n.
  \algint{\aterm}_{\constvaluation,\varvaluation \cup \setop x
  \mapsto n\setcl} \geqwm
  \fatlambda n.\algint{\bterm}_{\constvaluation,\varvaluation \cup
  \setop x \mapsto n\setcl} = \algintc{\abs{\avar}{\bterm}}$
\item $\algintc{\app{\aterm}{\cterm}} \geqwm \algintc{\app{\bterm}{
  \cterm}}$ by weak monotonicity of
  $\fatlambda f.\fatlambda n.\max(f(n),n(\vec{0}))$
\item $\algintc{\app{\cterm}{\aterm}} \geqwm \algintc{\app{\cterm}{
  \bterm}}$ by weak monotonicity of
  $\fatlambda f.\fatlambda n.\max(f(n),n(\vec{0}))$
\item $\algintc{\afun(\ldots,\aterm,\ldots)} \geqwm \algintc{\afun(
  \ldots,\bterm,\ldots)}$ by weak monotonicity of
  $\constvaluation_\afun$
\end{iteMize}
In addition, $\geqterm$ contains $\cbeta$ because $\algintc{\app{(\abs{
\avar}{\aterm})}{\bterm}} = \max(\algint{\aterm}_{\constvaluation,
\varvaluation \circ [\avar\mapsto\algintc{\bterm}]},\algintc{\bterm}(
\vec{0})) \geqwm \algint{\aterm}_{\constvaluation,\varvaluation
\circ [\avar\mapsto\algintc{\bterm}]}$, which is exactly $\algintc{
\aterm[\avar:=\bterm]}$ by the substitution Lemma.
Compatibility is inherited from compatibility of $\gwm$ and $\geqwm$
on $\basealgebra$ (where $\geqwm$ is the reflexive closure of $\gwm$):
if $\aterm \gterm \bterm \geqterm \cterm$, then $\algintc{\aterm} =
\fatlambda n_1 \ldots n_k.f(\vec{n}) \gwm \algintc{\bterm} =
\fatlambda \vec{n}.g(\vec{n}) \geqwm \algintc{\cterm} = \fatlambda
\vec{n}.h(\vec{n})$, and we are done because also $f(\vec{n}) \gwm
h(\vec{n})$.

Finally, stability follows by the substitution Lemma: always
$\algintc{\cterm\asub} = \algint{\cterm}_{\constvaluation,
\varvaluation \circ \asub}$, so if $\algintc{\aterm} \gwm
\algintc{\bterm}$ for all valuations $\varvaluation$, this also holds
for the valuation $\varvaluation' = \varvaluation \circ \asub$.

As observed at the start of Section~\ref{sec:redpair},
$(\geqterm,\gterm)$ respects $\supterm^S$ if:
\begin{iteMize}{$\bullet$}
\item $\algintc{\app{\aterm}{\vec{\bterm}}} \geqwm
\algintc{\app{\bterm_i}{\vec{\c}}}$; this holds
because, by the use of $\max$ for
applications,
$\algintc{\app{\aterm}{\vec{\bterm}}} \geqwm \algintc{\bterm_i}(
\vec{0}) = \max(\ldots \max(\algintc{\bterm_i}(\vec{0}),0_{\atype_1}),
\ldots, 0_{\atype_n}) = \algintc{\app{\bterm_i}{\vec{\c}}}$;
\item $\algintc{\app{\afun(\vec{\aterm})}{\vec{\bterm}}} \geqwm
\algintc{\app{\aterm_i}{\vec{\c}}}$ if $\afun \in S$; this holds
because
$\algintc{\app{\afun(\vec{\aterm})}{\vec{\bterm}}} \geqwm
\constvaluation_\afun(\algintc{\aterm_1},\newline
\ldots,\algintc{\aterm_n},
\algintc{\bterm_1},\ldots,\algintc{\bterm_m})$, which by weak
monotonicity and assumption
$\geqwm \constvaluation_\afun(0_{\atype_1},\ldots,
\algintc{\aterm_i},\ldots, 0_{\atype_n}, 0_{\btype_1}, \ldots,
0_{\btype_m}) \geqwm \algintc{\aterm_i}(
\vec{0}) = \algintc{\app{\aterm_i}{\vec{\c}}}$.
\end{iteMize}\vspace{-\baselineskip}
\end{proof}

\noindent
Although it is not in general possible to (automatically) determine
whether a suitable interpretation exists, one could for instance try
polynomial interpretations over the natural numbers.  Since addition
and multiplication are both weakly monotonic, such an interpretation
is sound, and like in first-order rewriting there are some easily
automatable techniques for finding suitable polynomials.  The
automation of individual reduction pairs is beyond the scope of this
paper, however; we refer to~\cite{fuh:kop:12:1} for a more detailed
discussion.

\begin{example}\label{ex:twicemono}
To prove that $\scctwice$ is tagged-chain-free, write $\scctwice =
\P_1 \uplus \P_2$, where $\P_1$ consists of the two $\up{I}$
dependency pairs, and $\P_2$ of the remainder.  By
Theorem~\ref{thm:maintheoremweak} it suffices to prove that $\P_2$ is
tagged-chain-free, if we can find a reduction pair which respects
$\supterm^-$ and orients the following requirements:
\[
\begin{array}{rclrcl}
\up{\I}(\suc(n)) & \gterm & \twice(\abs{x}{\I^-(x)}) \cdot n &
\app{\up{\twice}(F)}{x} & \geqterm & F \cdot (F \cdot \c_\nat) \\
\up{\I}(\suc(n)) & \gterm & \up{\twice}(\abs{x}{\I^-(x)}) \cdot \c_\nat\ &
\up{\twice}(F) \cdot x & \geqterm & F \cdot \c_\nat \\
\twice(F) \cdot m & \geqterm & F \cdot (F \cdot  m) &
\twice(F) \cdot m & \geqterm & F \cdot m \\
\I(\suc(n)) & \geqterm & \suc(\twice(\abs{x}{\I^-(x)}) \cdot n) &
\twice(F) \cdot m & \geqterm & F \cdot (F \cdot m) \\
\afun^-(\vec{\avar}) & \geqterm & \afun(\vec{\avar})\ (\forall \afun
  \in \F) &
\afun^-(n) & \geqterm & \up{\afun}(n)\ (\forall \afun \in \Defineds) \\
\end{array}
\]
Let:
\begin{iteMize}{$\bullet$}
\item $\constvaluation_\I = \constvaluation_{\up{\I}} =
  \constvaluation_{\I^-} = \fatlambda n.n$
\item $\constvaluation_\nul = \constvaluation_{\nul^-} = 0$
\item $\constvaluation_\suc = \constvaluation_{\suc^-} =
  \fatlambda n.n$
\item $\constvaluation_\twice = \constvaluation_{\up{\twice}} =
  \constvaluation_{\twice^-} = \fatlambda f n.f(f(n))$
\end{iteMize}
Then it is clear that the bottom two constraints are satisfied.  With
some calculation we can see that the others hold as well, using the
following case distinction for the $\twice$ cases: if $m \geq F(m)$,
then also $m \geq F(F(m))$ by weak monotonicity; if $F(m) \geq m$,
then also $F(F(m)) \geq m$.  Therefore $\max(F(F(m)),m) \geq
\max(F(\max(F(m),m)),\max(F(m),m))$.

To complete the termination proof of $\twice$, it suffices to
find a reduction pair such that:
\[
\begin{array}{rclrcl}
\app{\twice(F)}{n} & \gterm & \app{F}{(\app{F}{n})}\ \ \ \  &
\app{\up{\twice}(F)}{n} & \gterm & \app{F}{(\app{F}{\c_\nat})} \\
\app{\twice(F)}{n} & \gterm & \app{F}{n} &
\app{\up{\twice}(F)}{n} & \gterm & \app{F}{\c_\nat} \\
\end{array}
\]
This is satisfied with an interpretation with
$\constvaluation_{\up{\twice}} = \constvaluation_\twice =
\fatlambda f n.\max(f(f(n)),n)+1$.
\end{example}

\begin{example}\label{ex:goodmap}
A well-known example of 
higher-order rewriting is $\map$:
\[
\begin{array}{rcl}
\map(F,\nil) & \arrz & \nil \\
\map(F,\cons(h,t)) & \arrz & \cons(\app{F}{h},\map(F,t)) \\
\end{array}
\]
Applying Theorem~\ref{thm:maintheoremweak}
to prove termination,
it suffices to
find a reduction pair $(\geqterm,\gterm)$ with:
\[
\begin{array}{rcl}
\up{\map}(F,\cons(h,t)) & \gterm & \app{F}{h} \\
\up{\map}(F,\cons(h,t)) & \gterm & \up{\map}(F,t) \\
\map(F,\cons(h,t)) & \geqterm & \cons(\app{F}{h},\map(F,t)) \\
\end{array}
\]
Note that the elements of $S$ do not occur in the rules, so we can
pretty much ignore them (see also Section~\ref{sec:algorithm}).
Using an interpretation in the natural numbers with the usual greater
than, $0$ and max-operator, consider the following
constant interpretation: $\constvaluation_{\up{\map}} = \fatlambda
f.\fatlambda x.f(x)+x,\ \constvaluation_\map = \fatlambda f.
\fatlambda x.x\cdot f(x)+x,\ \constvaluation_\cons = \fatlambda x.
\fatlambda y.x+y+1$.  Taking $\varvaluation = [F:=f,h:=n,t:=m]$, the
requirements above become:
\[
\begin{array}{rcl}
f(n+m+1)+n+m+1 & > & \max(f(n),n) \\
f(n+m+1)+n+m+1 & > & f(m)+m \\
(n+m+1)\cdot f(n+m+1) +n+m+1& \geq & \max(f(n),n) + m \cdot f(m) + m
  + 1 \\
\end{array}
\]
Taking into account that $f$ must be a weakly monotonic functional,
we of course have $f(n+m+1) \geq f(n), f(m)$.  Thus, it is not hard
to see that all requirements are true.
\vspace{-6pt}
\end{example}

\subsection{Argument Functions}\label{sec:argfun}

In Section~\ref{subsec:fo:argfil} we saw that, in the first-order
dependency pair approach, (simplification) orderings may be combined
with \emph{argument filterings}.  Such filterings either eliminate
some direct arguments $\aterm_i$ from a term $\afun(\aterm_1,\ldots,
\aterm_n)$, or replace the term by one of the $\aterm_i$.  We
consider an extension of this technique, called \emph{argument
functions}.

Let $\Sigma$ and $S$ be sets of function symbols, as introduced at
the start of Section~\ref{sec:redpair}.

\begin{definition}[Argument Function]
Let $\Fex$ be a set of function symbols, and $\argfil$ a
type-respecting function mapping terms $\afun(\avar_1,\ldots,\avar_n)$
with $\afun \in \Sigma$ to some term over $\Fex$; we require that
$\FV(\argfil(\afun(\vec{\avar}))) \subseteq \{\vec{\avar}\}$.
We extend $\argfil$ to a function $\filter(\ )$,
called an \emph{argument function},
on all terms as follows:
\[
\begin{array}{lll}
\filter(\app{\aterm}{\bterm}) & = & \app{\filter(\aterm)}{
  \filter(\bterm)} \\
\filter(\abs{\avar}{\aterm}) & = & \abs{\avar}{\filter(\aterm)} \\
\filter(\avar) & = & \avar\ \quad (\avar \in \setvar) \\
\filter(\afun(\aterm_1,\ldots,\aterm_n)) & = & 
  \argfil(\afun(\avar_1,\ldots,\avar_n))[\avar_1:=
  \filter(\aterm_1),\ldots,\avar_n:=\filter(\aterm_n)]
\end{array}
\]
\end{definition}\vspace{5 pt}

\noindent
An \emph{argument filtering} is an argument function where each
$\pi(\afun(\avar_1,\ldots,\avar_n))$ has the form
$\afun'(\avar_{i_1},\ldots,\avar_{i_k})$, or $\avar_i$.
However, we can choose $\argfil$ entirely different as well.
For example,
if
$\argfil(\twice(\avar)) = \abs{\bvar}{
\app{\avar}{\bvar}}$,
then the term $\app{\twice(\abs{n}{\I(n)})}{\nul}$
is mapped to $\app{(\abs{\bvar}{\app{(\abs{n}{\I(n)})}{\bvar}})}{
\nul}$.

An argument function $\argfil$ has the \emph{$S$-subterm property} if
$\FV(\argfil(\afun(\vec{\avar}))) = \{\vec{\avar}\}$ for all $\afun
\in S$.

\begin{theorem}\label{thm:argfunrespectbeta}
Let $(\geqterm,\gterm)$ be a reduction pair on terms over $\Fex$.
Define $\geq,>$ on terms over $\Sigma$ as follows:
$\aterm \geq \bterm$ iff $\filter(\aterm) \geqterm \filter(\bterm)$
and $\aterm > \bterm$ iff $\filter(\aterm) \gterm \filter(\bterm)$.

Then $(\geq,>)$ is a reduction pair.
If $\argfil$ has the $S$-subterm property and $(\geqterm,\gterm)$
respects $\supterm^!$ and $\aterm \geqterm \c_\atype$ for all terms
$\aterm$ of type $\atype$, then $(\geq,>)$ respects $\supterm^S$.
\end{theorem}

\begin{proof}
We first make the following observation (**): $\filter(\subst{\aterm}{
\asub}) = \subst{\filter(\aterm)}{\asub^\argfil}$, where
$\asub^\argfil(\avar) = \filter(\asub(\avar))$.  This holds by
induction on the form of $\aterm$.  The only non-obvious case is when
$\aterm = \afun(\aterm_1,\ldots,\aterm_n)$ and $\argfil(\afun(\avar_1,
\ldots,\avar_n)) = p$; then $\filter(\subst{\aterm}{\asub}) =
p[\avar_1:=\filter(\subst{\aterm_1}{\asub}),\ldots,\avar_n:=
\filter(\subst{\aterm_n}{\asub})] = p[\avar_1:=
\subst{\filter(\aterm_1)}{\asub^\argfil},\ldots,
\avar_n:=\subst{\filter(\aterm_n)}{\asub^\argfil}]$ by the induction
hypothesis, $= \subst{\filter(\aterm)}{\asub^\argfil}$
because $\FV(p) \subseteq \{\vec{\avar}\}$.

Having this, it is easy to see that $\geq$ and $>$ are both stable,
and compatibility, well-foundedness, transitivity and
(anti-)reflexivity are inherited from the corresponding properties of
$\geqterm$ and $\gterm$.  As for monotonicity of $\geq$, the only
non-trivial question is whether $\afun(\aterm_1,\ldots,\aterm_n) \geq
\afun(\aterm_1',\ldots,\aterm_n')$ when each $\aterm_i \geq
\aterm_i'$, but this is clear by monotonicity of $\geqterm$: if each
$\filter(\aterm_i) \geqterm \filter(\aterm_i')$, then
$\filter(\afun(\vec{\aterm})) = \argfil(\afun(\vec{\avar}))[
  \vec{\avar}:=\filter(\vec{\aterm})] \geqterm
\argfil(\afun(\vec{\avar}))[ \vec{\avar}:=\filter(\vec{\aterm'})] =
\filter(\afun(\vec{\aterm'})) $.  Also $\geq$ contains $\mathtt{beta}$
by (**): $\filter(\app{(\abs{ \avar}{\aterm})}{\bterm}) =
\app{(\abs{\avar}{\filter(\aterm)})}{ \filter(\bterm)} \geqterm
\filter(\aterm)[\avar:=\filter(\bterm)]$ (as $\geqterm$ contains
$\mathtt{beta}$), and this equals $\filter(\aterm[ \avar:=\bterm])$ as
required.  Thus, $(\geq,>)$ is a reduction pair.

As observed at the start of Section~\ref{sec:redpair}, to see that
$(\geq,>)$ respects $\supterm^S$ it suffices if
$\app{\aterm}{\bterm_1} \cdots \bterm_n \geq \app{\bterm_i}{\vec{\c}}$
and $\app{\afun(\aterm_1,\ldots,\aterm_m)}{\bterm_1} \cdots \bterm_n
\geq \app{\aterm_i}{\vec{\c}}$ if both sides have base type, and
$\afun \in S$.  Suppose $(\geqterm,\gterm)$ respects $\supterm^!$ and
$\argfil$ has the S-subterm property.
Then obviously $\filter(\app{\aterm}{\bterm_1} \cdots \bterm_n) =
\app{\filter(\aterm)}{\filter(\bterm_1)} \cdots \filter(\bterm_n)
\geqterm \app{\filter(\bterm_i)}{\vec{\c}} =
\filter(\app{\bterm_i}{\vec{\c}})$.

For the second part, let $\afun \in S$ and $\fterm = \app{\afun(
\vec{\aterm})}{\vec{\bterm}}$; we must see that $\filter(\fterm) =
\app{\argfil(\afun(\avar_1,\ldots,\avar_m)[\avar_1\linebreak:=\filter(\aterm_1),
\ldots,\avar_m:=\filter(\aterm_m)]}{\filter(\bterm_1)} \cdots
\filter(\bterm_n) \geqterm \filter(\app{\aterm_i}{\vec{\c}})$.
Since $\argfil$ has the S-subterm property,
$\avar_i$\linebreak
occurs in $\argfil(\afun(\vec{\avar}))$.
Therefore, $\filter(\aterm_i)$ is a subterm of $\filter(\fterm)$ and
it contains no free variables which \linebreak are bound in
$\filter(\fterm)$.
But then we must have $\filter(\app{\afun(\vec{\aterm})}{\vec{\bterm}
}) \suptermeq^! \app{\filter(\aterm_i)}{\vec{\cterm}}$ for some terms
$\vec{\cterm}$; since the $\c_\atype$ are minimal elements by
assumption, this term $\geqterm \app{\filter(\aterm_i)}{\vec{\c}} =
\filter(\app{\cterm_i}{\vec{\c}})$ as required.
\end{proof}

Theorem \ref{thm:argfunrespectbeta} allows us to modify terms before
applying an ordering.  Even argument functions which respect
the $\F$-subterm property can be useful; for example, if there is a
rule
$\afun(\avar_1,\ldots,\avar_n) \arrz r$, then an argument
function with $\argfil(\afun(\vec{\avar})) = r$ is probably a good
idea.

\paragraaf{Using Argument Functions With CPO}
Unfortunately, Theorem~\ref{thm:argfunrespectbeta} cannot directly be
used with the higher-order recursive path ordering or the more
powerful computability path ordering: these relations are not
transitive, do not respect $\supterm^!$, and do not have minimal
elements.  But we \emph{can} alter either relation a bit, and obtain
a usable reduction pair.

Let $>_{\mathsf{CPO}}$ be the relation given by
CPO~\cite{bla:jou:rub:08:1}, with some fixed status function, type
ordering, and precedence $\geq_F$ on the symbols in $\Sigma \setminus
\Constants$.  Introducing new symbols $@^{\atype,\btype},
\Lambda^{\atype,\btype}$ and $\c_\atype$ for all types, extend the
precedence with $\afun >_F @^{\atype,\btype}, \Lambda^{\atype,\btype},
\c_\atype$ if $\afun \in \Sigma \setminus \Constants$, for all types
$@^{\atype,\btype} >_F \Lambda^{\atype',\btype'} >_F \c_\ctype$ and
moreover $@^{\atype,\btype} >_F @^{\ctype,\dtype}$ if $\ctype
\typepijl \dtype$ is a strict subtype of $\atype \typepijl \btype$.
Define $(\geqtermcpo,\gtermcpo)$
as follows: $\aterm \geqtermcpo \bterm$ if for all closed
substitutions $\asub$: $\mu(\aterm\asub) >_{\mathsf{CPO}}^*
\mu(\bterm\asub)$, and $\aterm \gtermcpo \bterm$ if for all closed
substitutions $\asub$: $\mu(\aterm\asub) >_{\mathsf{CPO}}^+
\mu(\bterm\asub)$.  Here, $\mu$ is given by:
\[
\begin{array}{rcll}
\mu(\avar) & = & \avar & \mathrm{if}\ \avar \in \setvar \\
\mu(\abs{\avar}{\aterm}) & = & \Lambda^{\atype,\btype}(
  \abs{\avar}{\mu(\aterm)}) & \mathrm{if}\ \abs{\avar}{\aterm} :
  \atype \typepijl \btype \\
\mu(\app{\aterm}{\bterm}) & = & @^{\atype,\btype}(\mu(\aterm),
  \mu(\bterm)) & \mathrm{if}\ \aterm : \atype \typepijl \btype \\
\mu(\afun(\aterm_1,\ldots,\aterm_n)) & = & \afun(\mu(\aterm_1),
  \ldots,\mu(\aterm_n)) \\
\end{array}
\]
Since $\mu(\aterm)[\avar:=\mu(\bterm)] = \mu(\aterm[\avar:=\bterm])$
(as demonstrated by an easy induction),
it is not hard to see that $(\geqtermcpo,\gtermcpo)$ is indeed a
reduction pair, that the $\c_\atype$ are minimal elements, and that
we can ignore the $\asub$: since $>_{\mathsf{CPO}}$ is stable,
$\aterm \gtermcpo \bterm$ holds if
$\mu(\aterm) >_{\mathsf{CPO}}^+ \mu(\bterm)$.

Moreover, always $\app{\aterm}{\bterm_1} \cdots \bterm_n \geqtermcpo
\app{\bterm_i}{\vec{\c}}$ and $\app{\afun(\aterm_1,\ldots,\aterm_n)}{
\vec{\bterm}} \geqtermcpo \app{\aterm_i}{\vec{\c}}$:
\begin{iteMize}{$\bullet$}
\item In CPO, $\bfun(\cterm,\dterm) >_{\mathsf{CPO}} \eterm$ if
  $\cterm >_{\mathsf{CPO}} \eterm$, regardless of type differences.
  Since the $@^{\atype,\btype}$ are function symbols, we thus have:
  $\mu(\app{\cterm}{\dterm}) >_{\mathsf{CPO}} \eterm$ if
  $\mu(\cterm) >_{\mathsf{CPO}} \eterm$.

  Let $\bterm_i : \atype = \atype_1 \typepijl \ldots \typepijl
  \atype_n \typepijl \abasetype$.
  Then $\mu(\app{\aterm}{\vec{\bterm}})
  >_{\mathsf{CPO}} \mu(\app{\bterm_i}{\vec{\c}})$ if
  (selecting the first argument $n-i$ times)
  $\mu(\app{\aterm}{\bterm_1} \cdots \bterm_i) = @^{\atype,\btype}(
  @(\ldots @(\mu(\aterm), \mu(\bterm_1)),
  \ldots \mu(\bterm_{i-1})),\mu(\bterm_i))
  >_{\mathsf{CPO}} \mu(\app{\bterm_i}{\vec{\c}})$.
  Since $\atype$ is a strict subtype of $\atype \typepijl \btype$,
  so $@^{\atype,\btype} >_F @^{\atype_j,\atype_{j+1} \ftypepijl \ldots
  \atype_n \ftypepijl \abasetype}$ for all $j$, and because certainly
  $@^{\atype,\btype}(\ldots) >_{\mathsf{CPO}} \c_{\atype_j}$ for all
  $j$, this indeed holds, as $\mu(\app{\bterm_i}{\vec{\c}}) =
  @^{\atype_n,\abasetype}(\ldots @^{\atype_1,\atype_2 \ftypepijl
  \ldots \atype_n \ftypepijl \abasetype}(\mu(\bterm_i),
  \c_{\atype_1}),\ldots, \c_{\atype_n})$.
\item Similarly, $\mu(\app{\afun(\vec{\aterm})}{\vec{\bterm}})
  >_{\mathsf{CPO}} \mu(\app{\aterm_i}{\vec{\c}})$ if
  $\afun(\mu(\aterm_1),\ldots,\mu(\aterm_n))
  >_{\mathsf{CPO}} \mu(\app{\aterm_i}{\vec{\c}})$, which holds
  because $\afun >_F @^{\atype_j,\atype_{j+1} \ftypepijl \atype_n
  \ftypepijl \abasetype}$ for all $j$, and also $\afun >_F
  \c_{\atype_j}$.
\end{iteMize}

\noindent
Rather than altering existing definitions of HORPO and CPO with a
transformation like the one above, we might consider definitions
which natively have minimal elements -- this is an addition which
should not be hard to include in the well-foundedness proof.

\begin{example}\label{ex:from}
Consider the following AFS:
\[
\F = \left\{
\begin{array}{rclrcl}
\nul & : & \nat &
\ifex & : & [\booleanex \times \lijst \times \lijst] \decpijl \lijst
  \\
\suc & : & [\nat] \decpijl \nat &
\symb{lteq} & : & [\nat \times \nat] \decpijl \booleanex \\
\trueex & : & \booleanex &
\symb{from} & : & [\nat \times \lijst] \decpijl \lijst \\
\falseex & : & \booleanex &
\symb{chain} & : & [\nat \typepijl \nat \times \lijst] \decpijl
  \lijst \\
\nil & : & \lijst &
\symb{incch} & : & [\lijst] \decpijl \lijst \\
\cons & : & [\nat \times \lijst] \decpijl \lijst \\
\end{array}
\right\}
\]
\[
\Rules = \left\{
\begin{array}{rclrcl}
\symb{lteq}(\suc(\avar),\nul) & \arrz & \falseex &
\ifex(\trueex,\avar,\bvar) & \arrz & \avar \\
\symb{lteq}(\nul,\avar) & \arrz & \trueex &
\ifex(\falseex,\avar,\bvar) & \arrz & \bvar \\
\symb{lteq}(\suc(\avar),\suc(\bvar)) & \arrz & \symb{lteq}(\avar,
  \bvar) &
\symb{from}(\avar,\nil) & \arrz & \nil \\
\symb{incch}(\avar) & \arrz & \symb{chain}(\abs{\bvar}{\suc(\bvar)},
  \avar) &
\symb{chain}(F,\nil) & \arrz & \nil \\
\symb{from}(\avar,\cons(\bvar,\cvar)) & \arrz &
  \multicolumn{4}{l}{
    \ifex(\symb{lteq}(\avar,\bvar),\cons(\bvar,\cvar),
    \symb{from}(\avar,\cvar))
  } \\
\symb{chain}(F,\cons(\bvar,\cvar)) & \arrz &
  \multicolumn{4}{l}{
    \cons(\app{F}{\bvar},\symb{chain}(F,\symb{from}(\app{F}{\bvar},
    \cvar)))
  } \\
\end{array}
\right\}
\]
\end{example}

\noindent
This AFS, which appears in the termination problem
database (see~\cite{termcomp}), has eight dependency pairs, and the
following dependency graph:
\vspace{-6pt}

\begin{wrapfigure}{r}{0.25\textwidth}
\vspace{-10pt}
\begin{tikzpicture}[->]
\begin{scope}[>=stealth]
\tikzstyle{veld} = [draw, fill=white,
  minimum height=0em, minimum width=0em, rounded corners]

\node (n6) [veld] {6};
\node (n1) [veld, above of=n6] {1};
\node (n4) [veld, right of=n1] {4};
\node (n5) [veld, below of=n4] {5};
\node (n3) [veld, below of=n5] {3};
\node (n2) [veld, left  of=n6] {2};
\node (n8) [veld, below of=n6] {8};
\node (n7) [veld, left  of=n8] {7};

\draw (n1) to[out=90,in=180,looseness=4] (n1);
\draw (n2) -- (n6);
\draw (n2) -- (n7);
\draw (n2) -- (n8);
\draw (n4) -- (n1);
\draw (n5) -- (n3);
\draw (n5) -- (n4);
\draw (n5) to[out=-45,in=45,looseness=4] (n5);
\draw (n6) -- (n1);
\draw (n6) -- (n2);
\draw (n6) -- (n3);
\draw (n6) -- (n4);
\draw (n6) -- (n5);
\draw (n6) to[out=100,in=170,looseness=5] (n6);
\draw (n6) -- (n7);
\draw (n6) -- (n8);
\draw (n7) -- (n6);
\draw (n7) to[out=270,in=180,looseness=4] (n7);
\draw (n7) -- (n8);
\draw (n8) -- (n4);
\draw (n8) -- (n5);
\draw (n8) -- (n3);
\end{scope}
\end{tikzpicture}
\vspace{-15pt}
\end{wrapfigure}

\ 
\begin{enumerate}[(1)]
\item\label{dp:lteq}
  $\up{\symb{lteq}}(\suc(\avar),\suc(\bvar)) \dppijl
  \up{\symb{lteq}}(\avar,\bvar)$
\item\label{dp:incch}
  $\up{\symb{incch}}(\avar) \dppijl \up{\symb{chain}}(
  \abs{\bvar}{\suc(\bvar)},\avar)$
\item\label{dp:from:1}
  $\up{\symb{from}}(\avar,\cons(\bvar,\cvar)) \dppijl
  \up{\ifex}(\symb{lteq}(\avar,\bvar),\cons(\bvar,\cvar),
  \symb{from}(\avar,\cvar))$
\item\label{dp:from:2}
  $\up{\symb{from}}(\avar,\cons(\bvar,\cvar)) \dppijl
  \up{\symb{lteq}}(\avar,\bvar)$
\item\label{dp:from:3}
  $\up{\symb{from}}(\avar,\cons(\bvar,\cvar)) \dppijl
  \up{\symb{from}}(\avar,\cvar)$
\item\label{dp:chain:1}
  $\up{\symb{chain}}(F,\cons(\bvar,\cvar)) \dppijl \app{F}{\bvar}$
\item\label{dp:chain:2}
  $\up{\symb{chain}}(F,\cons(\bvar,\cvar)) \dppijl
  \up{\symb{chain}}(F,\symb{from}(\app{F}{\bvar},\cvar))$
\item\label{dp:chain:3}
  $\up{\symb{chain}}(F,\cons(\bvar,\cvar)) \dppijl
  \up{\symb{from}}(\app{F}{\bvar},\cvar)$
\end{enumerate}

\noindent
We consider the SCC consisting of pairs \ref{dp:incch},
\ref{dp:chain:1} and \ref{dp:chain:2}.  As this is a \llfe\ AFS, we
may use almost unrestricted argument functions.
Let $\argfil(\suc(\avar)) = \argfil(\suc^-(\avar)) = \avar,\ 
\argfil(\up{\symb{incch}}(\avar)) = \up{\symb{chain}}(\abs{\bvar}{
\bvar},\avar),\ \argfil(\symb{incch}(\avar)) = \symb{chain}(
\abs{\bvar}{\bvar},\avar),\ \argfil(\symb{lteq}(\avar,\bvar)) =
\symb{lteq}'$ and $\argfil(\symb{from})(\avar,\bvar) =
\symb{from}'(\bvar)$; for other symbols, $\argfil$ is the identity.
Then the three dependency pairs of interest are oriented with
$(\geqtermcpo,\gtermcpo)$ if $\cons >_F \symb{from}'$:
\begin{enumerate}[(1)]
\item[(\ref{dp:incch})] $\filter(\up{\symb{incch}}(\avar)) =
  \up{\symb{chain}}(\abs{\bvar}{\bvar},\avar) \geqtermcpo
  \up{\symb{chain}}(\abs{\bvar}{\bvar},\avar) =
  \filter(\up{\symb{chain}}(\abs{\bvar}{\suc^-(\bvar)},\avar))$
\item[(\ref{dp:chain:1})] $\up{\symb{chain}}(F,\cons(\bvar,\cvar))
  \gtermcpo \app{F}{\bvar}$
\item[(\ref{dp:chain:2})] $\up{\symb{chain}}(F,\cons(\bvar,\cvar))
  \gtermcpo \up{\symb{chain}}(F,\symb{from}'(\cvar))$
\end{enumerate}
All rules, even the non-formative ones, are oriented as well, with
for instance the precedence $\symb{chain},\up{\symb{chain}} >_F \cons
>_F \symb{from}' >_F \ifex,\symb{lteq}' >_F \falseex,\trueex$
\[
\begin{array}{rclrcl}
\symb{lteq}' & \geqtermcpo & \falseex &
\ifex(\trueex,\avar,\bvar) & \geqtermcpo & \avar \\
\symb{lteq}' & \geqtermcpo & \trueex &
\ifex(\falseex,\avar,\bvar) & \geqtermcpo & \bvar \\
\symb{lteq}' & \geqtermcpo & \symb{lteq}' &
\symb{from}'(\nil) & \geqtermcpo & \nil \\
\symb{chain}(\abs{\bvar}{\bvar},\avar) & \geqtermcpo &
  \symb{chain}(\abs{\bvar}{\bvar},\avar) &
\symb{chain}(F,\nil) & \geqtermcpo & \nil \\
\symb{from}'(\cons(\bvar,\cvar)) & \geqtermcpo &
  \multicolumn{4}{l}{
    \ifex(\symb{lteq}',\cons(\bvar,\cvar),\symb{from}'(\cvar))
  } \\
\symb{chain}(F,\cons(\bvar,\cvar)) & \geqtermcpo &
  \multicolumn{4}{l}{
    \cons(\app{F}{\bvar},\symb{chain}(F,\symb{from}'(\cvar)))
  } \\
\end{array}
\]
The remaining dependency pair (2) is clearly tagged-chain-free, so
this AFS is terminating if the two SCCs $\{($\ref{dp:lteq}$)\}$ and
$\{($\ref{dp:from:3}$)\}$ are tagged-chain-free.

\section{Non-collapsing Dependency Pairs}\label{sec:noncollapse}

Many powerful aspects of the first-order dependency pair framework,
such as the subterm criterion and usable rules, break in the presence
of collapsing dependency pairs.  But for \emph{parts} of a
termination proof, we may have non-collapsing dependency pairs.
Consider for example the system from Example~\ref{ex:from}: the
dependeny graph has three SCCs, and two of them (the SCCs
$\{($\ref{dp:lteq}$)\}$ and $\{($\ref{dp:from:3}$)\}$) are
non-collapsing.  Their chain-freeness can be
demonstrated with a reduction triple which does not have the
limited subterm property.

In this section, we will briefly discuss the subterm criterion and
usable rules, limited to non-collapsing sets $\P$.  To avoid double
work between the basic and tagged dependency
pair approach, we note that tags add very little advantage in this
setting: tagging mainly plays a role in the limited or tagged subterm
property.  The primary advantage of Theorem~\ref{thm:weakfinish} over
Theorem~\ref{thm:typepreserve}, therefore, is the addition of
formative rules.  So let us revise the definition of
chain-free: \emph{a set $\P$ of dependency pairs is chain-free if
there is no minimal dependency chain $\rijtje{(\rho_i,\aterm_i,
\bterm_i) \mid i \in \N}$ with all $\rho_i \in \P \cup \{\cbeta\}$
such that $\bterm_i \arrr{\formrules(\P)} \aterm_{i+1}$ for all $i$}.
Here, $\formrules(\P) = \Rules$ if the AFS under consideration is not
\llfe.

\begin{lemma}\label{lem:ignoretags}
A set $\P$ of dependency pairs is tagged-chain-free if it is
chain-free.
\end{lemma}

\begin{proof}
Obvious consequence of Theorem~\ref{thm:dependencychainweak}, by
removing the tags.
\end{proof}

We will consider only how to prove chain-freeness for the sets
under consideration.

\subsection{The Subterm Criterion}\label{subsec:subcrit}

Let $\P$ be non-collapsing, and let
$\mathcal{H}$ be the set of function symbols $\afun$ such that a
left- or right-hand side of a dependency pair has the form
$\app{\afun(\vec{\aterm})}{\vec{\bterm}}$.  A \emph{projection
function} for $\mathcal{H}$ is a function $\nu$ which assigns to each
$\afun \in \mathcal{H}$ a number $i$ such that for all
$l \dppijl p \in \P$, the following function $\overline{
\nu}$ is well-defined for both $l$ and $p$:
\[
\overline{\nu}(\app{\afun(\aterm_1,\ldots,\aterm_m)}{\aterm_{m+1}}
\cdots \aterm_n) = \aterm_{\nu(\afun)}
\]
This differs from the first-order definition only in that we cater
for dependency pairs which do not have the form $\up{\afun}(l_1,
\ldots,l_n) \dppijl \up{\bfun}(p_1,\ldots,p_m)$.

\begin{theorem}\label{thm:subcrit}
Let $\P = \P_1 \uplus \P_2$ be a set of non-collapsing dependency
pairs, and suppose a projection function $\nu$ exists such that
$\overline{\nu}(l) \supterm \overline{\nu}(p)$ for all $l \dppijl p
\in \P_1$ and $\overline{\nu}(l) = \overline{\nu}(p)$ for all $l
\dppijl p \in \P_2$.
Then $\P$ is chain-free if and only if $\P_2$ is.
\end{theorem}

\begin{proof}
If $\P$ is chain-free then obviously $\P_2$ is chain-free.
For the other direction, suppose $\P_2$ is chain-free and
that we have a minimal dependency chain
$\rijtje{(\rho_i,\aterm_i,\bterm_i) \mid i \in \N}$ over $\P$ and
suitable projection function $\nu$.
As before, we can assume this
chain does not use $\cbeta$, so always $\rho_i = l_i \dppijl p_i$.
As the subterm relation is stable, $\overline{\nu}(l_i) \suptermeq
\overline{\nu}(p_i)$ implies that $\overline{\nu}(\aterm_i)
\suptermeq \overline{\nu}(\bterm_i)$.

Since each $\bterm_i \arrr{\Rules} \aterm_{i+1}$, we have a reduction
$\overline{\nu}(\aterm_1) \suptermeq \overline{\nu}(\bterm_1)
\arrr{\Rules} \overline{\nu}(\aterm_2) \suptermeq
\overline{\nu}(\bterm_2) \arrr{\Rules} \ldots$
Since $\P_2$ is chain-free, $\rho_i \in \P_i$ for infinitely many $i$,
so $\overline{\nu}(\aterm_i) \supterm \overline{\nu}(\bterm_i)$
infinitely often.
Thus we have an infinite $\arr{\Rules} \cup \supterm$ reduction
starting in a terminating term $\overline{\nu}(\aterm_1)$,
contradiction.
\end{proof}

\begin{example}
Both of the remaining SCCs of Example~\ref{ex:from} are chain-free,
as is evident with a projection $\nu(\up{\symb{lteq}}) = 1$ for
the SCC $\{($\ref{dp:lteq}$)\}$ and $\nu(\up{\symb{from}}) = 2$ for
the SCC $\{($\ref{dp:from:3}$)\}$.
\vspace{-6pt}
\end{example}

\subsection{Usable Rules}

Formative rules provide a nice
counterpart of usable
rules, but it would be even nicer if we had both!  Usable rules can
be defined, but with severe restrictions; we follow
the ideas in~\cite{suz:kus:bla:11}, where usable rules
are defined for static dependency pairs.

\begin{definition}[Usable Rules]\label{def:userules}
A term $\aterm$ is considered \emph{risky} if it has a subterm
$\app{\avar}{\bterm}$ with $\avar \in \FV(\aterm)$.
Let $\afun \gsymbuse \bfun$ if there is a rewrite
rule $\afun(l_1,\ldots,l_n) \arrz r$ where either $r$ contains the
symbol $\bfun$, or $r$ is risky, or $r$ is an abstraction or
variable of functional type.
The reflexive-transitive closure of $\gsymbuse$ is
denoted by $\gsymbuse^*$.
Overloading notation, let $\aterm \gsymbuse^* \bfun$ denote that
either $\aterm$ contains a symbol $\afun$ with $\afun \gsymbuse^*
\bfun$, or $\aterm$ is risky.

The set of usable rules of a term $\aterm$, notation $\userules(
\aterm,\Rules)$, consists of those rules $\app{\afun(l_1,\ldots,
l_n)}{l_{n+1}} \cdots l_m \arrz r \in \Rules$ such that $\aterm
\gsymbuse^* \afun$.
The set of usable rules of a non-collapsing dependency pair $l
\dppijl \app{\afun(p_1,\ldots,p_n)}{p_{n+1}} \cdots p_m$ is the union
$\userules(p_1,\Rules) \cup \ldots \cup \userules(p_m,\Rules)$.
The set of usable rules of a set of dependency pairs, $\userules(\P,
\Rules)$, is the union $\bigcup_{\rho \in \P} \userules(\rho,\Rules)$
if $\P$ is non-collapsing and just $\Rules$ otherwise.
\end{definition}

\begin{example}
Consider an AFS for list manipulation which has four rules:
\[
\begin{array}{rcl}
\map(F,\nil) & \arrz & \nil \\
\map(F,\cons(h,t)) & \arrz & \cons(\app{F}{h},\map(F,t)) \\
\append(\nil,l) & \arrz & l \\
\append(\cons(h,t),l) & \arrz & \cons(\append(h,t),l) \\
\end{array}
\]
The dependency graph of this system has two SCCs: the set
$\{ \up{\map}(F,\cons(h,t)) \dppijl \app{F}{h},\ 
\up{\map}(F,\cons(h,t)) \dppijl \up{\map}(F,t)) \}$ and
the set $\{ \up{\append}(\cons(h,t),l) \dppijl \up{\append}(h,t) \}$.
The former contains a dependency pair with $\app{F}{h}$ in the
right-hand side, so its usable rules are just $\Rules$.  But the
latter set has only the $\append$ rules as usable rules.
\end{example}

Very similar to the first-order case, we can prove the following
result:

\begin{theorem}\label{thm:userules}
Let $\P = \P_1 \uplus \P_2$ be a set of non-collapsing dependency
pairs, and suppose there is a reduction pair $(\geqterm,\gterm)$ such
that:
\begin{iteMize}{$\bullet$}
\item $\overline{l} \gterm \overline{p}$ for all $l \dppijl p \in
  \P_1$
\item $\overline{l} \geqterm \overline{p}$ for all $l \dppijl p \in
  \P_2$
\item $l \geqterm r$ if $l \arrz r \in \userules(\P,\formrules(\P))$
\item $\symb{p}_\atype(\avar,\bvar) \geqterm \avar,\bvar$ for fresh
  symbols $\symb{p}_\atype : [\atype \times \atype] \decpijl \atype$.
\end{iteMize}
Then $\P$ is chain-free if and only if $\P_2$ is chain-free.
\end{theorem}

Here, the pair $\overline{l},\overline{p}$ is determined from $l,p$
in a systematic way like in Section~\ref{subsec:typechange}.  The
$\symb{p}_\atype(\avar,\bvar) \geqterm \avar,\bvar$ constraints are
trivially oriented both with algebra interpretations and CPO.

The proof of Theorem~\ref{thm:userules} takes some work, but has
no novelties compared to the proof for the static method
in~\cite{bla:jou:rub:08:1}.  We shall be relatively brief about it.

\begin{proof}[Proof Sketch of Theorem~\ref{thm:userules}]
Let $\symb{p}_\atype : [\atype \times \atype] \decpijl \atype$ and
$\bot_\atype : \atype$ be new symbols for all types $\atype$, and let
$\Ce$ be the set of all rules $\symb{p}_\atype(\avar,\bvar) \arrz
\avar,\ \symb{p}_\atype(\avar,\bvar) \arrz \bvar$.  We will see that
any minimal dependency chain $\rijtje{(\rho_i,\aterm_i,\bterm_i) \mid
i \in \N}$ with all $\rho_i \in \P$ can be transformed into a (not
necessarily minimal) dependency chain $\rijtje{(\rho_i,\aterm_i',
\bterm_i') \mid i \in \N}$ which uses the same dependency pairs,
but where $\bterm_i' \arrr{\userules(\P,
\formrules(\P)) \cup \Ce} \aterm_{i+1}'$ for all $i$.
Following the proof of Theorems~\ref{thm:maintheorem}
and~\ref{thm:typepreserve}, if a reduction pair as described in the
theorem exists, then this transformed chain must use dependency pairs
in $\P_2$ infinitely often, so the same holds for the original.
Thus, if $\P$ is not chain-free, then neither is $\P_2$.  Of course,
if $\P$ is chain-free then so is its subset $\P_2$.

To transform the minimal dependency chain
$\rijtje{(\rho_i,\aterm_i,\bterm_i) \mid i \in \N}$, we note that:

\noindent(I) \emph{every term has only finitely many direct reducts.}

\noindent
This is obvious because $\Rules$ is finite.
Let $\Rules_1 := \formrules(\P)$ and say a symbol $\afun$ is a
\emph{usable symbol} if $p \gsymbuse^* \afun$ for some dependency
pair $l \dppijl p \in \P$ (where $\gsymbuse^*$ is based on $\Rules_1$
rather than $\Rules$), or $\afun$ is a constructor symbol.  We
assume:\medskip

\noindent(II) \emph{$\userules(\P,\Rules_1) \neq \Rules_1$.}

\noindent
This is a safe assumption, because if $\userules(\P,\Rules_1) =
\Rules_1$ the theorem is automatically satisfied.  Now, because
of (I) we can define the function $\varphi$ on terminating terms,
given by the clauses:
\begin{iteMize}{$\bullet$}
\item $\varphi(\abs{\avar}{\cterm}) = \abs{\avar}{\varphi(\cterm)}$
\item $\varphi(\app{\avar}{\cterm_1} \cdots \cterm_n) =
  \app{\avar}{\varphi(\cterm_1)} \cdots \varphi(\cterm_n)$
\item $\varphi(\app{\afun(\cterm_1,\ldots,\cterm_n)}{\cterm_{n+1}}
  \cdots \cterm_m) = \app{\afun(\varphi(\cterm_1),\ldots,
  \varphi(\cterm_n))}{\varphi(\cterm_{n+1})} \cdots \varphi(\cterm_m)
  $ if $\afun$ is a usable symbol
\item $\varphi(\aterm) = \symb{p}_\atype(\app{\afun(
  \varphi(\cterm_1),\ldots,\varphi(\cterm_n))}{\varphi(\cterm_{n+1})}
  \cdots \varphi(\cterm_m),D_\atype(\{\bterm \mid \aterm
  \arr{\Rules_1} \bterm\}))$ if $\aterm : \atype$ and $\aterm =
  \app{\afun(\cterm_1,\ldots,\cterm_n)}{\cterm_{n+1}} \cdots \cterm_m$
  with $\afun$ not a usable symbol
\item $\varphi(\aterm) = D_\atype(\{\bterm \mid \aterm \arr{\Rules_1}
  \bterm\})$ if $\aterm : \atype$ and $\aterm$ does not have any of
  these forms
\end{iteMize}
Here, $D_\atype(\emptyset) = \bot_\atype$ and $D_\atype(X) =
\p_\atype(\varphi(\cterm),D_\atype(X \setminus \{\cterm\}))$ if $X$
is non-empty and $\cterm$ is lexicographically its smallest element
(this is only defined for finite sets $X$).

Writing $\asub^\varphi$ for a substitution $[\avar:=\varphi(\asub(
\avar)) \mid \avar \in \domain(\asub)]$, it is not hard to see that:\medskip

\noindent(III) \emph{if $\aterm$ has no subterms $\app{(\abs{\avar}{\bterm})}{
\cterm}$ or $\app{\avar}{\cterm}$ with $\avar \in \domain(\asub)$,
then $\varphi(\aterm\asub) \arrr{\Ce} \aterm\asub^\varphi$.  If
all symbols occurring in $\aterm$ are usable symbols, then even
$\varphi(\aterm\asub) = \aterm\asub^\varphi$.}

\noindent
The proof holds by induction on the form of $\aterm$, noting that the
last case of the definition of $\varphi$ is never applicable.
With induction on the size of $X$ it follows easily that
$D_\atype(X) \arrr{\Ce} \varphi(\aterm)$ for any $\aterm \in
X$.  Combining this with (III) we can derive:\medskip

\noindent(IV) \emph{if $\aterm \arr{\Rules_1} \bterm$ with $\aterm$
terminating, then $\varphi(\aterm) \arrr{\userules(\P,\Rules_1) \cup
\Ce} \varphi(\bterm)$.}

\noindent
This holds with induction on the form of $\aterm$; all
the induction cases are trivial either with the induction hypothesis
or with the observation that $D_\atype(X) \arrr{\Ce} \varphi(\cterm)$
if $\cterm \in X$.  By the same observation, the base case (a headmost
step) is easy if $\aterm = \app{\afun(\vec{\cterm})}{\vec{\dterm}}$
with $\afun$ not a usable symbol, or if $\aterm$ reduces with a
headmost $\beta$-step.  What remains is the case when $\aterm =
\app{\afun(\vec{\cterm})}{\vec{\dterm}}$ with $\afun$ a usable
symbol, and the reduction is headmost: $\aterm = \app{l\asub}{
\dterm_{k+1}} \cdots \dterm_m$ and $\bterm = \app{r\asub}{
\dterm_{k+1}} \cdots \dterm_m$.  But since $\afun$ is a usable
symbol, $l \arrz r$ is a usable rule, so since rules are
$\beta$-normal: $\varphi(\aterm) = \app{\varphi(l\asub)}{
\varphi(\dterm_{k+1})} \cdots \varphi(\dterm_m) \arrr{\Ce}
\app{l\asub^\varphi}{\varphi(\dterm_{k+1})} \cdots
\varphi(\dterm_m) \arr{\userules(\P,\Rules_1)}
\app{r\asub^\varphi}{\varphi(\dterm_{k+1})} \cdots \varphi(\dterm_m)
= \app{\varphi(r\asub)}{\varphi(\dterm_{k+1})} \cdots
\varphi(\dterm_m)$ by (III), which equals $\varphi(\bterm)$ because
either $r$ has the form $\app{\bfun(\vec{\eterm})}{\vec{\fterm}}$,
or $m = k$ (since otherwise $\userules(\P,\Rules_1) = \Rules_1$,
contradicting (II)).

Now we are almost there.  Let $\rijtje{(\rho_i,\aterm_i,\bterm_i)
\mid i \in \N}$ be a minimal dependency chain with all $\rho_i \in
\P$, so for each $i$ we can write $\rho_i = l_i \dppijl p_i$ and
$\aterm_i = l_i\asub_i,\ \bterm_i = p_i\asub_i$.  Let
$\aterm_i' = l_i\asub_i^\varphi$ and $\bterm_i' =
p_i\asub_i^\varphi$ (this is well-defined because the strict subterms
of $\aterm_i$ and $\bterm_i$ are terminating).  Then by (III) and (IV)
$p_i\asub_i^\varphi \arrr{\userules(\P,\formrules(\P)) \cup \Ce,in}
l_{i+1}\asub_{i+1}^\varphi$ as required.
\end{proof}

The notion of usable rules is in particular useful for systems with
a first-order subset, as the usable rules of first-order dependency
pairs will typically be first-order.  As discussed
in~\cite{fuh:kop:11:1}, we can use this to apply first-order
termination proving techniques (and termination tools!) to a part of
a higher-order dependency pair problem.\newpage

\section{Algorithm}\label{sec:algorithm}

\summary{In this section, we combine all the results so far in a
ready-to-use algorithm.}

All the results in this paper can be readily combined in one
algorithm, to prove termination of an AFS (\llfe\ or not):
\begin{enumerate}[(1)]
\item \label{alg:prepare}
  \begin{iteMize}{$\bullet$}
  \item
  determine whether the system is \llfe\ (Definition~\ref{def:llfe}), 
  \item
  complete the system (Definition~\ref{def:completing}),
  \item
  determine its dependency pairs (Definition~\ref{def:deppair})
  \item
  calculate an
  approximation $G$ for the dependency graph (Section~\ref{subsec:graph});
  \end{iteMize}
\item \label{alg:start}
  remove all nodes and edges from $G$ which are not part of a cycle
  (Section~\ref{subsec:graph});
\item \label{alg:maxcyc}
  if $G$ is empty return \texttt{terminating};
  otherwise choose an SCC $\P$ (Section~\ref{subsec:graph});
\item \label{alg:subcrit}
  if $\P$ is non-collapsing, find a projection function $\nu$
  such that $\overline{\nu}(l) \suptermeq \overline{\nu}(p)$ for all
  $l \dppijl p \in \P$, and $\overline{\nu}(l)
  \supterm \overline{\nu}(p)$ at least once
  (Section~\ref{subsec:subcrit});
  if this succeeds, remove all
  strictly oriented pairs from $G$ and continue
  with~\ref{alg:start}, otherwise continue
  with~\ref{alg:weaknonweak};
\item \label{alg:weaknonweak}
  a. if $\P$ is non-collapsing:
  \begin{iteMize}{$\bullet$}
  \item let $S := \emptyset$ and $\Sigma := \up{\F}_c$;
  \item let $\psi(\aterm)$ be defined as just $\aterm$;
  \item let $A\! :=\! \userules(\P,\formrules(\P))$ if
    $\Rules$ is \llfe, $A\! :=\! \userules(\P,\Rules)$ otherwise
    (Definition~\ref{def:formrules},\ref{def:userules});
  \end{iteMize}
  b. if $\P$ is collapsing and $\Rules$ is not \llfe:
  \begin{iteMize}{$\bullet$}
  \item let $S := \F$ and $\Sigma := \up{\F}_c$;
  \item let $\psi(\aterm)$ be defined as just $\aterm$;
  \item let $A := \Rules \cup \{ \afun(\vec{\avar}) \arrz \up{\afun}(
    \vec{\avar}) \mid \afun \in \Defineds\}$;
  \end{iteMize}
  c. if $\P$ is collapsing and $\Rules$ is \llfe:
  \begin{iteMize}{$\bullet$}
  \item let $S$ be the set of function symbols $\afun^- : \atype$
    where $\afun : \atype \in \F$ and $\afun$ occurs below an
    abstraction in a right-hand side of $\formrules(\P) \cup
    \P$, and let $\Sigma := \up{\F}_c \cup S$;
  \item let $\psi(\aterm)$ be defined as $\ttag(\aterm)$
    (Definition~\ref{def:ttag});
  \item let $A := \formrules(\P)^\ttag \cup \{ \afun^-(\vec{\avar})
    \arrz \afun(\vec{\avar}), \bfun^-(\vec{\avar}) \arrz \up{\bfun}(
    \vec{\avar}) \mid \afun^-,\bfun^- \in S,\ \bfun \in \Defineds
    \}$;
  \end{iteMize}
\item \label{alg:problem}
  determine a partitioning $\P = \P_1 \uplus \P_2$
  and a reduction pair (Section~\ref{subsec:redord})
  $(\geqterm,\gterm)$ with:
  \begin{enumerate}
  \item $\app{l}{\avar_1} \cdots \avar_n \gterm \app{\psi(p)}{\c_{
    \atype_1}} \cdots \c_{\atype_m}$ for $l \dppijl p \in \P_1$
    (both sides base type, fresh $\vec{\avar}$);
  \item $\app{l}{\avar_1} \cdots \avar_n \geqterm \app{\psi(p)}{\c_{
    \atype_1}} \cdots \c_{\atype_m}$ for $l \dppijl p \in \P_2$
    (both sides base type, fresh $\vec{\avar}$);
  \item $l \geqterm r$ for $l \arrz r \in A$;
  \item either $\P$ is non-collapsing, or $(\geqterm,\gterm)$
    respects $\supterm^S$ (Definition~\ref{def:suptermS});
  \item if $\Rules$ is \llfe\ and $\P$ collapsing: \\
    for all $\afun \in \F^- \setminus S$: $\afun^-(\vec{\avar})
    \geqterm \afun(\vec{\avar}),\up{\afun}(\vec{\avar})$ and
    $\app{\afun^-(\vec{\avar})}{\vec{\bvar}} \geqterm
    \app{\avar_i}{\vec{\c}}$ (**).
  \end{enumerate}
  if this step fails, return \texttt{fail}; \\
  \emph{suitable reduction pairs can be found with e.g. weakly
  monotonic algebras (Section~\ref{sec:monalg}), or argument
  functions and CPO (Section~\ref{sec:argfun})}
\item \label{alg:removed1}
  remove all pairs in $\P_1$ from the graph, and continue with
  (\ref{alg:start}).
\end{enumerate}
(**) Since the symbols $\afun \in \F^- \setminus S$ do not occur in
any of the other constraints, all the common reduction pairs (CPO,
weakly monotonic algebras and most existing first-order techniques)
can easily be extended to satisfy these requirements; for example,
given a path ordering which satisfies the rest of the requirements,
just add ``$\afun^- >_F \afun$'' to the function symbol precedence.
That is why we split off these requirements rather than including the
symbols in $S$, and why we generally ignore this restriction on
$(\geqterm,\gterm)$ in examples.

This algorithm iterates over the graph approximation, removing nodes
and edges until none remain; this technique originates
in~\cite{hir:mid:05:1}.  It is justified by the following
observation:

\begin{lemma}\label{lem:subcycles}
Let $G$ be an approximation of the dependency graph of an AFS
$\Rules$, and let $\P$ be a set of dependency pairs.  Suppose that
every SCC in the subgraph of $G$ which contains only nodes in $\P$ is
chain-free.  Then also $\P$ is chain-free.
\end{lemma}

\begin{proof}
A trivial adaptation of the proof of
Lemma~\ref{lem:cyclenondangerous}.
\end{proof}

The algorithm seeks to prove that $\DP(\Rules)$ is either
tagged-chain-free (if $\Rules$ is \llfe) or chain-free (if not),
which suffices by Theorems~\ref{thm:dependencychainweak}
and~\ref{thm:dependencychain} respectively.
\begin{enumerate}[(1)]
\item To prove that $\DP(\Rules)$ is (tagged-)chain-free, it suffices
  to prove that all SCCs of $G$ are either tagged-chain-free (if
  $\Rules$ is \llfe\ and $\P$ is non-collapsing) or chain-free (if
  $\Rules$ is not \llfe\ or $\P$ is collapsing); this is valid by
  Lemmas~\ref{lem:cyclenondangerous} and~\ref{lem:ignoretags}.
\item Thus, until $G$ is empty, we choose an SCC $\P$ and show that
  $\P$ is tagged-chain-free or chain-free as required if some strict
  subset $\P'$ is (tagged-)chain-free (using the subterm criterion
  or a reduction pair, which is justified by Theorems
  \ref{thm:typepreserve}, \ref{thm:weakfinish} and
  \ref{thm:subcrit}). By Lemma~\ref{lem:subcycles} this is the case
  if all sub-SCCs $\P_1,\ldots,\P_n$ of $\P'$ are (tagged-)chain-free.
  To see that all remaining SCCs of $G$ and $\P_1,\ldots,\P_n$ are
  (tagged-)chain-free, we remove all strictly oriented pairs from
  $G$, and remove all nodes not on a cycle, so the SCCs of the result
  are exactly the sets that must be proved (tagged-)chain-free.
\end{enumerate}

\begin{example}\label{ex:eval}
We apply the algorithm on the AFS $\eval$,
with the 
following symbols:
\[
\begin{array}{rclrclrcl}
\nul & : & \M &
\dom & : & [\M \times \M \times \M]  \decpijl \M &
\eval & : & [\M \times \M] \decpijl \M \\
\suc & : & [\M] \decpijl \M\ \ \  &
\fun & : & [\M \typepijl \M \times \M \times \M] \decpijl \M\ \ \  \\
\end{array}
\]
and the following rewrite rules:
\[
\begin{array}{rclrcl}
\dom(\suc(x),\suc(y),\suc(z)) & \arrz & \suc(\dom(x,y,z))\ \ \ \  &
\dom(x,y,\nul) & \arrz & x \\
\dom(\nul,\suc(y),\suc(y)) & \arrz & \suc(\dom(\nul,y,z)) &
\dom(\nul,\nul,z)& \arrz & \nul \\
\eval(\fun(F,x,y),z) & \arrz & \app{F}{(\dom(x,y,z))} \\
\end{array}
\]
The $\fun$ symbol represents a function over the natural numbers
with an interval it is defined on; $\eval$ calculates its value in
a point, provided the point is in the domain of the function.
We prove termination of $\eval$ by following the algorithm.

\emph{(\ref{alg:prepare}): complete the rules, calculate the
dependency pairs and approximate the dependency graph.}
The AFS $\eval$ is local.
Because $\eval$ has no rules of the form $l \arrz \abs{\avar}{r}$,
the completed system is the same as the original system.
We have the following four dependency pairs:
\[
\begin{array}{rclrcl}
\up{\eval}(\fun(F,x,y),z) & \dppijl & \app{F}{\dom(x,y,z)} &
\up{\dom}(\suc(x),\suc(y),\suc(z)) & \dppijl & \up{\dom}(x,y,z) \\
\up{\eval}(\fun(F,x,y),z) & \dppijl & \up{\dom}(x,y,z) &
\up{\dom}(\nul,\suc(y),\suc(z)) & \dppijl & \up{\dom}(\nul,y,z) \\
\end{array}
\]
We will use the following (approximation of the) dependency graph:

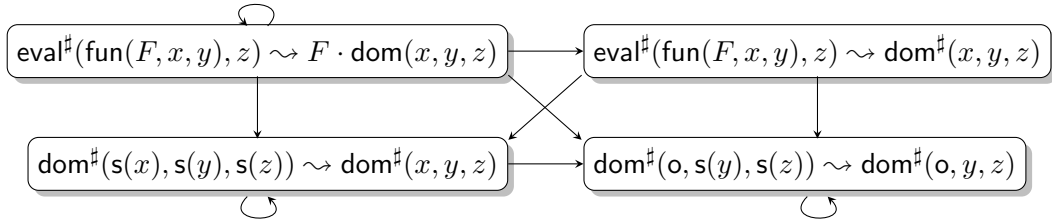
\begin{tikzpicture}[->]
\begin{scope}[>=stealth]

\tikzstyle{veld} = [draw, fill=white,  drop shadow, 
  minimum height=0em, minimum width=0em, rounded corners]

\node (topcentre) { };
\node (botcentre) [below of=topcentre, node distance=15mm] { };

\node [veld] (e1) [left of=topcentre, anchor=east, node distance=5mm] {
    $\up{\eval}(\fun(F,x,y),z) \dppijl \app{F}{\dom(x,y,z)}$
  };
\node (e1a) [above of=e1, node distance=2mm] {};
\node (e1b) [below of=e1, node distance=2mm] {};
\node (e1c) [right of=e1b, node distance=32mm] {};

\node [veld] (e2) [right of=topcentre,anchor=west,node distance=5mm] {
    $\up{\eval}(\fun(F,x,y),z) \dppijl \up{\dom}(x,y,z)$
  };
\node (e2a) [below of=e2, node distance=2mm] {};
\node (e2b) [left of=e2a, node distance=30mm] {};

\node [veld] (d1) [left of=botcentre,anchor=east,node distance=5mm] {
    $\up{\dom}(\suc(x),\suc(y),\suc(z)) \dppijl \up{\dom}(x,y,z)$
  };
\node (d1a) [below of=e1b, node distance=11mm] {};
\node (d1b) [right of=d1a, node distance=32mm] {};
\node (d1c) [below of=d1a, node distance=5mm] {};

\node [veld] (d2) [right of=botcentre,anchor=west,node distance=5mm] {
    $\up{\dom}(\nul,\suc(y),\suc(z)) \dppijl \up{\dom}(\nul,y,z)$
  };
\node (d2a) [below of=e2a, node distance=11mm] {};
\node (d2b) [left of=d2a, node distance=30mm] {};
\node (d2c) [below of=d2a, node distance=5mm] {};

\draw (e1a) to[out=45,in=135,looseness=5] (e1a);
\draw (e1) -- (e2);
\draw (e1b) -- (d1a);
\draw (e1c) -- (d2b);
\draw (e2b) -- (d1b);
\draw (e2a) -- (d2a);
\draw (d1) -- (d2);
\draw (d1c) to[out=225,in=315,looseness=5] (d1c);
\draw (d2c) to[out=225,in=315,looseness=5] (d2c);
\end{scope}
\end{tikzpicture}

\emph{Step~\ref{alg:start}: remove nodes and edges not part of a
cycle.}  This removes the second $\up{\eval}$ dependency pair, and
leaves the following graph: \\
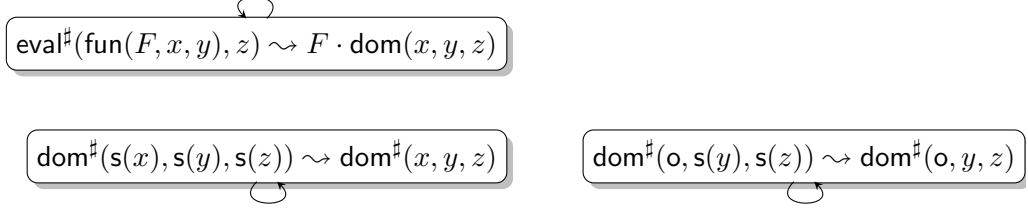
\begin{tikzpicture}[->]
\begin{scope}[>=stealth]

\tikzstyle{veld} = [draw, fill=white,  drop shadow, 
  minimum height=0em, minimum width=0em, rounded corners]

\node (topcentre) { };
\node (botcentre) [below of=topcentre, node distance=15mm] { };

\node [veld] (e1) [left of=topcentre, anchor=east, node distance=5mm] {
    $\up{\eval}(\fun(F,x,y),z) \dppijl \app{F}{\dom(x,y,z)}$
  };
\node (e1a) [above of=e1, node distance=2mm] {};

\node [veld] (d1) [left of=botcentre,anchor=east,node distance=5mm] {
    $\up{\dom}(\suc(x),\suc(y),\suc(z)) \dppijl \up{\dom}(x,y,z)$
  };
\node (d1c) [below of=d1, node distance=2mm] {};

\node [veld] (d2) [right of=botcentre,anchor=west,node distance=5mm] {
    $\up{\dom}(\nul,\suc(y),\suc(z)) \dppijl \up{\dom}(\nul,y,z)$
  };
\node (d2c) [below of=d2, node distance=2mm] {};

\draw (e1a) to[out=45,in=135,looseness=5] (e1a);
\draw (d1c) to[out=225,in=315,looseness=5] (d1c);
\draw (d2c) to[out=225,in=315,looseness=5] (d2c);
\end{scope}
\end{tikzpicture}

\emph{Step~\ref{alg:maxcyc}: return \texttt{terminating} if $G$ is
empty, otherwise choose an SCC $\P$.}
Since the graph is not empty, we choose $\P = \{
\up{\dom}(\nul,\suc(y),\suc(z)) \dppijl \up{\dom}(\nul,y,z)\}$.

\emph{Step~\ref{alg:subcrit}: apply the subterm criterion if
possible.}
The chosen set is indeed non-collapsing; we choose a projection
function $\nu(\up{\dom}) = 2$, and have $\overline{\nu}(\up{\dom}(
\nul,\suc(y),\suc(z))) = \suc(y) \supterm u = \overline{\nu}(
\up{\dom}(\nul,y,z))$, and can remove this dependency pair from $G$.

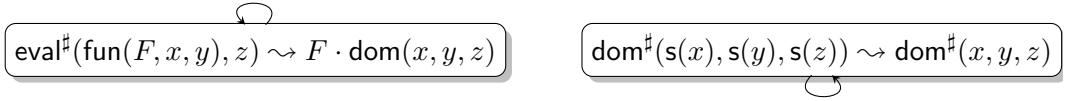
\begin{tikzpicture}[->]
\begin{scope}[>=stealth]

\tikzstyle{veld} = [draw, fill=white,  drop shadow, 
  minimum height=0em, minimum width=0em, rounded corners]

\node (topcentre) { };

\node [veld] (e1) [left of=topcentre, anchor=east, node distance=5mm] {
    $\up{\eval}(\fun(F,x,y),z) \dppijl \app{F}{\dom(x,y,z)}$
  };
\node (e1a) [above of=e1, node distance=2mm] {};

\node [veld] (d1) [right of=topcentre,anchor=west,node distance=5mm] {
    $\up{\dom}(\suc(x),\suc(y),\suc(z)) \dppijl \up{\dom}(x,y,z)$
  };
\node (d1c) [below of=d1, node distance=2mm] {};

\draw (e1a) to[out=45,in=135,looseness=5] (e1a);
\draw (d1c) to[out=225,in=315,looseness=5] (d1c);
\end{scope}
\end{tikzpicture}

\emph{Step~\ref{alg:start}: remove nodes and edges not part of a
cycle.} Everything is on a cycle.

\emph{Step~\ref{alg:maxcyc},\ref{alg:subcrit}: choose the next SCC
and apply the subterm criterion.}
We choose the set $\P = \{\up{\dom}(\suc(x),\suc(y),\suc(z)) \dppijl
\up{\dom}(x,y,z)\}$, which is non-collapsing, and apply the subterm
criterion with $\nu = 2$.  Since $\suc(y) \supterm y$, this leaves
the graph with a single node:

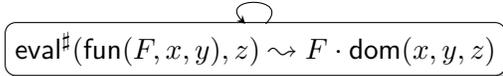
\begin{tikzpicture}[->]
\begin{scope}[>=stealth]

\tikzstyle{veld} = [draw, fill=white,  drop shadow, 
  minimum height=0em, minimum width=0em, rounded corners]

\node [veld] (e1) {
    $\up{\eval}(\fun(F,x,y),z) \dppijl \app{F}{\dom(x,y,z)}$
  };
\node (e1a) [above of=e1, node distance=2mm] {};

\draw (e1a) to[out=45,in=135,looseness=5] (e1a);
\end{scope}
\end{tikzpicture}

\emph{Step~\ref{alg:start}: remove nodes and edges not part of a
cycle.} Everything is on a cycle.

\emph{Step~\ref{alg:maxcyc}: return \texttt{terminating} if $G$ is
empty, otherwise choose an SCC $\P$.}
Since the graph is not empty, we choose $\P = \{
\up{\eval}(\fun(F,x,y),z) \dppijl \app{F}{\dom(x,y,z)} \}$.

\emph{Step~\ref{alg:subcrit}: apply the subterm criterion if
possible.}  This is not possible, as $\P$ is collapsing.

\emph{Step~\ref{alg:weaknonweak}: determine $\Sigma,\ S,\ A$ and
$\psi$.}
Noting that the system is \llfe, but has abstractions in neither left-
nor right-hand sides, we let $\Sigma := \{ \nul,
\suc, \dom, \fun, \eval, \up{\dom}, \up{\eval} \}$ and $S :=
\emptyset$,
while $\psi$ is the tagging function.  Since $S = \emptyset$, we
calculate $A = \formrules(\P) =
\formrules(\up{\eval}(\fun(F,x,y),z) \dppijl \app{F}{\dom(x,y,z)}) =
\formrules(\fun(F,x,y)) \cup \formrules(z) = \formrules(\fun(F,x,y))$,
which consists of the rules $\eval(\fun(F,x,y),z) \arrz \app{F}{\dom(
x,y,z)},\ \dom(x,y,\nul) \arrz x$ and $\dom(\nul,\nul,z) \arrz \nul$
(since not $\fun \gsymbform^* \suc$).

\emph{Step~\ref{alg:problem}: determine a suitable partitioning and
reduction pair.}  As $\P$ consists of only one dependency pair,
we must choose $\P_1 = \P$.  We have the following ordering
constraints:
\[
\begin{array}{rclrcl}
\up{\eval}(\fun(F,x,y),z) & \gterm & \app{F}{\dom(x,y,z)}\ \ \ \  &
\dom(x,y,\nul) & \geqterm & x \\
\eval(\fun(F,x,y),z) & \geqterm & \app{F}{\dom(x,y,z)} &
\dom(\nul,\nul,z) & \geqterm & \nul \\
\end{array}
\]
The first constraint is for the dependency pair, the others for the formative
rules.  Using an argument function $\argfil$ which maps $\dom(x,y,z)$ 
to $\dom'(x,y)$ and otherwise maps $\afun(\vec{\avar})$ to itself
($\argfil$ is an \emph{argument filtering}), this leaves the
following constraints:
\[
\begin{array}{rclrcl}
\up{\eval}(\fun(F,x,y),z) & \gterm & \app{F}{\dom'(x,y)}\ \ \ \  &
\dom'(x,y) & \geqterm & x \\
\eval(\fun(F,x,y),z) & \geqterm & \app{F}{\dom'(x,y)} &
\dom'(\nul,\nul) & \geqterm & \nul \\
\end{array}
\]
This is satisfied with $(\geqtermcpo,\gtermcpo)$ from
Section~\ref{sec:argfun}, using a precedence $\fun >_F
\dom' >_F \suc,\nul$.

\emph{Steps~\ref{alg:removed1},\ref{alg:start},\ref{alg:maxcyc}:
remove all pairs in $\P_1$ and everything that's not on a cycle
anymore.}  We remove the last node; the graph is empty.  We conclude:
this AFS is terminating.
\vspace{-6pt}
\end{example}

\section{Extensions}\label{sec:improvements}

\summary{In this section we discuss two improvements to the method
in this paper: a combination with the \emph{static dependency pair}
approach, and a way to deal with AFSs with infinitely many rules, in
particular those generated from a polymorphic AFS.
}

\subsection{Merging the Static and Dynamic Approach}\label{sec:merge}

The dynamic style of dependency pairs presented in this paper is not
the only way to do dependency pairs; as discussed in
Section~\ref{sec:relatedwork} there are several strong results in
\emph{static} dependency pairs.  Most importantly, in a static
dependency pair approach we do not have to consider collapsing
dependency pairs, and hence the dependency graph is usually simpler.
Usable rules and the subterm criterion are more often applicable,
and argument filterings can be used without restrictions.

Since the dynamic approach is applicable to a larger class of AFSs,
but the static approach gives easier constraints,
it seems sensible for a termination tool to implement both.  
Static dependency pairs are defined for HRSs, but extending the
proof to AFSs takes next to no effort, and using
Lemma~\ref{lem:formativereduction} we can add formative rules to it
as well.  Below, we give a short overview without proofs; for the
full work on static dependency pairs,
see~\cite{kus:iso:sak:bla:09:1,suz:kus:bla:11}.

\begin{definition}[AFS with Base Output Types]
An AFS $(\F,\Rules)$ \emph{has base output types} if for all
$\afun : [\atype_1 \times \ldots \times \atype_n] \decpijl \btype \in
\F$, the output type $\btype$ is a base type.
\end{definition}
The AFS for $\eval$ from Example~\ref{ex:eval} has base output
types, but the AFS $\twice$ does not, because of the symbol $\twice :
[\nat \typepijl \nat] \decpijl \nat \typepijl \nat$.
It is always possible to transform an AFS into an AFS with base
output types by $\eta$-expanding the rules~\cite{kop:11:1}, without
losing non-termination; however, termination may be lost by this
transformation.

\begin{definition}[Plain Function Passing AFS] An AFS is \emph{plain
function passing} (PFP) if for all rules $\afun(l_1,\ldots,l_n) \arrz
r$: if $r$ contains a functional variable $F$, then $F$ is one of the
$l_i$.
\end{definition}
\noindent
The notion is slightly simplified from the original definition, for
the sake of easy explanation.

The set of \emph{static dependency pairs} of a PFP AFS with
base output types consists of the pairs
$\up{l} \dppijl \up{\bfun}(\vec{\aterm})$
with $l \arrz r$ a rewrite rule,
$r \suptermeq \bfun (\vec{\aterm})$, 
and $\bfun$ a defined symbol.
Note that the right-hand side of a static dependency pair 
may contain variables which do not occur in its left-hand side.
For example, the rewrite rule
$\I(\suc(n)) \arrz \twice(\abs{\avar}{\I(\avar)},n)$
has two static dependency pairs:
$\up{\I}(\suc(n)) \dppijl \up{\twice}(\abs{\avar}{\I(\avar)},n)$ and
$\up{\I}(\suc(n)) \dppijl \up{\I}(\avar)$.

A \emph{static dependency chain} is a sequence 
$\rijtje{(\rho_i,\aterm_i,\bterm_i) \mid i \in \N}$ where:
\begin{iteMize}{$\bullet$}
\item each $\rho_i$ is a static dependency pair $l \dppijl p$;
\item there is some substitution $\asub$ such that $\aterm_i =
  l\asub$ and $\bterm_i = r\asub$;
\item each $\bterm_i \arrr{\formrules(\rho_{i+1})} \aterm_{i+1}$
  (where $\formrules(\rho_{i+1}) = \Rules$ in a non-\llfe\ AFS);
\item the immediate subterms of all $\bterm_i$ are ``computable''
  (which implies termination).
\end{iteMize}

\begin{claim}\label{cla:static}
A plain function passing AFS with base output types is terminating if
it does not have a static dependency chain.
\end{claim}

Claim~\ref{cla:static} can be verified by mimicking the proof
in~\cite{kus:iso:sak:bla:09:1} for the different setting, and
additionally using Lemma~\ref{lem:formativereduction}; this proof
contains no novelties.

\smallskip \noindent
Observing that all static dependency pairs are non-collapsing,
it is worth noting that a static dependency chain is also a dynamic
minimal
dependency chain as defined in Definition~\ref{def:dependencychain};
just for a different set of dependency pairs.  In the results of
this paper, we did not use the property that the right-hand side of
dependency pairs contains no new variables; thus, we can also use the
algorithm from Section~\ref{sec:algorithm} to prove absence of static
dependency chains.  With static dependency pairs, we can always use
case~\ref{alg:weaknonweak}a, so we have unrestricted argument
filterings, the subterm criterion and usable rules, as we had in the
HRS setting by~\cite{suz:kus:bla:11}.

Now an automatic tool could, for example, try dynamic dependency
pairs, and if that fails but the AFS is PFP, $\eta$-expand the rules
and use the same module to attempt a static approach.
But in some cases we can do better.  
For instance, the dynamic dependency pairs for $\map$ from
Example~\ref{ex:goodmap} are:
\[
\begin{array}{llcl}
(A) & \up{\map}(F,\cons(h,t)) & \dppijl & \app{F}{h} \\
(B) & \up{\map}(F,\cons(h,t)) & \dppijl & \up{\map}(F,t) \\
\end{array}
\]
In the static approach, we only have pair (B).  Thus, if there is a
dependency chain on the static dependency pairs, there is
also one on the dynamic ones.  In cases like this, there is no point
using the dynamic approach instead of the static one.

\begin{definition}
We say an AFS $(\F,\Rules)$ is \emph{strongly plain function passing}
(SPFP) if it is plain function passing, has base output types, and
the right-hand sides of $\Rules$ do not have subterms of the form
$\abs{\avar}{C[\afun(\aterm_1,\ldots,\aterm_n)]}$ where $\afun$ is a
defined symbol and $\avar \in \FV(\afun(\vec{\aterm}))$.
\end{definition}

\begin{theorem}
A SPFP AFS admits a minimal dynamic
dependency chain if it admits a minimal static dependency chain.
A SPFP and left-linear AFS admits a minimal static dependency
chain if and only if it is non-terminating.
\end{theorem}

\begin{proof}
Given a SPFP AFS, all its static dependency pairs are also
dynamic dependency pairs, so any minimal static dependency chain
\emph{is} a minimal dynamic dependency chain.

If the SPFP AFS is moreover left-linear, then termination
implies that there is no minimal dynamic dependency chain (by
Theorem~\ref{thm:maintheoremll}), let alone a static one.
\end{proof}

This theorem gives a strengthening of the dynamic approach, by
allowing us to remove collapsing dependency pairs in an AFS if it is
strongly plain function passing.  It also provides a completeness
result for the static approach, for the class of SPFP systems.

\smallskip \noindent
In summary, we can combine the static and dynamic approaches as
follows: \\
{\tt
\indent if $\Rules$ is strongly plain function passing:
  return DPframework(STATIC) \\
\indent if DPframework(DYNAMIC) = TERMINATING:
  return TERMINATING \\
\indent if $\F$ does not have base output types:
  $\eta$-expand $(\F,\Rules)$ \\
\indent return DPframework(STATIC)
}

\smallskip \noindent
Here, \texttt{DPframework} is the algorithm of
Section~\ref{sec:algorithm}, and the argument (\texttt{STATIC} or
\texttt{DYNAMIC}) determines how the dependency pairs are calculated.
Alternatively, the two frameworks may be run in parallel, or we
might consider optimisations to avoid double work.  At present,
however, we have not done this, because testing on the current
termination problem database shows no improvement in
strength obtained from using the static approach after the dynamic
approach has failed (although using the static approach for SPFP
AFSs \emph{does} significantly improve the
power).  Experimental results and statistics
are given in Section~\ref{sec:experiments}.

\subsection{Polymorphism}\label{subsec:polymorphism}

In this paper we consider monomorphic AFSs.
In other definitions of the AFS formalism, 
a kind of polymorphism \`a la ML is admitted.
In a polymorphic AFS, we can have for instance 
a function symbol 
$\mathtt{if} : [\mathtt{bool} \times \alpha \times \alpha] \decpijl \alpha$ 
with $\alpha$ a type variable, 
and rewrite rules
$\mathtt{if}(\mathtt{true},\avar,\bvar) \arrz \avar$ and
$\mathtt{if}(\mathtt{false},\avar,\bvar) \arrz \bvar$.

Here $\alpha$ can be instantiated by base types, but also by
functional types such as $\mathtt{nat} \typepijl \mathtt{nat}$.
A polymorphic AFS can be transformed into a monomorphic AFS
with infinitely many rewrite rules by considering all possible
instantiations for the type variables in the rewrite rules.
In the example above, we get for instance 
$\mathtt{if}_{[\mathtt{bool} \times \nat \times
\nat] \decpijl \nat}(\mathtt{true},\avar,\bvar) \arrz \avar$ and
$\mathtt{if}_{[\mathtt{bool} \times (\nat \typepijl \nat) \times
(\nat \typepijl \nat)] \decpijl \nat \typepijl \nat}(\mathtt{false},
\avar,\bvar) \arrz \bvar$.
Such a transformation is discussed in~\cite{kop:11:1}.

In this paper we considered monomorphic AFSs
with finitely many rules, so we can neither deal with 
polymorphic AFSs directly nor with their transformations
to monomorphic AFSs. 
However, a further inspection of the proofs shows that
finiteness of the set of rewrite rules is primarily used in dealing
with the dependency graph.
So suppose we did not pose this
restriction, and went on
as before.  This
leads to a dependency graph with infinitely many nodes.  Now consider
the following definition of a dependency graph approximation:

\begin{definition}[Infinite Dependency Graph Approximation]
An (infinite) dependency graph $A$ is approximated by the finite graph
$G$ if there is a mapping $f$ from the nodes in $A$ to the nodes in
$G$, such that: if $A$ contains an edge from node $a$ to node $b$,
then $G$ contains an edge from node $f(a)$ to node $f(b)$.
\end{definition}

The original notion of a dependency graph approximation is an
instance of this definition, where $f$ is the identity.
With this definition, all proofs go through essentially
unmodified, only when considering a set $\P$ of nodes,
we are interested not in the elements of $\P$ but in the
corresponding dependency pairs.  For instance in
Theorem~\ref{thm:maintheorem}, we must see that $l \gterm p$ for all
$l \dppijl p \in \DP$ with $f(l \dppijl p) \in \P_1$ and
$l \geqterm p$ for all $l \dppijl p \in \DP$ with $f(l \dppijl
p) \in \P_2$.

\begin{example}\label{ex:polymorphicgraph}
Consider the polymorphic AFS with $\mathtt{if}$ as defined before,
and $\mathtt{append} : [\mathtt{list}(\alpha) \times \mathtt{list}(
\alpha)] \decpijl \mathtt{list}(\alpha)$.
\[
\begin{array}{rcl}
\mathtt{if}(\mathtt{true},\avar,\bvar) & \arrz & \avar \\
\mathtt{if}(\mathtt{false},\avar,\bvar) & \arrz & \bvar \\
\mathtt{append}(\nil,\avar) & \arrz & \avar \\
\mathtt{append}(\cons(h,t),\avar) & \arrz & \cons(h,\mathtt{append}(
  t,\avar)) \\
\end{array}
\]
The corresponding monomorphic AFS has 
infinitely many rules, and infinitely many
dependency pairs.  
All the rules generated from the $\mathtt{append}$
rules (such as append for $\mathtt{list}(\nat)$ or for
$\mathtt{list}(\nat \typepijl \mathtt{bool})$) have only one
dependency pair
$\up{\mathtt{append}_\atype}(\cons(h,t),\avar) \dppijl
\up{\mathtt{append}_\atype}(t,\avar)$.
The $\mathtt{if}$ rules of basic types (such as when $\alpha$ is
instantiated with $\nat$) have no dependency pairs, but if we for
instance consider $\atype := [\mathtt{bool} \times \nat \typepijl
\nat \times \nat \typepijl \nat] \decpijl \nat \typepijl \nat$ then
the rule $\mathtt{if}_\atype(\mathtt{true},\avar,\bvar) \arrz \avar$
does give a dependency pair $\app{\mathtt{if}_\atype(\mathtt{true},
\avar,\bvar)}{\cvar} \dppijl \app{\avar}{\cvar}$.

The dependency graph is infinite, and part of it could be sketched
like this:
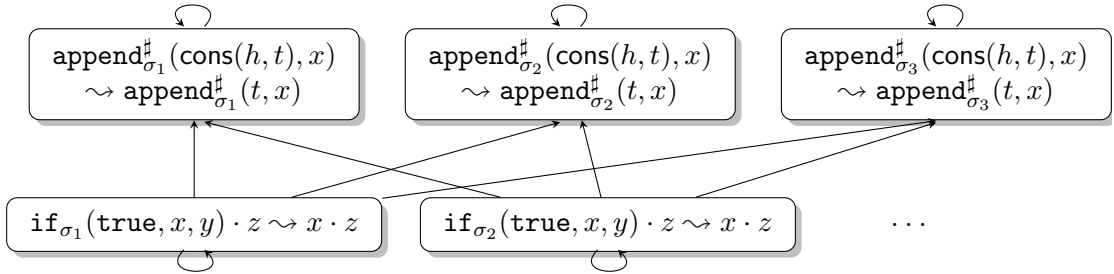
\begin{center}
\begin{tikzpicture}[->]
\begin{scope}[>=stealth]

\tikzstyle{veld} = [draw, fill=white,  drop shadow, 
  minimum height=0em, minimum width=0em, rounded corners]

\node [veld] (append1) {
    $\begin{array}{c}\up{\mathtt{append}_{\atype_1}}(\cons(h,t),\avar) \\
    \dppijl \up{\mathtt{append}_{\atype_1}}(t,\avar) \end{array}$
  };
\node (append1loop) [above of=append1,node distance=5mm] {};
\node (append1bot) [below of=append1,node distance=6mm] {};

\node [veld] (append2) [right of=append1,node distance=5cm] {
    $\begin{array}{c}\up{\mathtt{append}_{\atype_2}}(\cons(h,t),\avar) \\
    \dppijl \up{\mathtt{append}_{\atype_2}}(t,\avar) \end{array}$
  };
\node (append2loop) [above of=append2,node distance=5mm] {};
\node (append2bot) [below of=append2,node distance=6mm] {};

\node [veld] (append3) [right of=append2,node distance=5cm] {
    $\begin{array}{c}\up{\mathtt{append}_{\atype_3}}(\cons(h,t),\avar) \\
    \dppijl \up{\mathtt{append}_{\atype_3}}(t,\avar) \end{array}$
  };
\node (append3loop) [above of=append3,node distance=5mm] {};
\node (append3bot) [below of=append3,node distance=6mm] {};

\node [veld] (if1) [below of=append1,node distance=2cm] {
    $\begin{array}{c}\app{\mathtt{if}_{\atype_1}(\mathtt{true},
    \avar,\bvar)}{\cvar} \dppijl \app{\avar}{\cvar}\end{array}$
  };
\node (if1loop) [below of=if1,node distance=2mm] {};

\node [veld] (if2) [right of=if1,node distance=5.5cm] {
    $\begin{array}{c}\app{\mathtt{if}_{\atype_2}(\mathtt{true},
    \avar,\bvar)}{\cvar} \dppijl \app{\avar}{\cvar}\end{array}$
  };
\node (if2loop) [below of=if2,node distance=2mm] {};

\node (dots) [right of=if2,node distance=4cm] { \large{$\ldots$} };

\draw (append1loop) to[out=45,in=135,looseness=5] (append1loop);
\draw (append2loop) to[out=45,in=135,looseness=5] (append2loop);
\draw (append3loop) to[out=45,in=135,looseness=5] (append3loop);
\draw (if1loop) to[out=225,in=315,looseness=5] (if1loop);
\draw (if2loop) to[out=225,in=315,looseness=5] (if2loop);
\draw (if1) -- (append1);
\draw (if1) -- (append2bot);
\draw (if1) -- (append3bot);
\draw (if2) -- (append1bot);
\draw (if2) -- (append2);
\draw (if2) -- (append3bot);
\end{scope}
\end{tikzpicture}
\end{center}

For a dependency graph approximation, it makes sense to combine all
the ``similar'' dependency pairs in one node.  We choose the
following finite graph approximation:
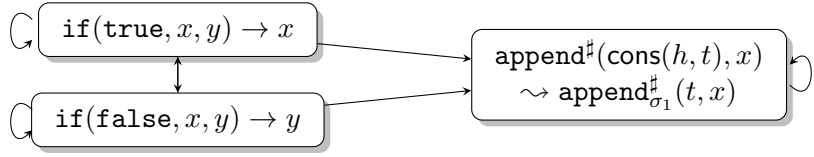
\begin{center}
\begin{tikzpicture}[->]
\begin{scope}[>=stealth]

\tikzstyle{veld} = [draw, fill=white,  drop shadow, 
  minimum height=0em, minimum width=0em, rounded corners]

\node [veld] (append) {
    $\begin{array}{c}\up{\mathtt{append}}(\cons(h,t),\avar) \\
    \dppijl \up{\mathtt{append}_{\atype_1}}(t,\avar) \end{array}$
  };
\node (appendloop) [right of=append,node distance=2cm] {};
\node (central) [left of=append,node distance=6cm] {};

\node [veld] (if1) [above of=central,node distance=0.6cm] {
    $\begin{array}{c}\mathtt{if}(\mathtt{true},\avar,\bvar)
    \arrz \avar\end{array}$
  };
\node (if1loop) [left of=if1,node distance=1.8cm] {};

\node [veld] (if2) [below of=central,node distance=0.6cm] {
    $\begin{array}{c}\mathtt{if}(\mathtt{false},\avar,\bvar)
    \arrz \bvar\end{array}$
  };
\node (if2loop) [left of=if2,node distance=1.8cm] {};

\draw (appendloop) to[out=-45,in=45,looseness=5] (appendloop);
\draw (if1loop) to[out=135,in=225,looseness=5] (if1loop);
\draw (if2loop) to[out=135,in=225,looseness=5] (if2loop);
\draw (if1) -- (append);
\draw (if2) -- (append);
\draw (if1) -- (if2);
\draw (if2) -- (if1);
\end{scope}
\end{tikzpicture}
\end{center}
Here, all dependency pairs $\up{\mathtt{append}_\atype}(\cons(h,t),
\avar) \dppijl \up{\mathtt{append}_\atype}(t,\avar)$ are mapped to
the $\mathtt{append}$ node, and the $\app{\mathtt{if}_\atype(\mathtt{
true},\avar,\bvar)}{\cvar_1} \cdots \cvar_n \dppijl \app{\avar}{
\cvar_1} \cdots \cvar_n$ are mapped to $\mathtt{if}(\mathtt{true},
\avar,\bvar) \arrz \avar$.  Note the different arrows: the
$\mathtt{append}$ node uses a $\dppijl$ arrow, to indicate that it
catches dependency pairs of the \emph{first kind} (see
Definition~\ref{def:deppair}) which have the given form, whereas the
$\mathtt{if}$ nodes use a $\arrz$ arrow and catch
dependency pairs of the \emph{second kind}.
\end{example}

Example~\ref{ex:polymorphicgraph} gives a good suggestion for how to
split the dependency pairs into a finite number of groups.  For a
polymorphic rule, we generate the set of dependency pairs of the
first kind as before, $\DP_1(\Rules) = \{ \up{l} \dppijl \up{p} \mid
l \arrz r \in \Rules \mid p \in \candidatesof{r} \}$; for all
instances of the rule, the dependency pairs of the first kind will be
type instantiations of a pair in $\DP_1(\Rules)$.  For the nodes of
the graph approximation we choose $\DP_1(\Rules)$, and in addition
those rules $l \arrz r$ whose type is not a base type and where
$\head(r)$ is either a variable or a term $\afun(\vec{r})$ with
$\afun \in \Defineds$.  Now each dependency pair in the instantiated
system is either mapped to the corresponding pair in $\DP_1(\Rules)$,
or, for a pair of the second kind, to the rule that generated
it.

With this choice for the graph approximation we can avoid any
infinite reasoning, \emph{if} we can find a reduction pair that is
equipped to prove statements of the form ``$\overline{l} \gterm
\overline{p}$ for all type instantiations'' given $l$ and $p$.
Finding such reduction pairs is the object of separate study,
however.
As for formative rules, it is not hard to devise an algorithm which,
given a symbol $\aterm$, finds a set of polymorphic rules
$\mathit{PFR}$ such that the rewrite relation generated by
$\mathit{PFR}$ is included in
$\arr{\formrules(\aterm)}$ (the difficulty is in determining how good
the result of such an algorithm is, as it may well include more rules
than strictly necessary).
Usable rules require only that every term has only finitely many
reducts (item I in the proof of Theorem~\ref{thm:userules}); this is
the case for systems generated from a finite set of polymorphic
rules as well.

\section{Experimental Results}\label{sec:experiments}

\summary{In this section, we present the
performance of the higher-order termination tool \wanda\
in different settings:
using only dynamic dependency pairs,
using only static dependency pairs, 
using both dynamic and static dependency pairs,
and using only rule removal.}

The dependency pair framework described in this paper,
and in particular the algorithm of Section~\ref{sec:algorithm},
forms the core of the
higher-order termination tool \wanda~\cite{kop:10:1}.
If the input system is strongly plain function passing
as described in Section~\ref{sec:merge},
then collapsing dependency pairs are dropped.
If it is not strongly plain function passing, 
then a static approach is
attempted only when the dynamic approach fails.
The tool \wanda\ uses a simple form of
argument functions immediately (if there is any constraint
$\afun(\avar_1,\ldots,\avar_n) \geqterm r$ or $\afun(\avar_1,\ldots,
\avar_n) {\,{}_{_(}\!\!\geqterm_{_)}} r$ then the argument function
$\argfil(\afun(\vec{\avar})) = r$ is used), a proof-of-concept
implementation of polynomial interpretations in the natural numbers,
and 
our own
version of a recursive path ordering (based on~\cite{kop:raa:08:1})
with argument filterings.

The current termination problem database (8.0.1) contains 156
higher-order benchmarks; out of these, 138 are \llfe, and 133 are
plain function passing.  The following table sums up how the various
restrictions relate.

\begin{figure}[h]
\begin{tabular}{l||rrrr}
 & \llfe & not \llfe & not left-linear & not fully extended \\
\hline
PFP & 123 & 10 & 9 & 1 \\
Strong PFP & 88 & 9 & 8 & 1 \\
Not PFP & 15 & 8 & 4 & 6 \\
\end{tabular}
\end{figure}

\medskip \noindent
We executed \wanda\ 1.6\footnote{This is not the version which was
used in the last termination competition; there, polynomial
interpretations had not yet been implemented.} with and without
dependency pairs, and also once with the locality improvement
disabled,
on a dual Intel(R) Core(TM) i5-2520M CPU @ 2.50GHz computer with 4G
RAM and a 60 second timeout; the results are summarised in the table
below.  An evaluation page with more details is available at
\url{http://www.few.vu.nl/~femke/lmcs2012/}.

\begin{figure}[h]
\begin{tabular}{l||rrrrr}
 & full \wanda & dynamic DP & static DP & rule removal & non-\llfe \\
\hline
YES & 122 & 119 & 114 & 75 & 100 \\
NO & 9 & 9 & 9 & 9 & 9 \\
MAYBE & 23 & 26 & 31 & 69 & 44 \\
TIMEOUT & 2 & 2 & 2 & 3 & 3 \\
average runtime & 1.33 & 1.97 & 1.40 & 0.89 & 1.82 \\
\end{tabular}
\end{figure}

\medskip \noindent
In \emph{rule removal} a reduction pair $(\geqterm,\gterm)$ is used
where $\gterm$ is monotonic, and for all rules either $l \gterm r$
or $l \geqterm r$; strictly oriented rules can be removed.
The timeouts are mainly due to \wanda's strategy of trying
more advanced polynomial interpretations if simple forms fail.
It is noticeable that the static style is a fair bit faster than the
dynamic one; this is to be expected, since it gives both fewer and
easier constraints, and the subterm criterion in particular often
supplies a quick solution.
The results for \llfe\ AFSs of Section~\ref{sec:weakdp} give a
significant improvement over the basic dynamic dependency pair
approach, but even the basic result is a real improvement over rule
removal.

This table demonstrates that using a dependency pair approach, either
static or dynamic, significantly improves the strength of an automatic
termination tool for higher-order rewriting.  Moreover, the
combination of static and dynamic dependency pairs is stronger than
either technique on its own; the two styles are truly incomparable.

\section{Conclusion}\label{sec:conclusion}

We have presented a method to prove termination of AFSs
using dynamic dependency pairs.
A first improvement of the method is obtained, 
for the subclass of local AFSs,
by using formative rules, which are
a variation of the usable rules from the first-order case.
Further improvements are obtained by the use of argument functions,
which are an extension of the argument filterings from the
first-order case, and the subterm criterion and usable rules for
non-collapsing dependency pairs.
The ordering constraints generated by the dependency pair approach
can be solved not only using reduction pairs such as
higher-order recursive path orderings,
but also using weakly monotonic algebras.

The dependency pair approach as presented here is 
implemented in the higher-order termination tool
\wanda~\cite{kop:10:1} by the first author.
It can be used on its own,
or together with a static approach in an automatic termination tool.
Experiments with the termination problem database have shown
that the styles are incomparable, but have similar strength.

With this work we aim to contribute to the understanding
of termination of higher-order rewriting, and
more in particular, to the understanding of dependency pairs, and
to the developments of 
tools to automatically prove termination of higher-order rewriting.

For future work, it would be interesting to generalise the dependency
pair \emph{framework}~\cite{gie:thi:sch:05:2} to the higher-order
setting.  This poses challenges like the question whether
tags and formative rules should be flags, or can be phrased as
dependency pair processors.  In a different direction, we might study
(and implement) the new method of formative rules in the first-order
setting.

\section{Acknowledgements}

We gratefully acknowledge remarks and suggestions from 
Vincent van Oostrom and Jan Willem Klop.
We thank the anonymous referees of earlier versions,
and the anonymous referees of the present version 
for their constructive and helpful remarks.

\bibliographystyle{plain}
\bibliography{references}

\end{document}